\documentclass[11pt]{article}

\usepackage[utf8]{inputenc}
\usepackage[T1]{fontenc}
\usepackage[letterpaper,margin=1in]{geometry}
\usepackage{setspace}
\usepackage{tgtermes}
\usepackage{newtxtext}
\usepackage{newtxmath}
\usepackage{bm}

\usepackage{amsmath,amssymb,amsthm,mathtools}
\usepackage{bbm}                
\DeclareMathOperator*{\argmin}{arg\,min}

\providecommand{\MANUSCRIPTNO}[1]{}

\usepackage{graphicx}
\usepackage{booktabs}
\usepackage{subfig}
\usepackage[most]{tcolorbox}
\usepackage{tikz}
\usepackage{algorithm}
\usepackage{algpseudocode}
\usepackage[shortlabels]{enumitem}
\usepackage[dvipsnames,table]{xcolor}

\usepackage{endnotes}

\usepackage[english]{babel}
\usepackage[autostyle, english=american]{csquotes}
\MakeOuterQuote{"}

\usepackage{natbib}
\bibpunct[, ]{(}{)}{,}{a}{}{,}%
\usepackage[hidelinks]{hyperref}
\usepackage[nameinlink,noabbrev]{cleveref}

\theoremstyle{plain}
\newtheorem{theorem}{Theorem}
\newtheorem{lemma}[theorem]{Lemma}
\newtheorem{proposition}[theorem]{Proposition}

\theoremstyle{definition}
\newtheorem{definition}[theorem]{Definition}
\newtheorem{assumption}[theorem]{Assumption}
\newtheorem{example}[theorem]{Example}

\theoremstyle{remark}
\newtheorem{remark}[theorem]{Remark}

\crefname{theorem}{Theorem}{Theorems}
\Crefname{theorem}{Theorem}{Theorems}

\crefname{assumption}{Assumption}{Assumptions}
\Crefname{assumption}{Assumption}{Assumptions}

\crefname{definition}{Definition}{Definitions}
\Crefname{definition}{Definition}{Definitions}

\crefname{proposition}{Proposition}{Propositions}
\Crefname{proposition}{Proposition}{Propositions}

\crefname{lemma}{Lemma}{Lemmas}
\Crefname{lemma}{Lemma}{Lemmas}

\crefname{corollary}{Corollary}{Corollaries}
\Crefname{corollary}{Corollary}{Corollaries}

\crefname{example}{Example}{Examples}
\Crefname{example}{Example}{Examples}

\crefname{algorithm}{Algorithm}{Algorithms}
\Crefname{algorithm}{Algorithm}{Algorithms}

\crefname{section}{Section}{Sections}
\Crefname{section}{Section}{Sections}

\crefname{appendix}{Appendix}{Appendices}
\Crefname{appendix}{Appendix}{Appendices}

\usepackage{color}
\usepackage{color-edits}
\addauthor{Will}{red}
\addauthor{Xintong}{blue}

\newcommand{\bI}{\mathbbm{1}}
\newcommand{\cI}{\mathcal{I}}
\newcommand{\cN}{\mathcal{N}}

\newcommand{\cS}{\mathcal{S}}
\newcommand{\hphi}{\hat{\phi}}

\newcommand{\RMSEsoft}{\mathrm{RMSE}^\mathrm{soft}}
\newcommand{\minBF}{\mathsf{minBF}}
\newcommand{\BF}{\mathsf{BF}}
\newcommand{\Mult}{\mathsf{Mult}}

\newcommand{\Npar}{N^{\mathrm{par}}}
\newcommand{\Leaf}{\mathrm{Leaf}}

\usepackage{xspace}
\newcommand{\apple}{\raisebox{-0.185em}{\includegraphics[height=1.1em]{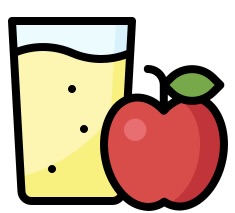}}\xspace}
\newcommand{\orange}{\raisebox{-0.2em}{\includegraphics[height=1.15em]{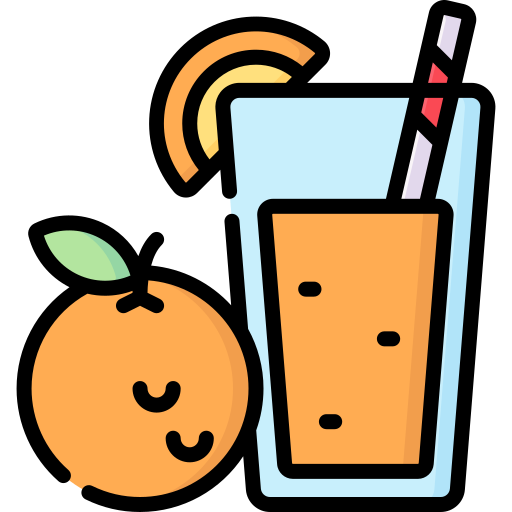}}\xspace}
\newcommand{\milk}{\raisebox{-0.2em}{\includegraphics[height=1.1em]{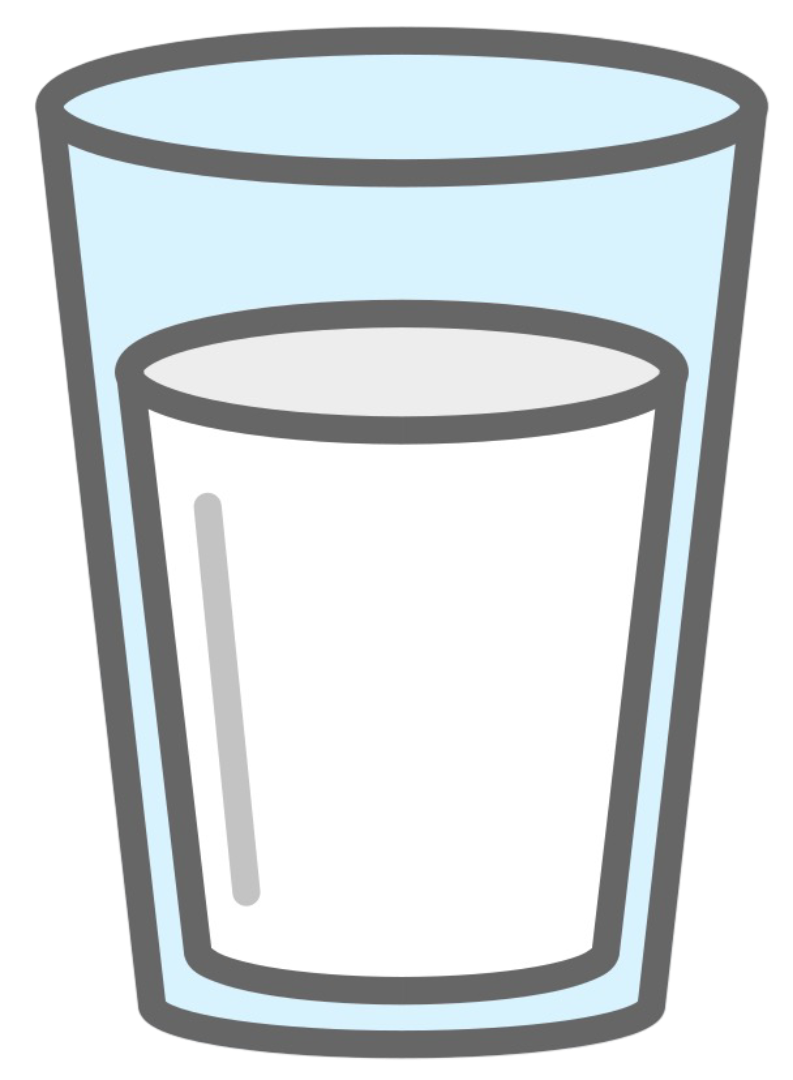}}\xspace}
\newcommand{\boba}{\raisebox{-0.2em}{\includegraphics[height=1.15em]{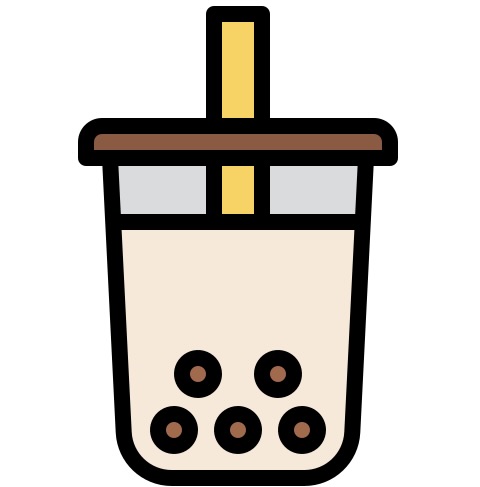}}\xspace}
\newcommand{\tea}{\raisebox{-0.2em}{\includegraphics[height=1.15em]{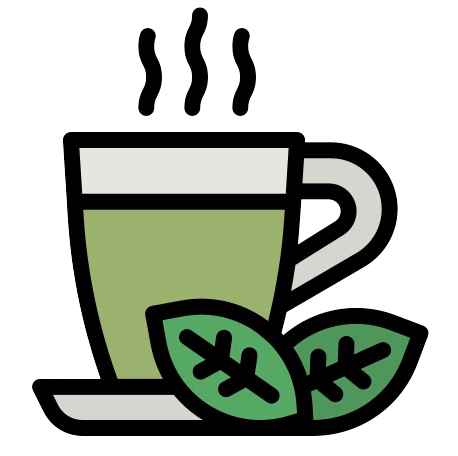}}\xspace}
\newcommand{\wine}{\raisebox{-0.2em}{\includegraphics[height=1.1em]{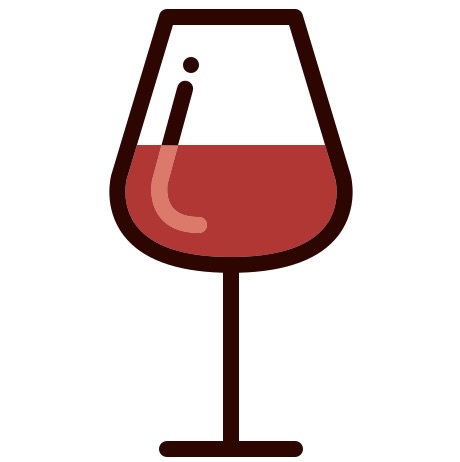}}\xspace}
\newcommand{\coffee}{\raisebox{-0.2em}{\includegraphics[height=1.15em]{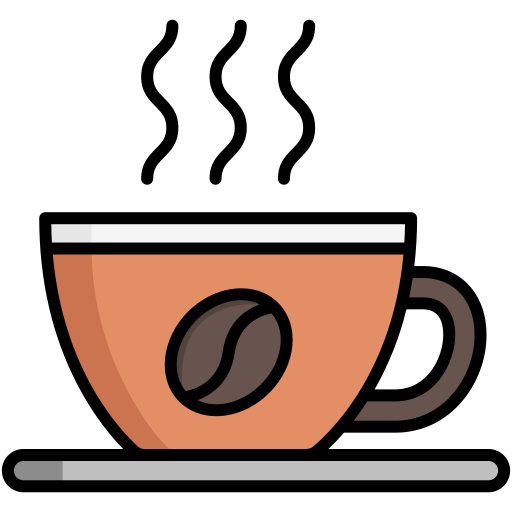}}\xspace}
\newcommand{\beer}{\raisebox{-0.2em}{\includegraphics[height=1.05em]{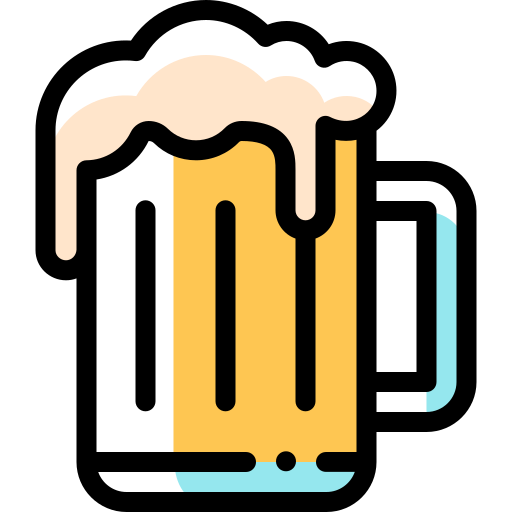}}\xspace}
\newcommand{\emojiCheck}{\raisebox{-0.2em}{\includegraphics[height=1em]{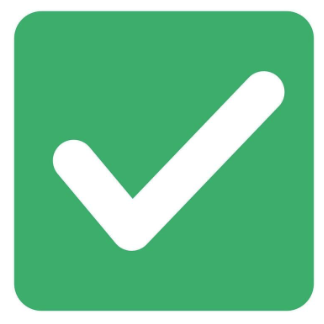}}\xspace}
\newcommand{\cross}{\raisebox{-0.2em}{\includegraphics[height=1.0em]{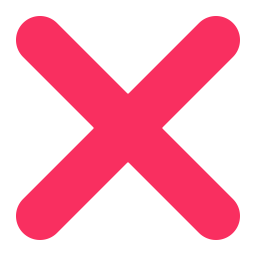}}\xspace}

\begin{document}

\title{Experimental Assortments for Choice Estimation and Nest Identification%
  \thanks{A preliminary version of this paper appeared in the proceedings of EC 2026.}}

\author{%
  Xintong Yu\thanks{Columbia University, \texttt{XYu29@gsb.columbia.edu}}
  \and
  Will Ma\thanks{Columbia University, \texttt{wm2428@gsb.columbia.edu}}
  \and
  Michael Zhao\thanks{Dream Sports, \texttt{michael.zhao@dreamsports.group}}%
}

\date{}                

\maketitle

\begin{abstract}
What assortments (subsets of items) should be offered to collect data for estimating a choice model over $n$ total items? We propose a structured, non-adaptive experiment design requiring only $O(\log n)$ distinct assortments, each offered repeatedly. The design consistently outperforms randomized and other heuristic designs across an extensive numerical benchmark that estimates multiple different choice models under a variety of (possibly mis-specified) ground truths.

We then focus on Nested Logit choice models, which cluster items into \emph{nests} of close substitutes.
Whereas existing Nested Logit estimation procedures assume the nests to be known and fixed, we present a logical deduction algorithm that identifies the nest partition from collected data.
Combined with our experiment design, the algorithm recovers the correct partition from the super-exponentially many possible nest structures using only $O(\log n)$ non-adaptive experimental assortments, a rate that we show to be information-theoretically optimal.

Our experiment design was deployed to collect data from over 70 million users at Dream11, an Indian fantasy sports platform that offers different types of contests, with rich substitution patterns between them. We identify nests based on the collected data, which lead to better out-of-sample choice prediction than ex-ante clustering from contest features. Our identified nests are ex-post justifiable to Dream11 management.
\end{abstract}

\medskip
\maketitle


\section{Introduction}

Understanding agent choice is important in many applications:
a retailer wants to know how customers choose among different brands of a product;
a car dealer wants to know how buyers select from its available models;
a policymaker wants to know how citizens substitute among transportation modes.
The goal is to estimate a choice model, which specifies, for any menu or "assortment" of options,
the expected market share that each option would receive.

Choice models capture how offering a smaller assortment concentrates more market share on each remaining option.
In the simplest Multi-Nomial Logit (MNL) choice model, these market shares are assumed to all increase by the same \%. For example,
this can be illustrated by the sales on Day 2 in \Cref{table:introExample}: when milk (\milk) and boba (\boba) were not offered, the sales of apple juice (\apple) and orange juice (\orange) both increased by 20\% relative to Day 1, where all drinks were offered.
This suggests that on Day 2, some of the customers whose favorite drink would have been \milk or \boba chose \apple or \orange instead, following a 2 to 1 ratio that is consistent with the sales on Day 1.

\begin{table}
\centering
\begin{minipage}[c]{0.62\textwidth}
\centering
\begin{tabular}{c|c|c|c|c}
& \apple Sales & \orange Sales & \milk Sales & \boba Sales \\
\hline
Day 1 & 200 & 100 & 100 & 100 \\
Day 2 & 240 & 120 & not offered & not offered \\
Day 3 & not offered & 200 & 120 & 120 \\
\end{tabular}
\end{minipage}%
\hfill
\begin{minipage}[c]{0.34\textwidth}
\caption{Toy example showing sales of drinks under different assortments offered}
\label{table:introExample}
\end{minipage}
\end{table}

However, the MNL assumption is often violated, e.g.\ on Day 3 in \Cref{table:introExample}: when \apple was not offered, the sales of \orange doubled relative to Day 1 sales compared to the 20\% increases of \milk and \boba.
This calls for richer choice models such as Nested Logit, which partitions the items into "nests" of close substitutes.
Under this model, when a customer's favorite drink (\apple) is not offered, they are more likely to switch to a drink in the same nest (in this case, a different juice \orange) than to drinks in other nests (\milk, \boba).  Note that fewer total sales were lost on Day 3, because at least one item from each nest was offered, justifying the importance of learning the grouping of items into nests.

\paragraph{Experiment design problem.}
Some assortment variation is necessary for learning richer choice models, and this variation may need to be deliberately designed.  In the example in \Cref{table:introExample}, it is in fact impossible to discern whether \milk and \boba are close substitutes, because on each day they were either both offered or both unavailable.
This motivates our primary research question.
\begin{center}
\begin{tcolorbox}[colback=gray!20, colframe=white, width=0.8\textwidth, boxrule=0pt, arc=2mm, auto outer arc]
\itshape
How to deliberately design the assortments offered so that complex substitution patterns can be detected and rich choice models can be estimated?
\end{tcolorbox}
\end{center}
This question is largely ignored in papers that estimate empirical choice models, because they use observational data in which the assortments offered were beyond the researchers' control.
Meanwhile, papers that estimate choice models from synthetic data generally draw observations from randomized assortments, which is not the most efficient approach to data collection.

In this paper, we introduce a combinatorial design that deliberately arranges the items into a small number of experimental assortments, which are much more informative than randomized assortments.
We replicate the numerical framework of \citet{berbeglia2022comparative}, and show that by only replacing their randomized experimental assortments while keeping the ground truths and estimation methods fixed, estimation error is robustly decreased.
That is, even though our experiment design is motivated by Nested Logit, it significantly improves choice model estimation
regardless of whether the true and/or estimated models are Nested Logit.

For 21 days in Spring 2025, our experiment design was deployed across 70 million users at \textit{Dream11}, an Indian fantasy sports platform.  Dream11 offers different types of "contests" for users to join, and wants to understand how users choose between them.

\paragraph{Nest identification problem.}
After our experiments were deployed, the managers at Dream11 wanted to understand how the contests could be classified into nests of close substitutes.
The contests are diverse in many ways,
so a priori, it is difficult to subjectively classify them into nests (the way that one could classify \apple, \orange as "juices").
We instead want to classify items into nests based on the data collected, which motivates our second research question.
\begin{center}
\begin{tcolorbox}[colback=gray!20, colframe=white, width=0.8\textwidth, boxrule=0pt, arc=2mm, auto outer arc]
\itshape
How to automatically identify nests based on sales data, instead of relying on subjective classification?
\end{tcolorbox}
\end{center}

Historically, papers on Nested Logit choice estimation (see \Cref{sec:nestedLogitEstimation}) have assumed the nests are fixed, focusing instead on estimating the other model parameters for a combination of tractability and interpretability reasons \citep[\S4]{train2009discrete}.  Standard statistics packages (e.g. \texttt{nlogit} in Stata) also make this assumption.
Papers that consider nest identification are surprisingly scant, as we discuss in \Cref{sec:nestedLogitEstimation}.

Our second contribution is to propose a new algorithm for nest identification, based on the simple intuition from \Cref{table:introExample}.
To elaborate, for each item in each experimental assortment (Day 2, Day 3), we define its \textit{boost factor} to be the ratio of its sales relative to its sales in the control assortment (Day 1) where all items were offered.
The algorithm then makes the following two types of deductions.
\begin{enumerate}[I.]
\item (\textit{"Small" Boost Factor}) Because \apple had a "small" boost factor ($\frac{240}{200}=1.2$) on Day 2, the unavailable items are not close substitutes.  That is, neither \milk nor \boba are in the same nest as \apple.  We can similarly deduce that neither \milk nor \boba are in the same nest as \orange.
\item (\textit{"Large" Boost Factor}) Because \orange had a "large" boost factor ($\frac{200}{100}=2$) on Day 3, at least one unavailable item is a close substitute.  This item must be \apple, and hence \apple,\orange are in the same nest.
\end{enumerate}
Our theoretical algorithm assumes that sales boosts of different magnitudes can be perfectly distinguished, in which case we prove that correct nest identification is guaranteed (see \Cref{sec:introTheory}).
We also develop an empirical version of the algorithm that handles noise, which we compare to the nest identification algorithm of \citet{benson2016relevance} on synthetic and real data (see \Cref{sec:numericalResultsIntro}).
In general, our nest identification algorithm differs from the literature by leveraging repeated observations under a small number of assortments, as motivated by our experiment design.

\subsection{Our Experiment Design} \label{sec:designDescr}

Our experiment design prescribes a small number of pre-determined assortments, making it easy to deploy in practical experimentation settings.
Indeed, a retail chain may only be able to experiment with one assortment per brick-and-mortar store, and it may want to experiment in parallel, so earlier outcomes cannot be used to adaptively choose later assortments.
In our drinks example where the randomization is over time, each day can be assigned only one assortment, so our design allows the experimentation to be completed in fewer days.

We now explain our design, for which we recall the example from \Cref{table:introExample}.  There, it was possible to deduce that \apple,\orange are in the same nest which does not include \milk or \boba, but not possible to deduce whether \milk,\boba are in the same nest.
To rectify this, our experiment design promises the following property: for any (ordered) pair of items $i,j$, there exists an experimental assortment $S$ containing $i$ but not $j$.
By looking at the assortment $S$ containing \milk but not \boba, we get a sense of how the removal of \boba boosted \milk's sales, which helps deduce whether they are in the same nest.

To satisfy this desired property, our experiment design gives each of the $n$ items 
a unique binary encoding with $L:=\lceil\log_2 n\rceil$ digits.
For each $\ell=1,\ldots,L$, we include two experimental assortments: one that contains all items with a "1" in the $\ell$'th digit of their binary encodings, and the complement assortment that contains all items with a "0" in their $\ell$'th digits.
This experiment design satisfies the desired property, because two distinct items $i,j$ with unique binary encodings must have different digits for at least one position $\ell$, and one of the two experimental assortments $S$ for that $\ell$ will then contain $i$ but not $j$.
An example of our design for $n=8$ is given in \Cref{table:introDesign}.

\begin{table}
\centering
\begin{tabular}{c|c|c|c|c|c|c|c|c}
Offered? (\emojiCheck/\cross) &  \milk 000 &  \apple 001 &  \boba 010 & \coffee 011 & \tea 100 &  \orange 101 & \beer 110 & \wine 111 \\
\hline
Control & \emojiCheck & \emojiCheck & \emojiCheck & \emojiCheck & \emojiCheck & \emojiCheck & \emojiCheck & \emojiCheck \\
Experiment 1 & \cross & \cross & \cross & \cross & \emojiCheck & \emojiCheck & \emojiCheck & \emojiCheck \\
Experiment 2 & \emojiCheck & \emojiCheck & \emojiCheck & \emojiCheck & \cross & \cross & \cross & \cross \\
Experiment 3 & \cross & \cross & \emojiCheck & \emojiCheck & \cross & \cross & \emojiCheck & \emojiCheck \\
Experiment 4 & \emojiCheck & \emojiCheck & \cross & \cross & \emojiCheck & \emojiCheck & \cross & \cross \\
Experiment 5 & \cross & \emojiCheck & \cross & \emojiCheck & \cross & \emojiCheck & \cross & \emojiCheck \\
Experiment 6 & \emojiCheck & \cross & \emojiCheck & \cross & \emojiCheck & \cross & \emojiCheck & \cross \\
\end{tabular}
\caption{Our experiment design for $n=8$ items, requiring $2\lceil\log_2 n\rceil=6$ experimental assortments in addition to the control assortment where all items are offered.  The items with their binary encodings are depicted in the columns.
}
\label{table:introDesign}
\end{table}

\subsection{Theoretical Results for Nest Identification} \label{sec:introTheory}

We now elaborate on the inferences made by our nest identification algorithm, first in a noiseless setting where exact market shares under the control and experimental assortments are observed.
The goal is to identify a hidden partition $\cN$ of the items $[n]:=\{1,\ldots,n\}$ into nests $N$, satisfying $\bigcup_{N\in\cN}N=[n]$ and $N\cap N'=\emptyset$ for all $N\neq N'$.
We make a "general position" assumption that the boost factors of items from different nests cannot coincidentally be the same.

Under this assumption, we show that the following deductions can be made from each experimental assortment $S$, which extend the deductions I and II introduced earlier.
\begin{enumerate}[I.]
\item For each nest $N$ with $N\subseteq S$, we would observe a small boost factor for all items in $N$, which allows us to separate $N$ from being in the same nest as any items in $[n]\setminus S$.
\item For each nest $N$ with $N\nsubseteq S$, we would observe a large and distinct boost factor for all items in $S\cap N$, which allows us to identify that they are all in the same nest, different from the nest of any other items in $S$ (i.e., we can separate $S\cap N$ from $S\setminus N$).
\end{enumerate}
See \Cref{fig:illustrateDeductions} for an example of these deductions.

\begin{figure}
\centering
\begin{minipage}[c]{0.3\textwidth}
\centering
\includegraphics[width=\linewidth]{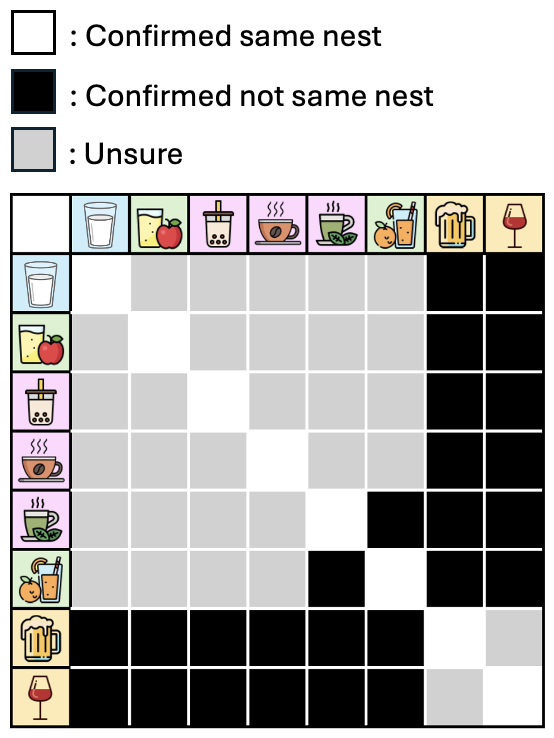}
\end{minipage}%
\hfill
\begin{minipage}[c]{0.69\textwidth}
\caption{Example deductions after Experiment~1 in \Cref{table:introDesign}.
"Small" boost factors were observed for $\beer$ and $\wine$, while "large" and distinct boost factors were observed for $\tea$ and $\orange$.  By deduction I., we know $\beer$ and $\wine$ are not in the same nest as any of the 6 other items, but we \textit{do not know} whether $\beer$ and $\wine$ are in the same nest.  By deduction II., we know $\tea$ and $\orange$ are also not in the same nest.
The table does not record the latent information that both $\tea$ and $\orange$ are in the same nest as at least one item outside $S$, even though this information is necessary for our nest identification. \\
The complete identification of nests is found in \Cref{sssec:algIllustration}.
}
\label{fig:illustrateDeductions}
\end{minipage}
\end{figure}

The hidden partition can be identified using $O(\log n)$ non-adaptive assortments via our experiment design followed by iterating these deductions over its experimental assortments.  That is, given any underlying partition of items into nests, our pipeline is guaranteed to correctly identify the ground truth and hence determine the nests for any Nested Logit choice model.
We also prove a matching lower bound: even adaptive designs require $\Omega(\log n)$ assortments.

Finally, we consider two extensions.
First, relaxing the assumption that exact market shares are observed, we prove finite-sample guarantees for our nest identification algorithms.
Second, the preceding deductions assume that an "outside option" (corresponding to not buying anything) can always be chosen, and that these choices are observed.  
We prove that our experiment design still guarantees correct nest identification without the outside option, by leveraging a more sophisticated nest identification algorithm that operates under weaker forms of deductions I and II.

\subsection{Numerical Results} \label{sec:numericalResultsIntro}

\paragraph{Comparing experiment designs.}
Following the numerical framework of \citet{berbeglia2022comparative}, we compare our deliberate experiment design with several baselines, including randomized assortments.
We first consider their mis-specified setting where the ground-truth choice model over 10 items follows a distribution over ranked lists.
Holding the number of experimental assortments fixed, we replace their random assortments of sizes 3--6 with our design while leaving the downstream estimation procedure unchanged.
Across 10 settings from their paper that vary the estimated choice model (Exponomial, MNL, Latent-class MNL, Nested Logit, or Markov Chain) and the data size, our design consistently lowers $\RMSEsoft$, with reductions up to 5.5\%, generally for smaller data sizes.
In this mis-specified setting, we also demonstrate the improvement from using our nest identification algorithm for Nested Logit.

When random designs are allowed to draw more distinct assortments under the same total observation budget, they can outperform ours for larger data sizes.
However, in well-specified settings where the estimated choice model agrees with the ground truth, our design can be the best across data sizes, even compared to \textit{individualized} random assortments.
The reduction in $\RMSEsoft$ can be as much as 16.9\%, being most significant under smaller data sizes.

\paragraph{Comparing nest identification pipelines.}
We compare our full pipeline to that of \citet{benson2016relevance}.
Under Nested Logit ground truths, our pipeline reduces $\RMSEsoft$ by as much as 46\%, mainly because our assortments have size approximately $n/2$ and reveal richer substitution information than \citet{benson2016relevance}'s size-2 or size-3 assortments.
When paired with our experiment design, their identification algorithm can perform comparably to ours under small data sizes by aggregating observations across assortments and reducing variance.
However, this aggregation introduces bias, leading to worse performance for large data.

Our nest identification algorithm is somewhat tied to our experiment design, but we also test it on the publicly available SFWork dataset \citep{KoppelmanBhat2006}, which has repeated observations but not our deliberate arrangement of assortments.
On this dataset, the out-of-sample prediction error of our nest identification algorithm is comparable to that of \citet{benson2016relevance}.

\subsection{Deployment at Dream11} \label{sec:dreamSportsIntro}

For a span of 21 days in spring 2025, our experiment design was deployed across 70 million users at Dream11.
Half of them were in the control group and saw all the contest types when deciding which contest to join.
For the other 35 million users, we experimented with $n=72$
contest types.
We created 14 experimental groups based on giving these 72 contests unique 7-digit binary encodings, ensuring that each experimental assortment had $\approx36$ contest types removed.
Each experimental assortment had to be manually checked by Dream11 management before proceeding, showing that even on a digital platform, it can be important to have a small number of experimental assortments.

After collecting the data, we estimate a Nested Logit choice model with nests identified by our algorithm, and show that it consistently outperforms the simpler MNL model on out-of-sample prediction error.
Importantly, it also outperforms a Nested Logit model with nests determined from $k$-means clustering on contest features (e.g., entry fee and contest size), which represents a standard, non-data-driven approach to nest identification.
Finally, although a Markov Chain choice model beats Nested Logit on prediction error given enough data, an important virtue of our approach is interpretability:
the learned nests reveal meaningful substitution patterns and strike a ``sweet spot'' between the over-simplified MNL and more flexible, higher-parameter alternatives.

\paragraph{Roadmap.}
Our experiment design is presented in \Cref{sec:ourExpDesign} and tested numerically in \Cref{sec:numericExpDesign}.
Our nest identification algorithm is analyzed theoretically in \Cref{sec:theory} and tested numerically in \Cref{sec:numericNestId}.
Deployment at Dream11 is documented in \Cref{sec:deployment}.
We provide a concluding discussion in \Cref{sec:conc}.

\section{Related Work} \label{sec:furtherRelated}

\subsection{Choice Estimation from Real Data} \label{sec:realDataEstimation}

Papers that estimate choice models from real data are often limited by the quality of the data, like whether product availability information is available in the first place \citep{conlon2013availability} and whether lost demand is observed \citep{vulcano2012estimating}.
Any further variation in product availability would have to come from stockouts and supply availability \citep[see][]{kok2007demand,bodea2009data}, or fluctuating covariates such as in a classical DVD dataset \citep{rusmevichientong2010dynamic,farias2013nonparametric}.
As a result, these papers are often bound by the data in terms of the choice models they estimate.

We have to our knowledge a rather unique luxury of estimating choice models from real data, when we also get to control the assortments offered to collect this data.
Our focus is therefore on how to design a small number of assortments to "maximize the variation" in product availability, instead of how to handle the lack of information \citep[see][]{vulcano2012estimating} or what is the right choice model complexity given the data at hand \citep[see][]{farias2013nonparametric}.

\subsection{Choice Estimation from Randomized Assortments} \label{sec:syntheticDataEstimation}

Although synthetic ground truths enable the trial of different experiment designs, papers that use them generally default to randomized experimentation.  There are examples of different randomization schemes:
\citet{csimcsek2018expectation,chen2025assortment} include each item with probability 1/2;
\citet{blanchet2016markov,berbeglia2022comparative} draw random assortments containing between 1/3 to 2/3 of the items; \citet{vanryzin2015market,vanryzin2017ranbased} draw random assortments with greater variation in sizes.
These papers focus on improving the choice estimation procedures, whereas we show that first-order improvements can be made by combinatorially (instead of randomly) selecting the assortments from which data is collected.

\subsection{Simple Experiment Designs for Estimation of Specific Choice Models} \label{sec:simpleDesigns}

We mention some non-random experiment designs that have been used to estimate different choice models.
For Markov Chain choice models, \citet[\S2.1]{blanchet2016markov} show how to identify all $\approx n^2$ parameters using $n+1$ assortments, consisting of the grand assortment $[n]$ and all assortments $[n]\setminus\{1\},\ldots,[n]\setminus\{n\}$ that remove one item.  We compare to this "Leave-one-out" experiment design in our simulations, which is generally not data-efficient because removing one item has a small effect that is difficult to statistically detect (see \Cref{sec:numericExpDesign}).

Meanwhile, for Latent-class MNL choice models, identification is often more algebraically involved \citep[e.g.][]{chierichetti2018learning,tang2020learning}.
The identification result of \citet{chierichetti2018learning} is based on an experiment design that consists of $O(n)$ adaptive or $O(n^2)$ non-adaptive assortments that have sizes 2 or 3, which is desirable for solving a system of equations.
For Nested Logit, \citet{benson2016relevance} use an experiment design with $O(n^2)$ adaptively-chosen assortments of sizes 2 or 3 to guarantee identification.
For MNL, identification is easy, but
\citet{shah2016estimation} use algebraic metrics to compare experiment designs, focusing on pairwise comparisons, i.e.\ assortments of size 2.
Finally, \citet{chen2025assortment} consider an experiment design for end-to-end assortment optimization in which the firm offers single-item assortments that are weighed non-uniformly due to asymmetric prices.
All in all, opposite to the problem of "Leave-one-out", the designs from these papers are data-inefficient due to being too small and not showing enough items (see \Cref{ssec:numericNestIdWO}).

In sum, our experiment design has higher general data efficiency than all of these designs, thanks to using assortments of size around $n/2$, which in some sense\footnote{This intuition is also consistent with the randomized designs discussed in \Cref{sec:syntheticDataEstimation}.} maximizes the information per observation.  At the same time, it uses fewer assortments than all of them---a \textit{sublinear} amount $O(\log n)$, based on binary encodings.
To that end, our theoretical results that guarantee correct identification for Nested Logit (\Cref{sec:theory}) are also quite surprising, relative to results above that require at least linearly many assortments for correct identification.

\subsection{Learning while Earning in Assortment Optimization}

We also mention a stream of work on experimental assortments that differs from ours by: (i) requiring adaptive experimentation; and (ii) caring about the revenue earned during experimentation.
These papers solve a bandit problem over a long horizon, which is very different in nature than our focus of successful choice estimation after a short horizon (during which not many distinct assortments can be tried).
Much work in this stream focuses on the MNL choice model \citep{agrawal2019mnl,chen2021optimal}, possibly with covariates \citep{cheung2017assortment,oh2021mnlcontextual}.
We defer to the recent paper \citet{li2025online} for further references.

\subsection{Combinatorial Design of Experiments}

Finally,
combinatorial arrangements that look similar to ours do frequently appear in the design of experiments, but our \textit{assortment} setting is quite different from the typical setting.  To elaborate, in the typical setting, there is a set of controls (e.g., temperature, time, pressure) that affect the outcome of interest, and the goal of the experiment design is to "cover" different interactions between these controls when it is not possible to test all combinations.  Common designs in this literature involve orthogonal arrays and Latin squares, dating back to \citet{fisher1971design,taguchi1986introduction}.
By contrast, in our experiment design, the controls are individual items, the outcome is multi-variate (i.e., the item chosen), and the high-level goal is to "separate" each pair of items.  Our deliberate arrangement of assortments ends up being a simple and natural combinatorial design in the context of \citet{colbourn2010crc}, but to our knowledge, it has not been previously used in experiment design.

\subsection{Nested Logit Choice Estimation and Nest Identification} \label{sec:nestedLogitEstimation}

Nested Logit has a long history dating back to \citet{williams1977formation,mcfadden1978}, and is a commonly-used estimation model in both econometrics \citep{train2009discrete} and assortment optimization \citep{gallego2014constrained}.
The majority of research on Nested Logit \citep[e.g.][]{ben1985discrete,brownstone1989efficient,hensher2002specification} has focused on estimating model parameters (preference weights for items and dissimilarity parameters for nests) assuming fixed nests, noting that even this is challenging because the Maximum Likelihood formulation is non-convex.
Another explanation for nests being fixed is that domain knowledge is often strong and desirable to impose (e.g., putting red-bus and blue-bus in the same nest; see \citet[\S4]{train2009discrete}).

That being said, several works since have brought to light the issue of nest identification.
\citet{benson2016relevance,kovach2022behavioral} are similar in spirit, using violations of the Independence of Irrelevant Alternatives (IIA) property between pairs of items $i,j$ to identify when $i$ and $j$ should be put into the same nest.
\citet{benson2016relevance} have a more explicit nest identification algorithm that we compare to in \Cref{ssec:numericNestIdWO}, while \citet[\S5]{kovach2022behavioral} solves a distance minimization problem that has restrictive data requirements as the number of items grows.
Relatedly, \citet{aboutaleb2020learning} find the nest partition and Nested Logit model parameters that collectively maximize likelihood, formulating the search problem using mixed-integer programming.
Their approach can naturally handle covariates, but also does not scale well in our experiments with more items.
These papers also consider multiple levels of nesting, which we discuss in \Cref{sec:dlevelLogit}.

All in all, our nest identification algorithm differs from these works by anchoring on our experiment design, assuming that a small number of experimental assortments $\cS$ have each been offered enough times such that the observed choice probabilities for each $S\in\cS$ (and resulting boost factors) are reliable.
Our algorithm does not aggregate observations across assortments, the way that e.g.\ \citet{benson2016relevance} considers all assortments containing $\{i,j\}$ when analyzing IIA for a given pair $i,j$, making our algorithm worse on "thin" datasets where each assortment has been observed only once.
Regardless, our approach is overall best when experiment design is part of the decision pipeline, as the experiment design suggested by \citet{benson2016relevance} is data-inefficient, offering assortments of a fixed small size (2 or 3) in order to stay unbiased (see \Cref{app:BensonImplementation}), which also causes it to require $\Omega(n^2)$ different assortments.

\subsection{Graph and Partition Recovery problems}

Our specific nest identification approach is also related to another stream of work: the oracle query–based graph reconstruction problem. The objective is to recover a hidden graph by querying oracles that reveal structural information about the graph. A variety of oracle outputs has been studied in the literature. Among them, the ones most closely related to our setting are edge-information oracles \citep{GREBINSKI1998147, GrebinskiKucherov2000, CHOI2010551} and, in more advanced settings, connected-components oracles \citep{black2025optimal, harviainen2025graph}. Under connected-components oracle, one submits a subset $S \subseteq V$ and receives either the number of connected components in $G[S]$ \citep{black2025optimal} or the explicit decomposition into components \citep{harviainen2025graph}.

Our oracle is tailored to Nested Logit and choice modeling, where the underlying
graph structure takes the form of cliques, as studied in \cite{alon2005}.
The oracle returns explicit components, yet it can provide even richer information due to the clique structure of the graph.
In prior work, deterministic non-adaptive bounds range from
$\Omega(n \log n)$ to $O(n \log^2 n)$ for edge-existence queries on vertex subsets \cite{alon2005},
while randomized adaptive methods achieve $O(n \log n)$ for general graph reconstruction \citep{harviainen2025graph}.
By contrast, leveraging both the oracle output and the clique structure reduces
the complexity to $O(\log n)$ in the deterministic non-adaptive setting.

\subsection{Substitute Detection via Price Variation}

A classical goal in empirical demand analysis is to recover own- and cross-price elasticities (or related substitution objects) from a demand system \citep[e.g.][]{DeatonMuellbauer1980AIDS}, with structural discrete-choice models often used to infer substitution patterns while addressing endogeneity of prices/attributes using instruments \citep{blp1995,Nevo2000CerealMerger}. These influential frameworks focus on identification from \textit{observational} price variation and \textit{endogeneity}, whereas our setting assumes the ability to \textit{randomly experiment} with product availability, where removing a product from the assortment is like setting its price to $\infty$.

A closer conceptual comparison is \citet{li2015value}, who ask how many price experiments are needed to learn a cross-elasticity matrix.  They have a similar punchline that $O(\log n)$ experiments suffice, but through very different mechanisms---in fact, randomized experiments work in their setting, but only under structural assumptions on the sparsity of the elasticity matrix, which depend on the continuous nature of pricing.  By contrast, our $O(\log n)$ result comes from our deliberate experiment design that "separates" pairs of items, and also our discrete combinatorial reasoning about nests (see \Cref{thm:mainResult,thm:idenWithoutOutside}).

\section{Our Experiment Design}
\label{sec:ourExpDesign}
Let there be $n$ items, denoted by the set $[n]$ where $[n]:=\{1,\ldots,n\}$.
Our experiment design gives each item $i\in[n]$ an encoding in base $b\ge2$ (the Introduction only considered $b=2$ and binary encodings).  We consider encodings with $L=\lceil\log_b n\rceil$ digits, which ensures that each item can have a unique encoding because $b^L\ge n$.  We let $\sigma(i)$ denote the base-$b$ encoding given to item $i$, with $\sigma_\ell(i)\in\{0,\ldots,b-1\}$ denoting the digit in position $\ell$ for all $\ell=1,\ldots,L$.  We place no requirements on the encodings other than two distinct items cannot have the same encoding.

Given these encodings, our experiment design $\mathcal{S}$ consists of the following $bL$ assortments: 
\begin{align*}
S_{\ell,-d} :=\{i\in[n]: \sigma_\ell(i)\neq d\} 
\\
\forall \ell=1,\ldots,L;d=0,\ldots,b-1.
\end{align*}
We assume that the control assortment $[n]$ is also offered, in addition to the assortments in $\mathcal{S}$.
\Cref{table:introDesign} illustrates our design for \( n = 8 \) and \( b = 2 \), with the experimental assortments appearing in order $S_{1,-0}$,$S_{1,-1}$, $S_{2,-0}$,$S_{2,-1}$,$S_{3,-0}$,$S_{3,-1}$.

A simple encoding $\sigma$ maps each item $i\in[n]$ to the number $i-1$ in base $b$.
In our numerical experiments, we instead use a balanced exposure design, which randomizes the encodings while ensuring that each item appears in nearly the same number of assortments.

\section{Theoretical Results for Nested Logit}
\label{sec:theory}

For all assortments $S\subseteq[n]$,
a \textit{choice function} $\phi(i,S)$ specifies the probability that each item $i\in S$ would be chosen when assortment $S$ is offered.
By default, we assume the existence of an outside option $i=0$ that could also be chosen, in which case the choice probabilities satisfy
$\sum_{i\in S\cup\{0\}}\phi(i,S)= 1$.  (We consider the case without outside option in \Cref{sec:outsideAbsTheo}.)

Our theoretical results assume that $\phi$ follows a \textit{Nested Logit} model\footnote{We also describe a generalized $d$-level Nested Logit model in \Cref{sec:dlevelLogit}.}.
That is, there is a partition $\mathcal{N}$ consisting of nests $N\subseteq [n]$, satisfying
\[
\bigcup_{N \in \mathcal{N}} N = [n], \qquad N \cap N' = \emptyset \ \forall\, N \neq N'.
\]
Each item $i$ is associated with a preference weight $v_i > 0$, and each nest $N \in \mathcal{N}$ is associated with a \emph{dissimilarity parameter} $\lambda_N \in [0,1]$, where a smaller $\lambda_N$ indicates greater within-nest correlation for nest $N$. $\lambda_N\in[0,1]$ is customary for ensuring a choice model consistent with random utility maximization; we defer to \citet{gallego2014constrained} for more background.
The preference weight for a nest $N$ within a subset $S\subseteq[n]$ is defined as:
\begin{align*}
v_N(S) &:=
\begin{cases}
(\sum_{i\in N\cap S}v_i)^{\lambda_N} &\lambda_N\in(0,1]
\\ v_N\bI(N\cap S\neq\emptyset) &\lambda_N=0
\end{cases}
\end{align*}
In the latter case, nest $N$ has an additional parameter $v_N$ representing its nest-level preference weight.

Let $N(i)$ denote the nest to which item $i$ belongs.
When offered an assortment $S\subseteq [n]$, the probability of $i\in S$ being chosen is:
\begin{align}
    &\phi(i, S) = P(N(i) \mid S) \cdot P(i \mid N(i), S) \quad \text{where}
    \\
    &P(N(i) \mid S) = \frac{v_{N(i)}(S)}{1+ \sum_{N \in \mathcal{N}} v_N(S)},
    \\
 & P(i \mid N(i), S) = \frac{v_i}{\sum_{j\in N(i)\cap S} v_j} . \label{eqn:nestProb}
\end{align}
$P(N(i) \mid S)$ denotes the probability of nest $N(i)$ being chosen from $S$, and $P(i \mid N(i), S)$ denotes the probability of $i$ being chosen conditional on $N(i)$ being chosen from $S$.  The "1" in the denominator corresponds to the outside option, which has a weight normalized to 1 and is always in its own nest.
The probability of the outside option being chosen is $$\phi(0,S)=\frac{1}{1+ \sum_{N \in \mathcal{N}} v_N(S)},$$
noting that this ensures $\sum_{i\in S\cup \{0\}}\phi(i,S)= 1$.

\begin{assumption}[Identifiability] \label{ass:identify}
$\lambda_N=1$ if and only if $|N|=1$.
\end{assumption}

\Cref{ass:identify} prevents identifiability issues in which different nest partitions induce the same choice probabilities $\phi(i,S)$, and can be made without loss of generality.
Indeed, any nest $N$ with $\lambda_N = 1$ and $|N| > 1$ can be replaced by $|N|$ singleton nests without affecting the
choice probabilities.
Conversely, a singleton nest $N=\{i\}$ with $\lambda_N < 1$ can be reparameterized by giving item $i$ the new weight $v_i^{\lambda_N}$ (or $v_N$ if $\lambda_N=0$) and setting $\lambda_N = 1$.

\begin{definition}[Nest Identification Problem, noiseless version]
Given a set of experimental assortments $\mathcal{S}$ under a Nested Logit model,
we observe choice probabilities $\phi(i,S)$ for all $i \in S\cup \{0\}$ and $S \in \mathcal{S}\cup\{[n]\}$.
The goal is to recover the underlying nest partition $\mathcal N$.
\end{definition}

\subsection{Boost Factors and Deductions about Nest Membership}\label{sec:boostFactor}

The key idea is to examine the \emph{boost} in an item's choice probability in an experimental assortment $S\in\mathcal{S}$, relative to the control assortment $[n]$ where all items are offered.

\begin{definition}[Boost Factors] \label{def:bf}
For all $S\in\mathcal{S}$, define the boost factor seen by item $i\in S\cup\{0\}$ as
$$
\mathsf{BF}(i,S):=\frac{\phi(i,S)}{\phi(i,[n])}.
$$
\end{definition}

Using \Cref{eqn:nestProb} for the Nested Logit model, we obtain the following:
\begin{align*}
\mathsf{BF}(i,S)
&=
\frac{\frac{v_{N(i)}(S)}{1+\sum_{N\in\mathcal{N}}v_N(S)}\cdot\frac{v_i}{\sum_{j\in N(i)\cap S} v_j}}
{\frac{v_{N(i)}([n])}{1+\sum_{N\in\mathcal{N}}v_N([n])}\cdot\frac{v_i}{\sum_{j\in N(i)} v_j}} 
\\
&=
\left(\frac{\sum_{j\in N(i)} v_j}{\sum_{j\in N(i)\cap S} v_j}\right)^{1-\lambda_{N(i)}}
\cdot
\frac{1+\sum_{N\in\mathcal{N}}v_N([n])}{1+\sum_{N\in\mathcal{N}}v_N(S)}
\\[0.5em]
\mathsf{BF}(0,S) 
&=
\frac{\phi(0,S)}{\phi(0,[n])}
=
\frac{1+\sum_{N\in\mathcal{N}}v_N([n])}{1+\sum_{N\in\mathcal{N}}v_N(S)}.
\end{align*}
Thus, we have $\mathsf{BF}(i,S)=\Mult(N(i),S)\mathsf{BF}(0,S)$, where $\Mult(N,S):=\left(\frac{\sum_{j\in N} v_j}{\sum_{j\in N\cap S} v_j}\right)^{1-\lambda_{N}}$ is a \textbf{nest-dependent multiplier} that, importantly, does not depend on the specific item.
For a nest $N\in\mathcal{N}$ such that $N\cap S\neq\emptyset$, we note the following about $\Mult(N,S)$:
\begin{enumerate}[I., nosep]
\item If $N\subseteq S$, then $\Mult(N,S)=1$, and hence $\mathsf{BF}(i,S)=\mathsf{BF}(0,S)$ for all $i\in N\cap S$;
\item If $N\nsubseteq S$, then $|N|\ge 2$, since $N\cap S\neq\emptyset$.  Under \Cref{ass:identify}, this means that $\lambda_N<1$, and hence $\Mult(N,S)>1$, with $\mathsf{BF}(i,S)=\mathsf{BF}(j,S)>\mathsf{BF}(0,S)$ for all $i,j\in N\cap S$.
\end{enumerate}

We make a "general position" assumption that, for two different nests not contained in $S$, their multipliers (which are strictly greater than 1) cannot coincidentally be the same.

\begin{assumption}[General Position] \label{ass:genPos}
For all $S\in\mathcal{S}$ and two different nests $N\neq N'$ such that $\emptyset\neq N\cap S\neq N$ and $\emptyset\neq N'\cap S\neq N'$, we have $\Mult(N,S)\neq\Mult(N',S)$.
\end{assumption}

Two items $i,j$ in the same nest are guaranteed to see the same boost factor, regardless of whether this boost factor equals $\mathsf{BF}(0,S)$.
\Cref{ass:genPos} ensures the converse: if $N(i)\neq N(j)$, then $\mathsf{BF}(i,S)\neq \mathsf{BF}(j,S)$, except when both $N(i)\subseteq S,N(j)\subseteq S$ and we see $\mathsf{BF}(i,S)=\mathsf{BF}(j,S)=\mathsf{BF}(0,S)$.

\begin{example} \label{eg:deductions}
Consider Experiment 2 from \Cref{table:introDesign}, which offers the assortment $S:=S_{1,-1}=\{\milk,\apple,\boba,\coffee\}$.  Suppose we see the boost factors $\mathsf{BF}(\milk,S)=1.3$, $\mathsf{BF}(\apple,S)=1.6$, $\mathsf{BF}(\boba,S)=1.9$, $\mathsf{BF}(\coffee,S)=1.9$, and $\mathsf{BF}(0,S)=1.3$.
We can take the contrapositives of the preceding observations to make the following deductions about nest membership.
\begin{itemize}[nosep]
\item $N(\milk)\subseteq S$, because if $N(\milk)\nsubseteq S$ then we would have seen $\mathsf{BF}(\milk,S)>\mathsf{BF}(0,S)$, which is not the case.  From this we can deduce that $N(\milk)\neq N(k)$ for all $k\notin S$.
\item $N(\boba)=N(\coffee)$, because if $N(\boba)\neq N(\coffee)$ then we could not coincidentally see $\mathsf{BF}(\boba,S)=\mathsf{BF}(\coffee,S)$, by \Cref{ass:genPos}.
\item $\{\milk\}$, $\{\apple\}$, and $\{\boba,\coffee\}$ are all part of different nests.  This follows by taking the contrapositive of the fact that items in the same nest see the same boost factor under each assortment $S$.
\end{itemize}
\end{example}

We summarize these observations in the following \namecref{prop:outsideNew}.

\begin{proposition}[Nest Deductions with Outside Option] \label{prop:outsideNew}
Suppose \Cref{ass:identify} and \ref{ass:genPos} hold and take any $S\in\mathcal{S}$.  For all $i\in S$:
\begin{enumerate}[I., nosep]
\item If $N(i)\subseteq S$, then $\mathsf{BF}(i,S)=\mathsf{BF}(0,S)$;
\item If $N(i)\nsubseteq S$, then $\mathsf{BF}(i,S)>\mathsf{BF}(0,S)$, and moreover for $j\in S\setminus\{i\}$, we have $N(i)=N(j)$ if and only if $\mathsf{BF}(i,S)=\mathsf{BF}(j,S)$.
\end{enumerate}
Therefore, we can make the following deductions about nest membership.
\begin{itemize}[nosep]
\item If $i\in S$ and $\mathsf{BF}(i,S)=\mathsf{BF}(0,S)$, then $N(i)\neq N(k)$ for all $k\notin S$.
\item If $i,j\in S$ and $\mathsf{BF}(i,S)=\mathsf{BF}(j,S)>\mathsf{BF}(0,S)$, then $N(i)=N(j)$.
\item If $i,j\in S$ and $\mathsf{BF}(i,S) \neq \mathsf{BF}(j,S)$, then $N(i)\neq N(j)$.
\end{itemize}
\end{proposition}
\subsection{Nest Identification Algorithm and Proof of Correctness} \label{sec:outsidePreTheo}

Our algorithms employ a graph-theoretic representation for the nest identification procedure. 
Each item corresponds to a vertex, and an edge exists between vertices $i$ and $j$ whenever $N(i)=N(j)$. 
The pairwise relations are recorded in a matrix $E \in \{0,1,\text{null}\}^{n \times n}$. Here, $E[i,j]=1$ indicates that $N(i)=N(j)$; $E[i,j]=0$ indicates that $N(i)\neq N(j)$; and $E[i,j]=\text{null}$ indicates that their relation is currently undetermined.
The identification task is complete once $E$ represents a collection of disjoint cliques, each corresponding to a distinct nest. 

Our nest identification algorithm is presented in Algorithm~\ref{alg:exactWithOutside}.  We note that the symmetry of matrix $E$ is preserved throughout the algorithm, and we ignore the values of $E$ along the diagonal, where we implicitly assume $i\neq j$ whenever we consider pairs $i,j$ in Algorithm~\ref{alg:exactWithOutside}.

\begin{algorithm*}
\caption{\textsc{Exact Nest Identification with Outside Option}}
\label{alg:exactWithOutside}
{\setlength{\baselineskip}{0.7\baselineskip}
\begin{algorithmic}[1]
\State Initialize adjacency matrix $E[i, j] \gets \text{null}$ for all $i, j \in [n]$
\For{$S \in \mathcal{S}$}
\For{$i,j\in S$}
\If{$\mathsf{BF}(i,S)\neq\mathsf{BF}(j,S)$}
\State $E[i,j]\gets 0$ \label{alg1:line5}
\ElsIf{$\mathsf{BF}(i,S)=\mathsf{BF}(j,S)>\mathsf{BF}(0,S)$}
\State $E[i,j]\gets 1$ \label{alg1:line7}
\EndIf
\EndFor
\State $\text{NoBoost}\gets\{i\in S:\mathsf{BF}(i,S)=\mathsf{BF}(0,S)\}$
\State $E[i,k]\gets 0$ and $E[k,i]\gets 0$ for all $i\in \text{NoBoost}$ and all $k\notin S$ \label{alg1:noboost}
\EndFor
\State $\text{OneHopTransitivity}\gets\{(i,j)\in[n]^2 : E[i,j]=\text{null}, \exists k\notin\{i,j\} \text{ s.t. } E[i,k]=E[j,k]=1\}$
\State \indent $E[i,j]\gets 1$ for all $(i,j)\in\text{OneHopTransitivity}$
\label{alg1:transitivity}
\State $\text{IdentifyMissingPairs}\gets\{(i,j)\in[n]^2: E[i,j]=\text{null},E[i,k]\neq 1\neq E[j,k]\ \forall k\notin\{i,j\}\}$
\State \indent $E[i,j]\gets 1$ for all $(i,j)\in \text{IdentifyMissingPairs}$ \label{alg1:uniqueNest}
\State $E[i,j] \gets 0$ for all $i,j\in[n]$ such that $E[i,j] = \text{null}$
\end{algorithmic}
}
\end{algorithm*}

\begin{theorem}[proven in Section~\ref{sec:mainResultPf}] \label{thm:mainResult}
Suppose that $\mathcal{S}$ is the experiment design from
Section~\ref{sec:ourExpDesign} and true market shares $\phi(i,S)$ are observed for all $S\in\mathcal{S}\cup\{[n]\}$ and $i\in S\cup \{0\}$.
Under Assumptions~\ref{ass:identify} and \ref{ass:genPos}, the adjacency matrix $E$ returned by Algorithm~\ref{alg:exactWithOutside} satisfies $E[i,j]=\mathbf{1}(N(i)=N(j))$ for all $i\neq j$, i.e., Algorithm~\ref{alg:exactWithOutside} is correct.
\end{theorem}

The correctness of the deductions on lines~\ref{alg1:line5}, \ref{alg1:line7}, and \ref{alg1:noboost} of Algorithm~\ref{alg:exactWithOutside} follow from \Cref{prop:outsideNew}. The proof of Theorem~\ref{thm:mainResult} is about showing that all nests are eventually found, i.e.\ the matrix $E$ is completed, which requires the "One Hop Transitivity" and "Identify Missing Pairs" operations.

\paragraph{Tightness.}
Algorithm~\ref{alg:exactWithOutside} and Theorem~\ref{thm:mainResult} show that our non-adaptive experimental design with $O(\log n)$  assortments is sufficient to guarantee nest identification.
We prove in \Cref{sec:logLowerBoundPf} that $\Omega(\log n)$ experimental assortments are also necessary, even if adaptively selected.

\paragraph{Extensions and loose ends.} We show in \Cref{sec:param_recovery} that, under a mild non-degeneracy condition (\Cref{ass:two_by_two_nonsingular}), our experiment design also allows successful recovery of the Nested Logit parameters $(v_i)_{i\in[n]}$, $(\lambda_N)_{N\in\cN}$, and $v_N$ for nests $N$ with $\lambda_N=0$, after identifying the nests.
We discuss the $d$-level Nested Logit model in \Cref{sec:dlevelLogit}.

\subsubsection{Illustration of \Cref{alg:exactWithOutside}.} \label{sssec:algIllustration}
To illustrate the logical deductions made by our nest identification algorithm, we return to the example from \Cref{table:introDesign}.
We provide hypothetical results for all experimental assortments in \Cref{table:coloredExp}, where we split the items in each experiment $S\in\mathcal{S}$ into equivalence classes based on their observed boost factors, further indicating the class that equals $\mathsf{BF}(0,S)$.
Our algorithm is purely combinatorial and does not use the exact numerical magnitudes of boost factors.

\begin{table}
    \centering
    \renewcommand{\arraystretch}{1.5}
    \newcommand{\noboost}{\(\texttt{--}\)}
    \newcommand{\boost}[1]{\(\uparrow(#1)\)}
    {\setlength{\baselineskip}{0.6\baselineskip}
    \begin{tabular}{c|c|c|c|c|c|c|c|c|}
    Item and Encoding
     &  \milk 000 &  \apple 001 &  \boba 010 &  \coffee 011 &  \tea 100 & \orange 101 &  \beer 110 & \wine 111 \\
    \hline
           Boost Factors in $S_{1,-0}$ 
        &  &  &  &  & \boost{1} & \boost{2} & \noboost & \noboost \\

        Boost Factors in $S_{1,-1}$  
        & \noboost & \boost{1} & \boost{2} & \boost{2} &  &  &  &  \\

        Boost Factors in $S_{2,-0}$ 
        &  &  & \boost{1} & \boost{1} &  &  & \noboost & \noboost \\

        Boost Factors in $S_{2,-1}$ 
        & \noboost & \noboost &  &  & \boost{1} & \noboost &  &  \\

        Boost Factors in $S_{3,-0}$ 
        &  & \noboost &  & \boost{1} &  & \noboost &  & \boost{2} \\

        Boost Factors in $S_{3,-1}$ 
        & \noboost &  & \boost{1} &  & \boost{1} &  & \boost{2} &  \\
    \end{tabular}}
    \caption{Hypothetical results for the experiment design from \Cref{table:introDesign}.  "$\texttt{--}$" indicates a boost factor equal to $\mathsf{BF}(0,S)$.  Items indicated by "$\uparrow(1)$" share the same boost factor  greater than $\mathsf{BF}(0,S)$; items indicated by "$\uparrow(2)$" share a different boost factor greater than $\mathsf{BF}(0,S)$, etc.
    }
    \label{table:coloredExp}
\end{table}

\Cref{fig:colored_iden2} then shows the evolution of the $E[i,j]$ matrix throughout the steps of Algorithm~\ref{alg:exactWithOutside}.
The first 6 steps make deductions from the 6 experimental assortments, and 
the final two steps demonstrate the importance of the "One Hop Transitivity" and "Identify Missing Pairs" operations.
The deductions from the
1st assortment were previously illustrated in \Cref{fig:illustrateDeductions}, while the deductions from
2nd assortment were previously illustrated in Example \ref{eg:deductions}.

\begin{figure} 
\centering
\includegraphics[width=\linewidth]{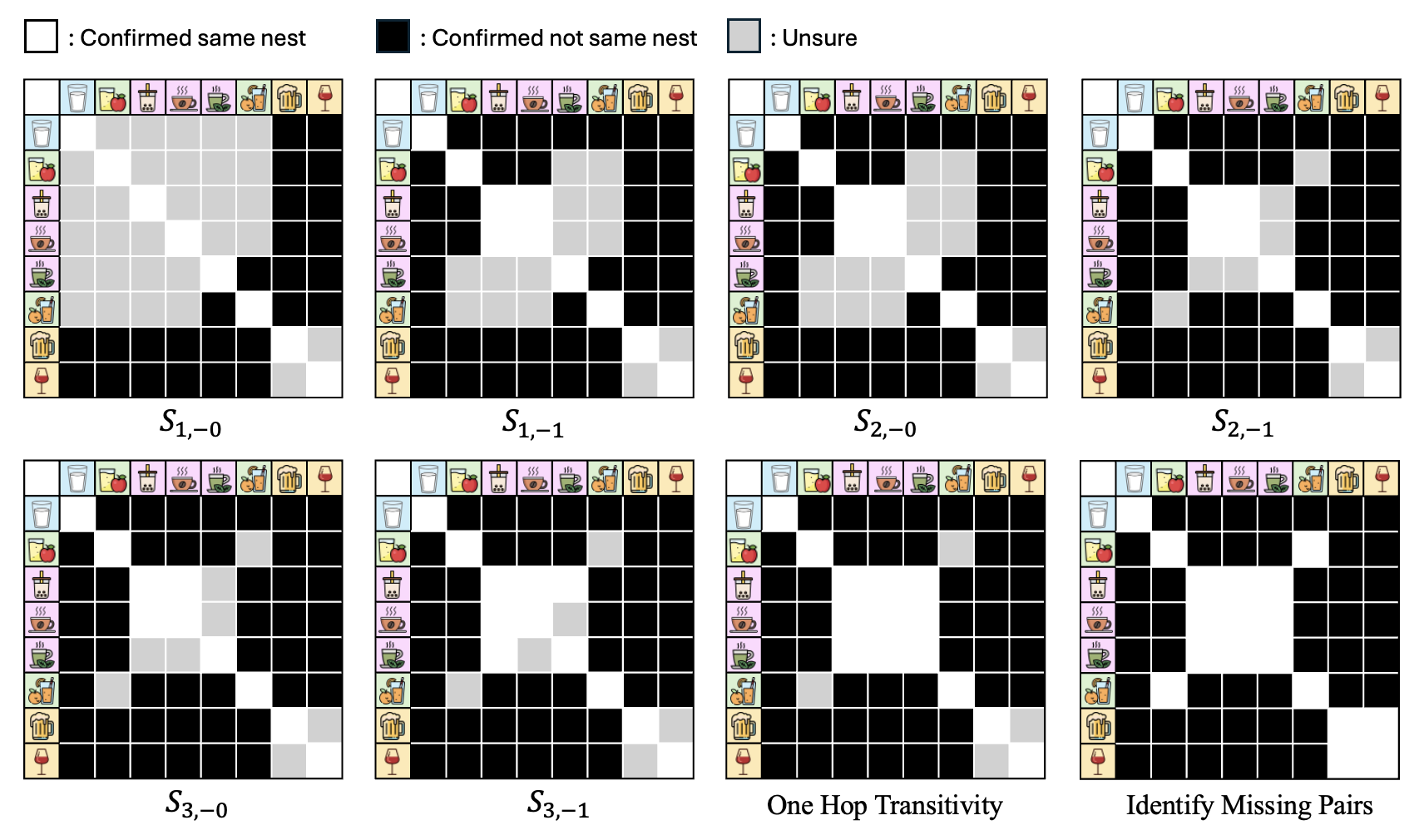}
\caption{Evolution of the adjacency matrix $E$ during nest identification.
White squares indicate $E[i,j]=1$ (same nest); black squares indicate $E[i,j]=0$ (different nests); while grey squares indicate $E[i,j]=\texttt{null}$ (not yet determined).
The state of the adjacency matrix $E$ is displayed after processing each of the 6 experimental assortments $S\in\mathcal{S}$, and after the "One Hop Transitivity" (line~\eqref{alg1:transitivity}) and "Identify Missing Pairs" (line~\eqref{alg1:uniqueNest}) operations.}
\label{fig:colored_iden2}
\end{figure}

\subsection{Finite-sample Guarantee} \label{sec:finiteSample}

Recall that Algorithm~\ref{alg:exactWithOutside} identifies the nest partition
\(\mathcal{N}\) by comparing boost factors
\[
\mathsf{BF}(i,S)=\frac{\phi(i,S)}{\phi(i,[n])},
\]
across \(i\in S\cup\{0\}\) for each experimental assortment
\(S\in\mathcal{S}\). We now consider the finite-sample setting, where each
assortment \(S\in\mathcal{S}\cup\{[n]\}\) is offered to \(m\) independent
customers, and \(\hphi(i,S)\) denotes the empirical choice probability of
alternative \(i\in S\cup\{0\}\).

A direct comparison of empirical boost factors
\(\hphi(i,S)/\hphi(i,[n])\) is undesirable because it involves ratios of noisy
estimates. Instead, we use the equivalent comparison.
Define
\[
q_{ij}(S)
:=
\frac{\phi(i,S)}{\phi(i,S)+\phi(j,S)} .
\]
Then
\begin{equation}
\mathsf{BF}(i,S)\ge \mathsf{BF}(j,S)
\quad \Longleftrightarrow \quad
q_{ij}(S)\ge q_{ij}([n]).
\label{eqn:statComparison}
\end{equation}
for all \(S\in\mathcal{S}\) and \(i,j\in S\cup\{0\}\). Thus, each
boost-factor comparison can be viewed as a two-proportion comparison: under
assortment \(S\) versus the control assortment \([n]\), we compare the
probability of choosing \(i\) conditional on choosing either \(i\) or \(j\).

We therefore define the pooled two-proportion statistic $z(i\succ j,S)$ to be
\begin{align}
\frac{
\frac{\hphi(i,S)}{\hphi(i,S)+\hphi(j,S)}
-
\frac{\hphi(i,[n])}{\hphi(i,[n])+\hphi(j,[n])}
}{
\sqrt{
\hat p_{ij,S}\left(1-\hat p_{ij,S}\right)
\left(
\frac{1/m}{\hphi(i,S)+\hphi(j,S)}
+
\frac{1/m}{\hphi(i,[n])+\hphi(j,[n])}
\right)
}
},
\label{eqn:zstat}
\end{align}
where
\[
\hat p_{ij,S}
=
\frac{\hphi(i,S)+\hphi(i,[n])}
{\hphi(i,S)+\hphi(j,S)+\hphi(i,[n])+\hphi(j,[n])}
\]
is the pooled empirical probability of choosing \(i\) conditional on choosing
either \(i\) or \(j\). Large positive values of \(z(i\succ j,S)\) provide
evidence that \(\mathsf{BF}(i,S)>\mathsf{BF}(j,S)\), large negative values
provide evidence that \(\mathsf{BF}(i,S)<\mathsf{BF}(j,S)\), and values close
to zero are treated as statistically indistinguishable boost factors. Note that
\(z(i\succ j,S)=-z(j\succ i,S)\).

In our empirical implementation in \Cref{sec:numericNestId}, we compare
\(z(i\succ j,S)\) to fixed normal critical values. For the theoretical guarantee
below, we instead choose a threshold depending on \(n\) and the desired failure
probability \(\delta\), and derive a sample-size condition on \(m\) under which
all boost-factor comparisons used by Algorithm~\ref{alg:exactWithOutside} are
correct with probability at least \(1-\delta\).

\begin{theorem}[proven in \Cref{sec:sampleComplexityPf}]\label{thm:sampleComplexity}
Take constants $\rho,\Delta,\delta\in(0,1)$ such that
$\phi(i,[n])\ge\rho$ for all $i\in [n] \cup \{0\}$ and $\forall S\in\mathcal{S};i,j\in S\cup\{0\}\text{ s.t. }\mathsf{BF}(i,S)\neq\mathsf{BF}(j,S)$, 
\begin{align*}
\left|\frac{\phi(i,S)}{\phi(i,S)+\phi(j,S)}-\frac{\phi(i,[n])}{\phi(i,[n])+\phi(j,[n])}\right| &\ge\Delta.
\end{align*}

If the number of samples satisfies
$$
m\ge 3C^2 \frac{\log(2K/\delta)}{\rho\Delta^2},
$$
then with probability at least $1-\delta$, for all $S\in\mathcal{S}$ and $i,j\in S\cup\{0\}$, we have $|z(i\succ j,S)| \le 8\sqrt{3\log(2K/\delta)}$ if and only if $\mathsf{BF}(i, S) = \mathsf{BF}(j, S)$, where $K:=(|\mathcal{S}|+1)(n+1+\binom{n+1}2)$ and $C$ is an absolute numerical constant.
\end{theorem}

We note that the dependence of the number of samples on $\rho\Delta^2$ is essentially tight, because an additive error of $O(\Delta)$ is needed to distinguish whether boost factors are identical, and the useful samples (where $i$ or $j$ is chosen) can be diluted by a factor of $\rho$.
The proof of Theorem~\ref{thm:sampleComplexity} standardly takes a union bound with $K$ multiplicative Chernoff bounds that each fail with probability $\delta/K$; however, several specialized tricks are needed to achieve this tight dependence on $\rho\Delta^2$, which leverage the complicated form of the test statistic $z(i\succ j, S)$ in~\eqref{eqn:zstat}.
Theorem~\ref{thm:sampleComplexity} implies that in Algorithm~\ref{alg:exactWithOutside}, if we replace all conditions $\mathsf{BF}(i,S)=\mathsf{BF}(j,S)$ with $|z(i\succ j,S)|\le8\sqrt{3\log(2K/\delta)}$ and all conditions $\mathsf{BF}(i,S)\neq\mathsf{BF}(j,S)$ (including $\mathsf{BF}(i,S)>\mathsf{BF}(0,S)$) with $|z(i\succ j,S)|>8\sqrt{3\log(2K/\delta)}$, then with probability at least $1-\delta$, no mistake deductions are made from any experimental assortments $S\in\mathcal{S}$ and hence the correct nest partition $\mathcal{N}$ is identified as per Theorem~\ref{thm:mainResult}.

\subsection{Nest Identification without the Outside Option} \label{sec:outsideAbsTheo}
We now show how our identifiability and the 
finite-sample guarantee continue to hold if there is no outside option that can be chosen by customers.  Nest identification becomes more difficult because we can no longer compare to the boost factor of the outside option, $\mathsf{BF}(0,S)$.

In each assortment $S\in\mathcal{S}\cup\{[n]\}$, the choice probabilities are now defined by 
\begin{align}
\label{eqn:noOutsideOption}
\phi(i,S)
= \frac{v_{N(i)}(S)}{\sum_{N\in\mathcal{N}} v_N(S)}
\cdot
\frac{v_i}{\sum_{j\in N(i)\cap S} v_j},
\qquad i\in S.
\end{align}
This satisfies $\sum_{i\in S} \phi(i,S) = 1$; note that the denominator no longer contains a "+1" corresponding to the outside option. 
The definition of boost factor remains the same as in Definition \ref{def:bf}, i.e.~$\mathsf{BF}(i, S) := \frac{\phi(i, S)}{\phi(i, [n])}$, but now both $\phi(i,S)$ and $\phi(i,[n])$ are defined under equation~\eqref{eqn:noOutsideOption} instead.
When we evaluate this expression, we obtain
\begin{align*}
\mathsf{BF}(i,S)
&=
\frac{\frac{v_{N(i)}(S)}{\sum_{N\in\mathcal{N}}v_N(S)}\cdot\frac{v_i}{\sum_{j\in N(i)\cap S} v_j}}
{\frac{v_{N(i)}([n])}{\sum_{N\in\mathcal{N}}v_N([n])}\cdot\frac{v_i}{\sum_{j\in N(i)} v_j}} 
\\
&=\Mult(N(i),S)\cdot
\frac{\sum_{N\in\mathcal{N}} v_N([n])}{\sum_{N\in\mathcal{N}} v_N(S)}
&\forall i\in S.
\end{align*}
$\Mult(N,S):=\left(\frac{\sum_{j\in N} v_j}{\sum_{j\in N\cap S} v_j}\right)^{1-\lambda_{N}}$ is defined the same as before, and satisfies $\Mult(N,S)\ge1$ with strict inequality if and only if $N\nsubseteq S$ (under \Cref{ass:identify}), for each nest $N\in\mathcal{N}$ and experimental assortment $S\in\mathcal{S}$.
However, we can no longer check whether $\Mult(N(i),S)=1$ by checking whether $\mathsf{BF}(i,S)=\mathsf{BF}(0,S)$.

Regardless, we can still summarize these observations without the outside option in the following \namecref{prop:noOutsideNew}, which leads to more ambiguous deductions compared to the previous \Cref{prop:outsideNew}.

\begin{proposition}[Nest Deductions without Outside Option] \label{prop:noOutsideNew}
Suppose the same \Cref{ass:identify} and \ref{ass:genPos} hold and take any $S\in\mathcal{S}$.  For all $i\in S$:
\begin{enumerate}[I.]
\item If $N(i)\subseteq S$, then we see that $\mathsf{BF}(i,S)\le\mathsf{BF}(k,S)$ for all $k\in S$ (in particular, $\mathsf{BF}(i,S)$ is identical for all items $i$ such that $N(i)\subseteq S$);
\item If $N(i)\nsubseteq S$, then  for $j\in S\setminus\{i\}$, we have $N(i)=N(j)$ if and only if $\mathsf{BF}(i,S)=\mathsf{BF}(j,S)$.
\end{enumerate}
Taking contrapositives, we can make the following deductions about nest membership.
\begin{itemize}
\item Define $\minBF(S):=\argmin_{k\in S}\mathsf{BF}(k,S)$ to be the set of items in $S$ attaining the minimum boost factor.  If $N(i)\subseteq S$ for any $i\in S$, then $\minBF(S)=\bigcup_{N\in\mathcal{N}:N\subseteq S} N$; otherwise, $\minBF(S)\subsetneq N$ for a single nest $N\in\mathcal{N}$ that is not fully contained in $S$.
\item If $i,j\in S$ and $\mathsf{BF}(i,S)=\mathsf{BF}(j,S)>\min_{k\in S}\mathsf{BF}(k,S)$, then $N(i)=N(j)$.
\item If $i,j\in S$ and $\mathsf{BF}(i,S)\neq \mathsf{BF}(j,S)$, then $N(i)\neq N(j)$.
\end{itemize}
\end{proposition}

We now present our nest identification algorithm and main result without the outside option.

\begin{algorithm*}
\caption{\textsc{Exact Nest Identification without Outside Option}}
\label{alg:exactWithoutOutside}
{\setlength{\baselineskip}{0.7\baselineskip}
\begin{algorithmic}[1]
\State Initialize adjacency matrix $E[i, j] \gets \text{null}$ for all $i, j \in [n]$
\For{$S \in \mathcal{S}$; $i,j\in S$}
\If{$\mathsf{BF}(i,S)\neq\mathsf{BF}(j,S)$}
\State $E[i,j]\gets 0$ \label{alg2:line4}
\ElsIf{$\mathsf{BF}(i,S)=\mathsf{BF}(j,S)>\min_{k\in S}\mathsf{BF}(k,S)$}
\State $E[i,j]\gets 1$ \label{alg2:pos_deduct1}
\EndIf
\EndFor
\For{$S\in\mathcal{S}$, $\minBF(S):=\argmin_{k\in S}\mathsf{BF}(k,S)$}
\If{$E[i,j]=0$ for some $i,j\in\minBF(S)$ with $i\neq j$}
\State $E[i,k]\gets 0,\ E[k,i]\gets 0$ for all $i\in \minBF(S),k\notin S$ \label{alg2:neg_deduct}
\Else
\State $E[i,j]\gets 1$ for all $i,j\in \minBF(S)$ \label{alg2:pos_deduct2}
\EndIf
\EndFor
\State $\text{OneHopTransitivity}\gets\{(i,j)\in[n]^2 : E[i,j]=\text{null}, \exists k\notin\{i,j\} \text{ s.t. } E[i,k]=E[j,k]=1\}$
\State \indent $E[i,j]\gets 1$ for all $(i,j)\in \text{OneHopTransitivity}$ \label{alg2:transitivity}
\State $\text{IdentifyMissingPairs}\gets\{(i,j)\in[n]^2: E[i,j]=\text{null},E[i,k]\neq 1\neq E[j,k]\ \forall k\notin\{i,j\}\}$
\State \indent $E[i,j]\gets 1$ for all $(i,j)\in \text{IdentifyMissingPairs}$ \label{alg2:uniqueNest}
\State $E[i,j] \gets 0$ for all $i,j\in[n]$ such that $E[i,j] = \text{null}$
\end{algorithmic}}
\end{algorithm*}

\begin{theorem}[proven in \Cref{sec:idenWithoutOutsidePf}] \label{thm:idenWithoutOutside}
Suppose that $\mathcal{S}$ is the experiment design from
Section~\ref{sec:ourExpDesign} and true market shares $\phi(i,S)$ are observed for all $S\in\mathcal{S}\cup\{[n]\}$ and $i\in S$.
Under Assumptions~\ref{ass:identify} and \ref{ass:genPos}, the adjacency matrix $E$ returned by Algorithm~\ref{alg:exactWithoutOutside} satisfies $E[i,j]=\mathbf{1}(N(i)=N(j))$ for all $i\neq j$, except in the case $|N(i)|=|N(j)|=1$, where $E[i,j]$ may be 1 (but this does not affect the correctness of the induced choice function $\phi(i, S)$).
\end{theorem}

Unlike Theorem~\ref{thm:mainResult}, in Theorem~\ref{thm:idenWithoutOutside} it is possible for $i,j$ to be incorrectly put into the same nest $N$ when in reality $|N(i)|=|N(j)|=1$.  However, in this case, the correct choice function $\phi$ can still be constructed, as nest $N$ effectively consists of singletons if parameter $\lambda_N$ is recovered to be 1 (see \Cref{ass:identify}).

Finally, we remark that the same finite-sample guarantees from \Cref{sec:finiteSample} can be easily translated to this setting without the outside option.  Indeed, Algorithm~\ref{alg:exactWithoutOutside} is based entirely on comparing whether two boost factors are the same, and Theorem~\ref{thm:sampleComplexity} provides a guarantee that ensures no mistakes are made when determining whether two boost factors are the same.

\section{Numerical Comparison of Experiment Designs} \label{sec:numericExpDesign}

We compare our combinatorial experiment design to randomized assortments and other naive designs for data collection, under mis-specified (\Cref{ssec:misSpec}) and well-specified (\Cref{ssec:wellSpec}) choice estimation. We generally try to replicate the ground truths and estimation procedures from \citet{berbeglia2022comparative}, changing only the data collection part of the pipeline.

\subsection{Setup} \label{ssec:numericSetup}

We consider estimation instances defined by
an unknown ground truth choice function $\phi$ over $n$ items, and a fixed budget of $T$ customers. 
An experiment design selects an assortment $S$ for each customer, from which we observe an IID choice drawn according to $\phi$.

We compare the following experiment designs:
\begin{itemize}[nosep]
\item \textbf{Our design} with base $b=2$, which prescribes $|\mathcal{S}|=2\lceil\log_2 n\rceil$ experimental assortments plus a control assortment, offering each one to $\approx T/(|\mathcal{S}|+1)$ customers;
\item \textbf{Randomized design}: randomly generates $G$ assortments and offers each one to $\approx T/G$ customers;
we consider $G=2\lceil\log_2 n\rceil+1$ which uses the same number of distinct assortments as our design,
$G=T/10$ which matches the original design of \citet{berbeglia2022comparative},
and $G=T$ which means each assortment is only offered once;
\item \textbf{Leave-one-out}: a design from \citet{blanchet2016markov} that offers each of the $n+1$ assortments $[n],[n]\setminus 1,\ldots,[n]\setminus n$ to $\approx T/(n+1)$ customers;
\item \textbf{Incremental}: a naive design that offers each of the $n$ assortments $\{1\},\{1,2\},\ldots,\{1,\ldots,n\}$ (after a random shuffling of indices) to $\approx T/n$ customers.
\end{itemize}
We re-iterate that our experiment design generally requires the least number of distinct assortments and hence should also be the cheapest to deploy, a further benefit not reflected in our comparisons.

Under any experiment design, we consider multiple ways to estimate a choice function $\phi^\text{est}$ from the observations, each of which fits a family of choice models using an estimation method.  
We copy exactly five estimation methods and code from \citet{berbeglia2022comparative}:
Exponomial (exp), Latent-class MNL (lc), Markov Chain (mkv),  Multinomial Logit (mnl), and Nested Logit with an arbitrary partition of items into two nests (nl-given).
For further details, see Appendix~\ref{app:berbegliaInsts}.
We add a sixth estimation method, which is Nested Logit with our nest identification Algorithm~\ref{alg:noisyNestWithOutside} (nl-our).


To evaluate how accurately the estimated $\phi^\text{est}$ approximates the ground truth $\phi$, we use the
soft RMSE metric like in \cite{berbeglia2022comparative}, defined as follows:
\begin{equation} \label{eqn:softRMSE}
\RMSEsoft(\phi, \phi^\text{est})
= 
\sqrt{
\frac{
\sum_{S \subseteq [n]} \sum_{i \in S \cup \{0\}}
\left( \phi(i, S) - \phi^\text{est}(i, S) \right)^2
}{
\sum_{S \subseteq [n]} (|S| + 1)
}
}.
\end{equation}
$\RMSEsoft$ evaluates the squared difference between the estimated choice probability $\phi^\text{est}(i,S)$ and the true choice probability $\phi(i,S)$, taking an average over all (item, assortment) pairs $(i,S)$.

\subsection{Mis-specified Instances from \citet{berbeglia2022comparative}}\label{ssec:misSpec}

We replicate 1440 estimation instances (see \Cref{app:berbegliaInsts}) from \citet{berbeglia2022comparative}, each defined by a choice function $\phi$ and a data size $T$, with 360 instances for each $T \in \{300,750,3000,6000\}$. 
There are always $n=10$ items, so our design uses $2\lceil\log_2 n\rceil+1=9$ assortments.
The choice functions are generated from a Ranked List model, implying that all estimated models are mis-specified. 

Here we use the binary encoding and balanced version of our experiment design.
This ensures that each item appears in half of the $8$ experimental assortments and each experimental assortment contains exactly $5$ items. 
For the randomized designs, we generate random assortments with size drawn uniformly from $\{3,4,5,6\}$, following \citet{berbeglia2022comparative}.
The randomized design results with $G=T/10$ pull numbers from their paper verbatim, while the other experiment designs (including the randomized one with $G=9$) draw fresh choice data using our own random seeds.
The soft RMSE results under different estimated choice models are displayed in \Cref{fig:misSpecified}.

\begin{figure}
\centering
\includegraphics[width=0.95\linewidth]{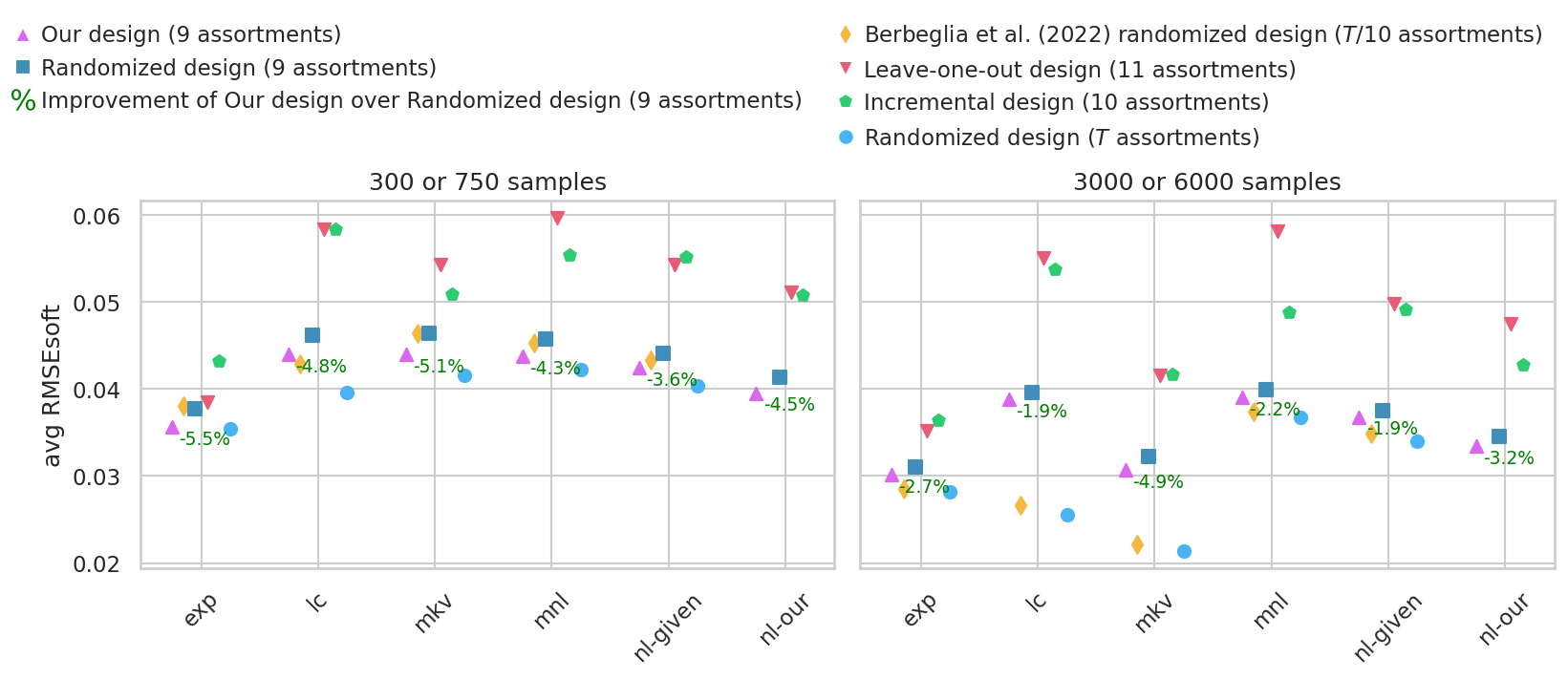}
\caption{Comparing experiment designs in a mis-specified setting}
\footnotesize

\textbf{Left}: average $\RMSEsoft$ over the 720 instances with $T\in\{300,750\}$

\textbf{Right}: average $\RMSEsoft$ over the 720 instances with $T\in\{3000,6000\}$

\label{fig:misSpecified}
\end{figure}

\paragraph{Findings.}
Our design with 9 deliberate assortments consistently outperforms the design with 9 randomized assortments, especially for small data sizes where $T\in\{300,750\}$.
The advantage subsides for $T\in\{3000,6000\}$, where our design is beaten by the original design from \citet{berbeglia2022comparative} which randomly draws $T/10$ (i.e., 300 or 600) experimental assortments.
The difference is particularly stark for the  Exponomial, Latent-class MNL and Markov Chain choice models, which have a lot of parameters, and hence our 9 deliberate assortments do not provide enough identification power.
The design with $T$ individual randomized assortments is generally best under mis-specified choice models, although
it is often impractical to offer so many distinct assortments.
Moreover, under well-specified choice models in \Cref{ssec:wellSpec}, we will see that our experiment design can be competitive even if the randomized design has far more experimental assortments.

Meanwhile, we note that the Leave-one-out and Incremental designs never perform well.  This is because they are offering too many assortments that are either too big or too small, instead of focusing on assortments of size around $n/2$ (see the discussion in \Cref{sec:simpleDesigns}).
Finally, we see that using our nest identification algorithm increases the appeal of Nested Logit especially for $T\in\{300,750\}$. 
Within the broader class of logit-family estimators (lc, mnl, nl-given), our design with our nest identification algorithm  
can even outperform estimators paired with the richer $T/10$ and $T$ randomized assortments, which are the empirically stronger design in this setting.

\subsection{Well-Specified Models}\label{ssec:wellSpec}

We now consider well-specified settings, where the ground truth choice function $\phi$ in the instance aligns with the choice model being fit by the estimation procedure.
We run the code of \citet{berbeglia2022comparative} to generate random ground truths, adhering to their parameter ranges but noting that these exact instances are freshly generated.
Because we have the flexibility to generate arbitrarily many ground truths, we hereafter report 95\% confidence intervals (see \Cref{app:confidenceIntervals}) for the average $\RMSEsoft$ values across instances; in general, we try to generate enough instances to have disjoint intervals around the average $\RMSEsoft$'s of different experiment designs.
We only run each design once on each instance, preferring to generate new instances instead of repeating estimation procedures and re-generating data for the same instance.

We fix the number of items to be $n=16$ and let $T\in\{270,450,900,1800,2700,4500,6750,9000\}$.
For the randomized experiment designs, we generate random assortments of size $n/2$, which we generally found to be best.
We first show results when both the ground truth and estimation procedure follow a Markov Chain choice model, in \Cref{fig:mkvMkv}, which is based on 500 random $\phi$'s.

\begin{figure}
\centering
\begin{minipage}[c]{0.63\textwidth}
\centering
\includegraphics[width=\linewidth]{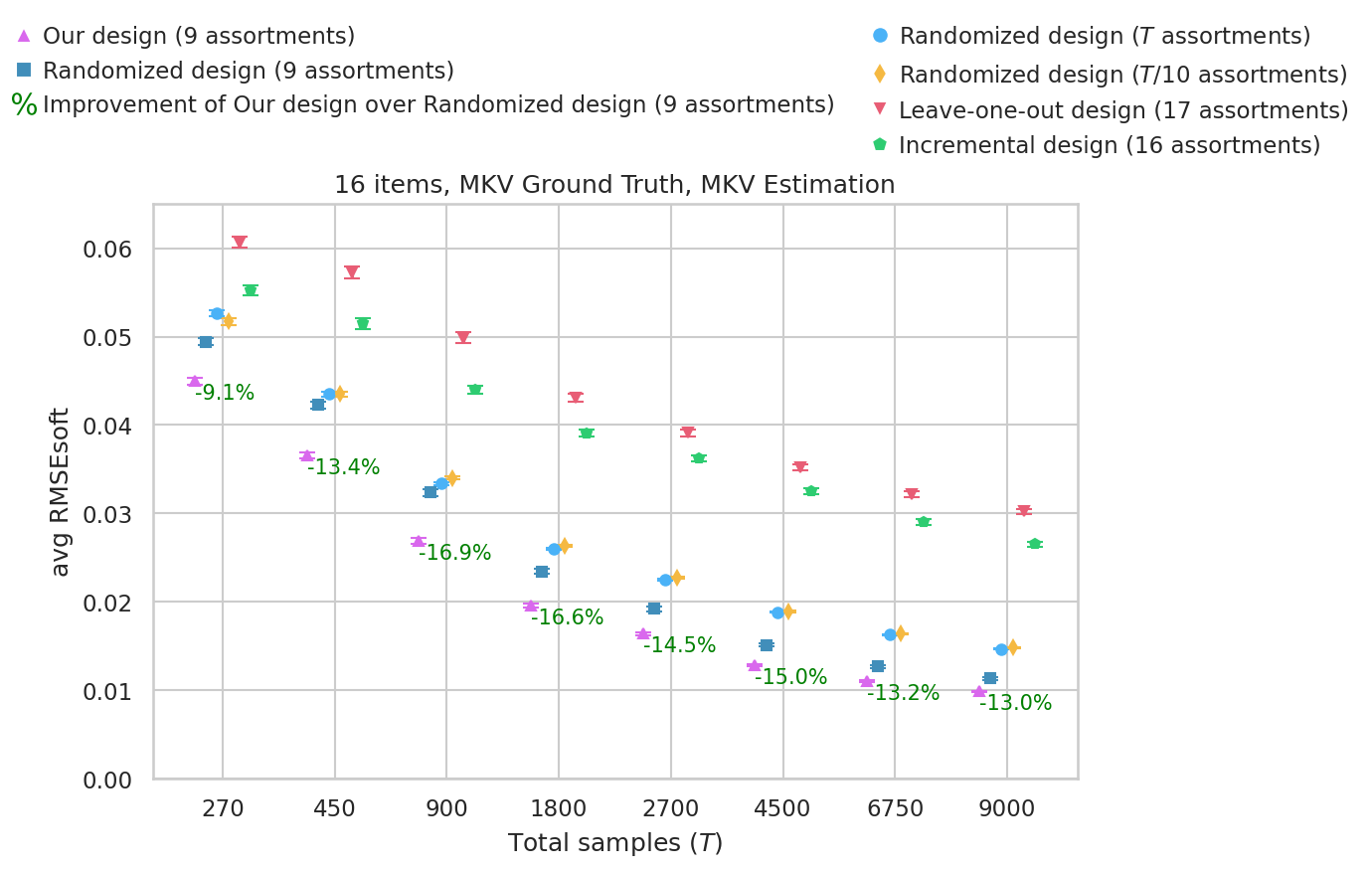}
\end{minipage}%
\hfill
\begin{minipage}[c]{0.33\textwidth}
\caption{Comparing experiment designs in a well-specified setting, where we display average $\RMSEsoft$ over the 500 Markov Chain ground truths under Markov Chain choice estimation}
\label{fig:mkvMkv}
\end{minipage}
\end{figure}

\paragraph{Findings.} Our experiment design dominates in \Cref{fig:mkvMkv}, reducing the soft RMSE compared to 9 randomized assortments by as much as 16.9\%, and even beating designs that can draw $T/10$ or $T$ (i.e., individualized) randomized assortments.  
The 95\% confidence intervals are also separated.

This dominant performance compared to \Cref{fig:misSpecified} suggests that back in \Cref{ssec:misSpec}, randomization in the experimental assortments could have helped "regularize" the model mis-specification.
Here without model mis-specification, we see that our structured experiment design is always preferred, even if the randomized design can have far more experimental assortments.
We also test well-specified estimation on the Exponomial and MNL choice models in \Cref{app:otherWellSpec}, and encounter similar (although weaker) findings.

Finally, we note that the Leave-one-out design was considered by \citet{blanchet2016markov} \textit{specifically for Markov Chain} choice estimation, yet it performs poorly in this exact setting (until one reaches asymptotic data sizes; see \Cref{app:mkvAsymptotic}).  This demonstrates the risks of using asymptotic optimality as a criterion for experiment design in practice, a point we will reinforce in \Cref{ssec:numericNestIdWO}.

\section{Numerical Comparison of Nest Identification Algorithms}\label{sec:numericNestId}
We test our nest identification algorithms on synthetic Nested Logit instances with (\Cref{ssec:numericNestIdWith}) and without (\Cref{ssec:numericNestIdWO}) an outside option, and on the public SFWork dataset (\Cref{ssec:sf}), which has a fixed experiment design and no outside option.  We run modifications of our \Cref{alg:exactWithOutside,alg:exactWithoutOutside}, that use statistical tests to handle noise from small data, and community detection algorithms to resolve inconsistencies in nest membership (see \Cref{app:algs3and4}).

\subsection{Synthetic Nested Logit instances with Outside Option}\label{ssec:numericNestIdWith}

We use the same terminology and setup as described in \Cref{ssec:numericSetup}.
Our estimation instances follow a Nested Logit ground truth over $n=16$ items, where parameters $(v_i)_{i\in[n]}$ and $(\lambda_N)_{N\in\mathcal{N}}$ are randomly generated following \citet{berbeglia2022comparative}, but importantly, we change the nest partition $\mathcal{N}$ to also be randomly generated each time (see \Cref{app:randomNestGeneration} for details).
We generate 500 random Nested Logit choice functions $\phi$ in this way and consider larger values \(T\in\{9{,}000,\ 45{,}000,\ 90{,}000,\ 180{,}000,\ 450{,}000\}\), which are needed to identify the nest partitions.
In addition to the $\RMSEsoft$ metric discussed in \Cref{ssec:numericSetup}, we also report the Rand index \citep{rand1971objective} between the estimated nest partition $\mathcal{N}^\mathrm{est}$ and the true one $\mathcal{N}$, which measures the \% of item pairs $i,j\in[n]$ for which $\mathcal{N}^\mathrm{est}$ correctly declares whether $i$ and $j$ are in the same nest.

We evaluate the soft RMSE and Rand index of our nest identification algorithm on these instances, both under our experiment design and under a randomized design with the same number (8) of experimental assortments drawn uniformly among subsets of size $n/2$ (plus a control assortment).
We also compare to the "default" Nested Logit estimation procedure in \citet{berbeglia2022comparative}, which fixes an arbitrary division\footnote{Because we changed the Nested Logit instance generation to have a random nest partition, this "default" fixed nest structure is no longer correct.  Therefore, it is a weak baseline.} of the items into 2 nests.  We note however that after determining the nests, we always use the same code (from \citet{berbeglia2022comparative}) to estimate the Nested Logit parameters $(v_i)_{i\in[n]}$ and $(\lambda_N)_{N\in\mathcal{N}}$.
The results are displayed in \Cref{fig:ourVsRest_v2}.

\begin{figure}
\centering
\includegraphics[width=0.95\linewidth]{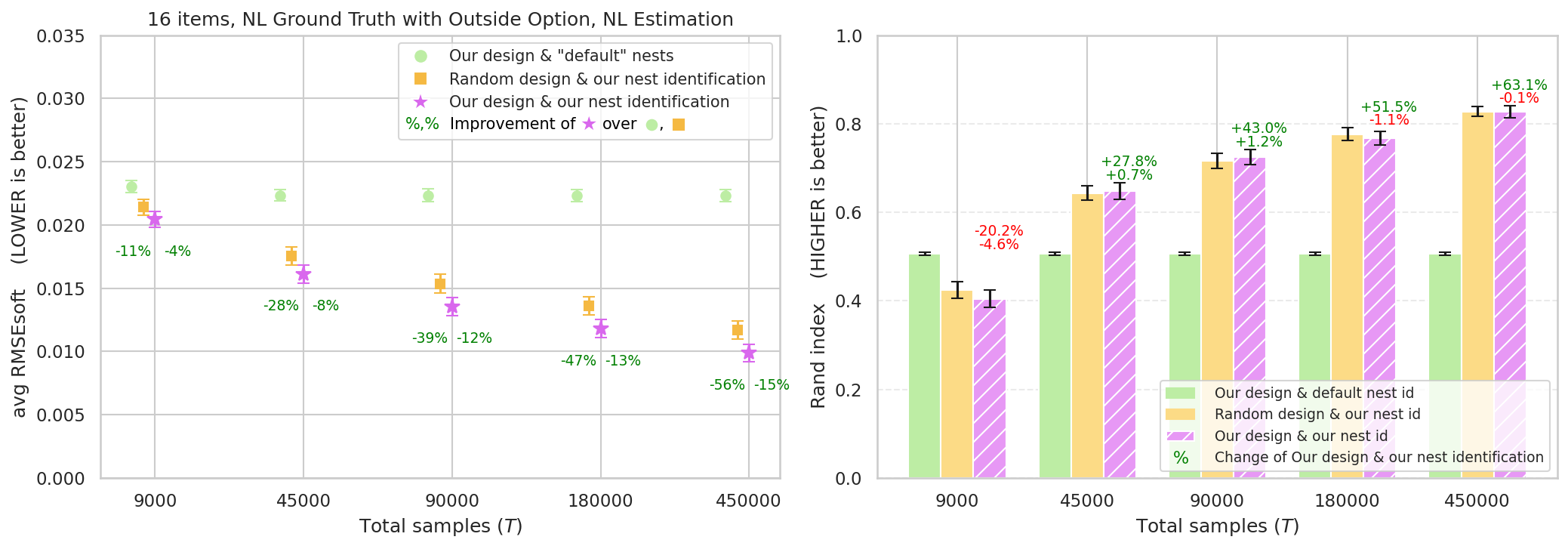}
\caption{Comparing experiment designs and nest identification algorithms, averaged over the 500 Nested Logit ground truths \textbf{with} an outside option}
\label{fig:ourVsRest_v2}
\end{figure}

\paragraph{Findings.}
Our experiment design with our nest identification algorithm performs best at all data sizes, under both metrics $\RMSEsoft$ and Rand index.
Changing our design to a random design hurts performance.
Unsurprisingly, the "default" nests do not perform well at all.

\subsection{Synthetic Nested Logit instances without Outside Option}\label{ssec:numericNestIdWO}

As in \Cref{ssec:numericNestIdWith}, we consider 500 random Nested Logit ground truths over $n=16$ items, and values of $T \in \{9{,}000, 45{,}000, 90{,}000, 180{,}000, 450{,}000\}$.  However, now there is no outside option in the definition of the choice function $\phi$ (see \Cref{sec:outsideAbsTheo}).

We can now compare to the nest identification algorithm of \citet{benson2016relevance}, which is designed for the setting with no outside option.  We compare our nest identification algorithm with data collected from our experiment design, to their nest identification algorithm with data collected from their experiment design (see \Cref{app:BensonImplementation} for implementation details).
We also test their nest identification algorithm combined with our experiment design, noting that we cannot test the converse, because our nest identification algorithm relies on repeat observations, while their experiment design draws fresh assortments.
The results are shown in \Cref{fig:ourVsBenson}.

\begin{figure}
\centering
\includegraphics[width=0.95\linewidth]{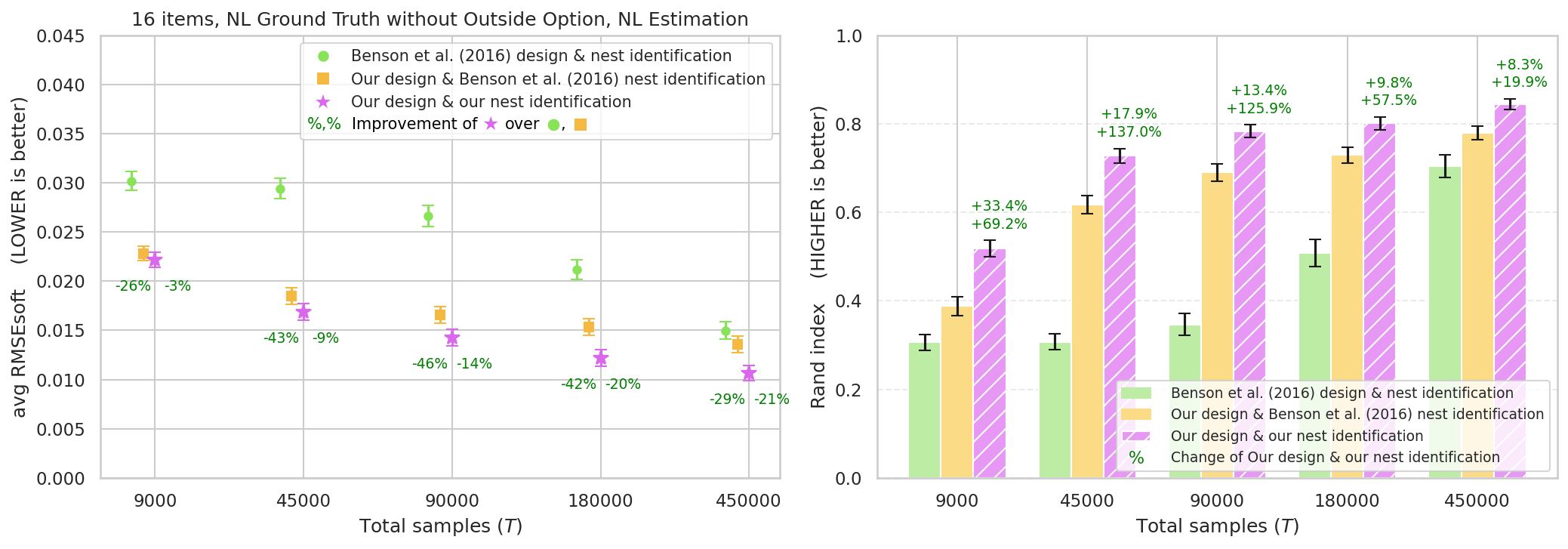}
\caption{Comparing experiment designs and nest identification algorithms, averaged over the 500 Nested Logit ground truths \textbf{without} an outside option}
\label{fig:ourVsBenson}
\end{figure}

\paragraph{Findings.}
Surprisingly, even for the nest identification algorithm of \citet{benson2016relevance}, our experiment design is much better than theirs.
In fact, switching to our experiment design reduces the average $\RMSEsoft$ by as much as 46\%, and makes their nest identification algorithm hard to beat, except in the asymptotic regime where it will be biased.  
That being said, their nest identification algorithm aggregates observations across different assortments to reduce variance, perform comparably to ours in the smallest data regime of $T= 9{,}000,$.

In light of \Cref{fig:ourVsBenson}, the main advantage of our nest identification algorithm over that of \citet{benson2016relevance} is that it encourages the (much) better experiment design.
In order to be theoretically unbiased, their nest identification algorithm requires assortments of sizes 2 and 3 (see \Cref{app:BensonImplementation}), which are much less informative for data collection in non-asymptotic numerical settings.

\subsection{SFWork dataset}\label{ssec:sf}

We consider choice estimation on the public SFWork dataset \citep{KoppelmanBhat2006}, which has $n=6$ options (without an outside option) for work commute in San Francisco.  There are $T=5000$ total observed choices from 12 distinct assortments.
Unlike the synthetic instances studied in \Cref{sec:numericExpDesign,sec:numericNestId} so far, here we cannot evaluate the true $\RMSEsoft$ because we do not know the ground truth, and we cannot change the experiment design and generate fresh observations.
Consequently, we measure $\RMSEsoft$ over only the fixed assortments observed in the dataset.
Let
\[
\mathcal A(S) :=
\begin{cases}
S\cup\{0\}, & \text{if the outside option is observed},\\
S, & \text{otherwise}.
\end{cases}
\]
We define
\begin{align}
\label{eqn:rmsesoftMod}
\RMSEsoft(\phi,\phi^\text{est})
=
\sqrt{
\frac{
\sum_{S\in\mathcal S}\sum_{i\in\mathcal A(S)}
\left(\phi(i,S)-\phi^\text{est}(i,S)\right)^2
}{
\sum_{S\in\mathcal S}|\mathcal A(S)|
}
}.
\end{align}

where $\mathcal{S}$ denotes the collection of assortments observed (i.e.,~$|\mathcal{S}|=12$ in this case).
We split the data into train/test and define the ground truth choice probabilities $\phi(\cdot,S)$ for each $S\in\mathcal{S}$ based on the empirical probabilities in the held-out testing set.

Because we do not know the ground truth, here we estimate different choice models like in the mis-specified setting (\Cref{ssec:misSpec}), including MNL (commonly used low-parameter model), Markov Chain (commonly used for choice prediction under sufficient data), and Nested Logit with data-driven nests, that come from either our algorithm or that of \citet{benson2016relevance}.  We also introduce a \textbf{Point Estimate} baseline that directly uses empirical probabilities in the training set to define $\phi^\text{est}(\cdot,S)$, which has no risk of mis-specification but forgoes the generalization capability of choice models.
The results under different train/test splits are displayed in \Cref{fig:sfwork_rmse}, where we average over 10000 shuffled orders of the dataset.

\begin{figure}
\centering
\begin{minipage}[c]{0.66\textwidth}
\centering
\includegraphics[width=\linewidth]{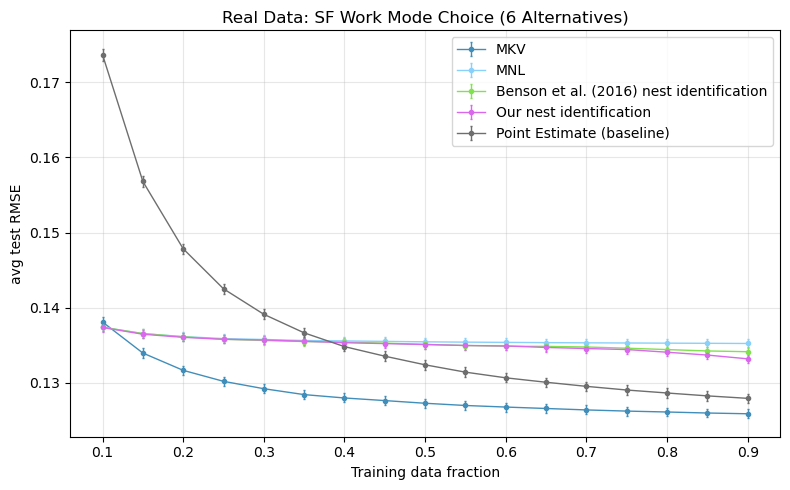}
\end{minipage}%
\hfill
\begin{minipage}[c]{0.33\textwidth}
\caption{
SFWork dataset: comparing models for out-of-sample choice prediction, with data-driven nests in Nested Logit, computed using our nest identification algorithm or that of \citet{benson2016relevance}}
\footnotesize

Estimated models: MNL, Nested Logit (NL), Markov Chain (MKV)
\label{fig:sfwork_rmse}
\end{minipage}
\end{figure}

\paragraph{Findings.}
Without the data coming from our experiment design, our nest identification algorithm does not have an advantage over that of \citet{benson2016relevance}, but can still match its performance thanks to the 5000 observations being concentrated on a small enough number of assortments.
The Nested Logit models are also no better than the simpler model of MNL.
All three of these models are 
much worse than Markov Chain at fitting the real-world data, suggesting that the ground truth is not Nested Logit, and/or Markov Chain is just generally a better model for choice prediction.

In the subsequent \namecref{sec:deployment}, we will see a more favorable comparison for Nested Logit. Regardless, we emphasize that the sole goal of choice modeling is not out-of-sample prediction, which would disproportionately favor high-parameter models like Markov Chain \citep{blanchet2016markov} or black-box machine learning structures \citep[e.g.][]{chen2022decision,aouad2025representing}.
To the contrary, one often wants to fit a Nested Logit model for the purpose of explaining how customers make choices, and in particular, which items are close substitutes, which is exactly what we do next for Dream11's management.

\section{Implementation at Dream11} \label{sec:deployment}

Following the experiment design described in \Cref{sec:ourExpDesign}, we set $b=2$ and formed $2\lceil\log_2 72\rceil=14$ experimental assortments for $72$ different contest options, where approximately half of these options were removed from each assortment. Dream11 ran an A/B test with these 14 assortments as treatment groups plus the full set as a control group on 70 million users, from May 20, 2025 to June 10, 2025.

\subsection{Out-of-sample Prediction.}
We make a plot like \Cref{fig:sfwork_rmse} in \Cref{ssec:sf}, with the following modifications. First, we show the soft RMSE \textit{divided by} that of the Point Estimate baseline, to mask absolute performance, also noting that in the definition~\eqref{eqn:rmsesoftMod} of soft RMSE, we let $\cS$ consist of the 14 experimental assortments and the control assortment.
Second, we do a single chronological train/test split for each of the 21 days, instead of averaging over shuffled orders of the dataset, which causes the results to be more volatile.
Finally, we also compare against an ex-ante partitioning of the contests into nests based on running $k$-means-clustering on the contest features, fixing $k=4$ which was usually the best number of clusters for prediction.
The results are displayed in \Cref{fig:dreamSports_rmse}. For further background on the Dream11 platform, see \Cref{sec:d11background}; for an interpretation of the identified nests, see \Cref{sec:interpretability}.

\begin{figure}
\centering
\begin{minipage}[c]{0.75\textwidth}
\centering
\includegraphics[width=\linewidth]{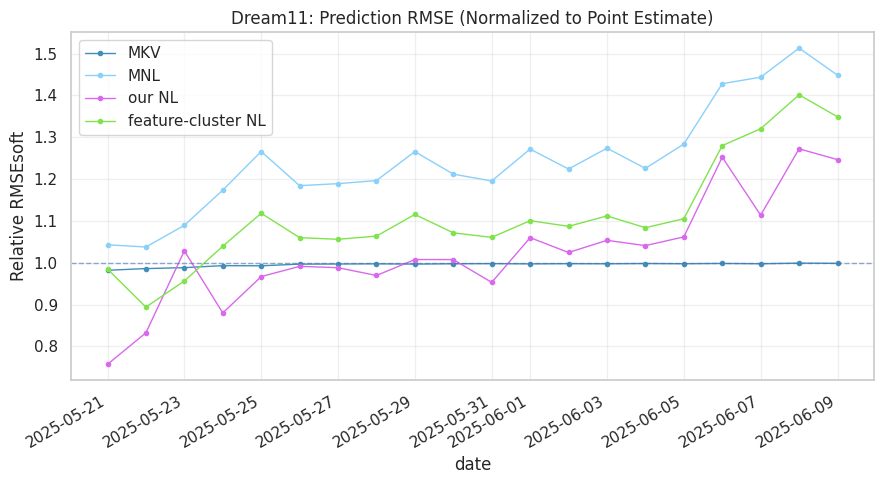}
\end{minipage}%
\hfill
\begin{minipage}[c]{0.24\textwidth}
\caption{
Dream11 data collected under our experiment design: comparing models for out-of-sample choice prediction, with different ways of identifying nests in Nested Logit}
\label{fig:dreamSports_rmse}
\end{minipage}
\end{figure}

\paragraph{Findings.} Data-driven nest identification via our algorithm broadly outperforms ex-ante clustering and achieves the lowest out-of-sample $\RMSEsoft$ during the first half of the horizon, before too much data is accumulated (in which case the point estimate is unbeatable).
Both our Nested Logit specification and the feature-cluster Nested Logit variant consistently outperform the simpler MNL model throughout the horizon (cf.\ \Cref{fig:sfwork_rmse}).
The Markov Chain choice model performs nowhere near as well as in \Cref{fig:sfwork_rmse}; its performance tracks very closely the point estimate baseline.

\subsection{Interpreting the Nests Identified} 
\label{sec:interpretability}

We estimate a Nested Logit choice model with data-driven nest identification to help Dream11's management understand which contests are close substitutes. 
The identified nests can often be justified ex post using either the original contest features or simple transformations of these features. 
The original features are Entry Fee, Prize Pool, \# of Contestants, and \# of Winners. 
We also consider two transformed features: the \textit{Winner Ratio}, defined as \# of Winners divided by \# of Contestants, and the \textit{Prize Ratio}, defined as the Prize Pool divided by the total entry fees collected from contestants. 
The latter measures the fraction of collected fees returned to contestants, so a higher Prize Ratio is more favorable to users.

Several examples illustrate that the nests identified by our algorithm are economically interpretable; the corresponding tables are reported in Appendix~\ref{app:nest-examples}. 
First, the algorithm identifies contests that are not especially close in the original feature space, but become nearly identical once Winner Ratio and Prize Ratio are considered. 
For example, the two contests in \Cref{tab:essentially_same} have very different Entry Fees, Prize Pools, and numbers of contestants, but both have Winner Ratio \(0.50\) and Prize Ratio \(0.83\). 
This suggests that users may perceive contests through normalized payoff and winning-probability features rather than only through raw contest attributes.

Second, the algorithm groups together several winner-take-all contests, as shown in \Cref{tab:small_size}. 
Ex ante, it may be surprising that contests with substantially different Entry Fees are close substitutes. 
However, this grouping is sensible ex post: winner-take-all contests are likely to attract power users, who may focus primarily on the high-upside prize structure and expected return rather than the absolute magnitude of the Entry Fee.

Third, the nest in \Cref{tab:irrational} suggests a more subtle behavioral pattern. 
It contains low-entry-fee contests with similar Winner Ratios, but with two different structures: some contests offer a higher Prize Ratio with fewer winners, while others offer a lower Prize Ratio with more winners. 
The latter contests are not necessarily better in terms of the probability of winning, yet they may appear more attractive because they advertise a larger number of winners. 
This suggests that users may respond to the salience of ``more winners'' even when the Winner Ratio is not higher.

Overall, these examples suggest that the nests identified by our algorithm, which also improve model fit, are meaningful ex post. 
They would have been difficult to specify ex ante because there are many possible ways to cluster contests in the original feature space, and many possible transformations one could define. 
We caveat that there may be other equally sensible and predictive nest structures not recovered by our algorithm. 
Nevertheless, our experimental design and nest identification pipeline produced nests that were sufficiently interpretable and useful for Dream11's management.
\section{Concluding Remarks and Future Directions} \label{sec:conc}

In this paper, we propose a simple combinatorial design for learning choice models over $n$ items that uses only $O(\log n)$ experimental assortments.
The small number of experiments significantly eases deployment in practice, e.g.\ for parallel experimentation across locations or switchback experiments over time.
In spite of this, our design also outperforms randomized assortments and other simple designs for data collection (see \Cref{sec:numericExpDesign}), sometimes even when they have substantially more experiments.
This is because our combinatorial arrangement tries to ensure enough "separations" between items are observed (see \Cref{sec:designDescr} and \Cref{table:introDesign}), to unveil richer substitution patterns.

That being said, our work leaves several questions for further investigation.  First, the theoretical justification for our experiment design is limited to Nested Logit, guaranteeing the correctness of nest identification (see \Cref{sec:theory}).
This does not justify its strong empirical performance for general choice models (see \Cref{sec:numericExpDesign}).
Second, our design treats the items as symmetric ex ante, whereas in practice, one often has domain knowledge on some products having larger market shares, being good candidates for lying in same nest, etc.
Designing experimental assortments that incorporate this domain knowledge remains an interesting and practical open question.

All in all, our paper brings to light the question of experiment design to learn choice substitution patterns, which had previously been mostly an afterthought (see the related work areas discussed in \Cref{sec:furtherRelated}).  The most salient contrast is with \citet{benson2016relevance}---their experiment design is based around their nest identification algorithm, whereas our nest identification algorithm is based around our experiment design.  We show that the data collection step can make a first-order difference in choice estimation performance, sometimes bigger than the difference in the choice model selected.
We came upon this question through our collaboration with Dream11, where we had the unique luxury of doing choice estimation with full control of the data collection process, as our experiment design was deployed across 70 million users on their platform.

\bibliographystyle{informs2014}
\bibliography{bibliography}

\appendix
\section{Proofs from \Cref{sec:theory}}

\subsection{Proof of Theorem~\ref{thm:mainResult}} \label{sec:mainResultPf}

We want to show that $E[i,j]=1$ if and only if $N(i)=N(j)$, for $i\neq j$.
It follows from \Cref{prop:outsideNew} that the deductions on lines~\ref{alg1:line5}, \ref{alg1:line7}, and \ref{alg1:noboost} of Algorithm~\ref{alg:exactWithOutside} are correct, and it follows by construction that the "One Hop Transitivity" operation is correct.
Because the final line sets $E[i,j]$ to 0 for any unfilled entries, we only need to prove that $E[i,j]$ gets set to 1 for all items $i\neq j$ with $N(i)=N(j)$, and that "Identify Missing Pairs" cannot incorrectly set any $E[i,j]$ to 1.
\begin{itemize}
\item We first show that if $N(i)=N(j)$ with $|N(i)|\ge3$, then $E[i,j]$ gets set to 1 either on line~\eqref{alg1:line7} or through $\text{OneHopTransitivity}$.  Indeed, the encodings of $i,j$ must satisfy $\sigma_\ell(i)\neq\sigma_\ell(j)$ for some position $\ell$.  Take a third item $k$ in the same nest as $i$ and $j$.
\begin{itemize}
\item If $\sigma_\ell(k)\neq\sigma_\ell(i)$, then for $S:=S_{\ell,-\sigma_\ell(i)}$, we have $j,k\in S$ and $i\notin S$, which implies that $E[j,k]$ would get set to 1 on line~\eqref{alg1:line7}.
\item If $\sigma_\ell(k)\neq\sigma_\ell(j)$, then we can symmetrically argue that $E[i,k]$ would get set to 1 on line~\eqref{alg1:line7}.
\end{itemize}
At least one of these two cases must hold (because $\sigma_\ell(i)\neq\sigma_\ell(j)$), which shows that $k$ gets connected to at least one of $i,j$ via line~\eqref{alg1:line7}.  We can symmetrically see that $i$ gets connected to at least one of $j,k$, and $j$ gets connected to at least one of $i,k$.  This implies that at least two of $\{i,j,k\}$ would get connected via line~\eqref{alg1:line7}, which means that the third is guaranteed to get connected via $\texttt{OneHopTransitivity}$.
\item Next, if $N(i)=N(j)=\{i,j\}$, then $E[i,j]$ cannot get incorrectly set to 0 before line~\eqref{alg1:uniqueNest}, and would get set to 1 via line~\eqref{alg1:uniqueNest}, because it is also guaranteed that neither $E[i,k]$ nor $E[j,k]$ can get (incorrectly) set to 1 for some $k\notin\{i,j\}$ before line~\eqref{alg1:uniqueNest}.
\item Finally, we must show that $E[i,j]$ cannot get set to 1 on line~\eqref{alg1:uniqueNest} if $N(i)\neq N(j)$.  If $|N(i)|\ge 3$ then $E[i,k]$ has already been set to 1 for all $k\in N(i)$ by the end of $\texttt{OneHopTransitivity}$, as we argued in the first bullet, and hence it is not possible for $(i,j)\in\texttt{IdentifyMissingPairs}$.  If $|N(i)|=1$, then there exists $S\in\mathcal{S}$ such that $i\in S,j\notin S$ (take a position $\ell$ where $\sigma_\ell(i)\neq\sigma_\ell(j)$, and consider $S=S_{\ell,-\sigma_\ell(j)}$), and we will have $\mathsf{BF}(i,S)=\mathsf{BF}(0,S)$ because $N(i)\subseteq S$, so $E[i,j]$ would get set to 0 on line~\eqref{alg1:noboost}.

In the last case where $N(i)=\{i,k\}$ for some $k\neq j$, we know $\sigma_\ell(k)\neq\sigma_\ell(j)$ for some position $\ell$, and either $\sigma_\ell(i)\neq\sigma_\ell(k)$ or $\sigma_\ell(i)\neq\sigma_\ell(j)$.
\begin{itemize}
\item If $\sigma_\ell(i)\neq\sigma_\ell(k)$, then for $S:=S_{\ell,-\sigma_\ell(k)}$, we have $i,j\in S$ and $k\notin S$.  Because $k\in N(i)$, we have $N(i)\nsubseteq S$, which guarantees that we will see different boost factors $\mathsf{BF}(i,S)\neq\mathsf{BF}(j,S)$ by \Cref{prop:outsideNew} II., and hence 
we will have set $E[i,j]=0$ already on line~\eqref{alg1:line5}.
\item If $\sigma_\ell(i)\neq\sigma_\ell(j)$, then for $S:=S_{\ell,-\sigma_\ell(j)}$, we have $i,k\in S$ and $j\notin S$.  Because $N(i)=\{i,k\}\subseteq S$, it is guaranteed that $\mathsf{BF}(i,S)=\mathsf{BF}(0,S)$, which means that either $E[i,j]$ would get set to 0 via line~\eqref{alg1:noboost} (where $j\notin S$).
\end{itemize}

In all cases, we have shown that $E[i,j]$ cannot get incorrectly set to 1 on line~\eqref{alg1:uniqueNest}, completing the proof of Theorem~\ref{thm:mainResult}.
\end{itemize}

\subsection{Proof of the $\Omega(\log n)$ Lower Bound for Nest Identification}
\label{sec:logLowerBoundPf}

\begin{theorem}
\label{thm:logLowerBound}
For any experiment design $\mathcal{S}$, whether adaptive or non-adaptive, that enables nest identification for $n$ items under the Nested Logit choice model and satisfies \Cref{ass:identify}, we have
\[
|\mathcal{S}| = \Omega(\log n).
\]
\end{theorem}

\begin{proof}
Let $\mathcal S=\{S_t\}_{t=1}^T$ be an (adaptive) experiment design, where each $S_t\subseteq[n]$ may depend on past observations. We prove $T=\Omega(\log n)$.

\paragraph{Step 1 (Every pair must be separated).}
Say $S$ \emph{separates} $(i,j)$ if $\mathbf 1\{i\in S\}\neq \mathbf 1\{j\in S\}$. Suppose some $i\neq j$ are never separated, i.e., $(i\in S_t)\Leftrightarrow (j\in S_t)$ for all $t$.
We construct two Nested Logit models with identical choice probabilities on every queried $S_t$ but different nest partitions.

\smallskip\noindent
\textbf{Model A:} $\{i,j\}$ is a two-item nest with $\lambda\in(0,1)$ and weights $v_i,v_j>0$; all other items are singleton nests with weights $\{v_k\}_{k\neq i,j}$.

\smallskip\noindent
\textbf{Model B:} all items are singleton nests; keep $v_k$ for $k\neq i,j$ and set
\[
v_i' := v_i (v_i+v_j)^{\lambda-1},\qquad
v_j' := v_j (v_i+v_j)^{\lambda-1},
\]
so that $v_i'+v_j'=(v_i+v_j)^\lambda$.

For any queried $S$, either $\{i,j\}\subseteq S$ or $\{i,j\}\cap S=\emptyset$. In the latter case, both models assign zero probability to choosing $i$ or $j$ and coincide on the remaining items. In the former case, under Model A,
\[
\phi_A(i,S)=\frac{(v_i+v_j)^\lambda}{1+\sum_{N\in\mathcal N} v_N(S)}\cdot \frac{v_i}{v_i+v_j}
=\frac{v_i (v_i+v_j)^{\lambda-1}}{1+\sum_{N\in\mathcal N} v_N(S)},
\]
while under Model B, $\phi_B(i,S)=\frac{v_i'}{1+\sum v_N(S)}$.
Moreover, the denominators match: when $\{i,j\}\subseteq S$, Model~A contributes
$v_{\{i,j\}}(S) = (v_i + v_j)^{\lambda}$
to $\sum_{N\in\mathcal N} v_N(S)$, while Model~B contributes
$v_i' + v_j' = (v_i + v_j)^{\lambda}$,
and all other nest weights are unchanged.
Hence $\phi_A(\cdot,S)=\phi_B(\cdot,S)$ for every queried $S_t$ and $i, j$, contradicting identifiability. Therefore, every identifying design must separate every pair.

\paragraph{Step 2 (At least $\lceil\log_2 n\rceil$ experiments).}
After $t$ experiments, let
\[
p_t(x):=\big(\mathbf 1\{x\in S_1\},\ldots,\mathbf 1\{x\in S_t\}\big)\in\{0,1\}^t.
\]
Items with the same pattern have not been separated. Let $M_t$ be a largest equivalence class under $p_t(\cdot)$, so $|M_0|=n$. The next experiment splits $M_t$ into $M_t\cap S_{t+1}$ and $M_t\setminus S_{t+1}$, hence
\[
|M_{t+1}|\ge \left\lceil \frac{|M_t|}{2}\right\rceil
\quad\Rightarrow\quad
|M_T|\ge \left\lceil \frac{n}{2^T}\right\rceil.
\]
If $T<\log_2 n$ then $|M_T|\ge2$, so some pair was never separated, contradicting Step 1. Thus $T\ge \lceil\log_2 n\rceil$ and $|\mathcal S|=\Omega(\log n)$.
\end{proof}

\subsection{Proof of Theorem~\ref{thm:sampleComplexity}}\label{sec:sampleComplexityPf}

$C$ is a constant that will be set at the end of the proof, but we assume throughout that $C\ge 2$.
For all $S\in\mathcal{S}\cup\{[n]\}$ and $i\in S \cup \{0\}$, by the multiplicative Chernoff bound, for any $\varepsilon\in(0,1)$,
\begin{align*}
P\left[\left|\frac{\hphi(i,S)}{\phi(i,S)}-1\right|>\varepsilon\right]
&\le2\exp(-\frac{m\phi(i,S)\varepsilon^2}{3})
\le2\exp(-\frac{m\phi(i,[n])\varepsilon^2}{3})
\le2\exp(-\frac{m\rho\varepsilon^2}{3}),
\end{align*}
where $\phi(i,S)\ge\phi(i,[n])$ since Nested Logit satisfies $\phi(i,S)\ge\phi(i,S')$ for all $i\in S\subseteq S'$.
Setting $\varepsilon=\Delta/C$, the assumption on $m$ ensures that
\begin{align} \label{eqn:concProperty1}
\left|\frac{\hphi(i,S)}{\phi(i,S)}-1\right|\le\frac{\Delta}{C}
\end{align}
with probability at least $1-\delta/K$.  Because $\Delta<1$ and $C\ge 2$, we can also derive from~\eqref{eqn:concProperty1} that
\begin{align} \label{eqn:concProperty1.5}
2\phi(i,S) \ge \hphi(i,S)\ge\phi(i,S)/2.
\end{align}

For all $S\in\mathcal{S}\cup\{[n]\}$ and $i\in S \cup \{0\}$, define $X(i,S):=m\hphi(i,S)$ to be the number of times $i$ was chosen from $S$ in the samples.
Now, for all $S\in\mathcal{S}\cup\{[n]\}$ and $\{i,j\}\subseteq S \cup \{0\}$, conditional on any value of $X(i,S)+X(j,S)$, the value of $X(i,S)$ is distributed as $\text{Binomial}(X(i,S)+X(j,S),q)$, where we have let $q=\frac{\phi(i,S)}{\phi(i,S)+\phi(j,S)}$.  By the multiplicative Chernoff bound, for any $\varepsilon\in(0,1)$,
\begin{align*}
P\left[\left|\frac{X(i,S)}{X(i,S)+X(j,S)}-q\right|>\varepsilon q\right]&\le2\exp(-\frac{(X(i,S)+X(j,S))q\varepsilon^2}{3}).
\end{align*}
Setting the RHS to $\delta/K$ and solving for $\varepsilon$, we get $\varepsilon=\sqrt{\frac{3\log(2K/\delta)}{(X(i,S)+X(j,S))q}}$, and hence
\begin{align} \label{eqn:concProperty2}
\left|\frac{X(i,S)}{X(i,S)+X(j,S)}-\frac{\phi(i,S)}{\phi(i,S)+\phi(j,S)}\right|\le\sqrt{\frac{3\log(2K/\delta)}{X(i,S)+X(j,S)}\cdot\frac{\phi(i,S)}{\phi(i,S)+\phi(j,S)}}
\end{align}
with probability at least $1-\delta/K$.

Union-bounding over at most $(|\mathcal{S}|+1)(n+1)+(|\mathcal{S}|+1)\binom{n+1}2=K$ "bad" events, we get that with probability at least $1-\delta$, \eqref{eqn:concProperty1}, \eqref{eqn:concProperty1.5}, and \eqref{eqn:concProperty2} all hold, for all $S\in\mathcal{S}\cup\{[n]\}$ and $i,j\in S \cup \{0\}$ such that $i\neq j$.
Assuming this, we now show for an arbitrary $S\in\mathcal{S}$ and $i,j\in S\cup\{0\}$ that
\begin{align} \label{eqn:sampCompCond}
|z(i\succ j,S)| \le 8\sqrt{3\log(2K/\delta)}
\end{align}
if and only if $\mathsf{BF}(i, S) = \mathsf{BF}(j, S)$.

\paragraph{Case 1: $\mathsf{BF}(i, S) = \mathsf{BF}(j, S)$.}
We must show~\eqref{eqn:sampCompCond} holds.
We know that $\frac{\phi(i,S)}{\phi(i,S)+\phi(j,S)}=\frac{\phi(i,[n])}{\phi(i,[n])+\phi(j,[n])}$, and we let $q$ denote this probability. By swapping indices $i,j$ as needed, we assume that $q\le1/2$. 

We first analyze the numerator in the definition of $z(i\succ j,S)$ (see~\eqref{eqn:zstat}).  We derive
\begin{align*}
\left|\frac{\hphi(i,S)}{\hphi(i,S)+\hphi(j,S)}-\frac{\hphi(i,[n])}{\hphi(i,[n])+\hphi(j,[n])}\right|
&\le \left|\frac{\hphi(i,S)}{\hphi(i,S)+\hphi(j,S)}-q\right|+\left|q-\frac{\hphi(i,[n])}{\hphi(i,[n])+\hphi(j,[n])}\right|
\\ &=\left|\frac{X(i,S)}{X(i,S)+X(j,S)}-q\right|+\left|q-\frac{X(i,[n])}{X(i,[n])+X(j,[n])}\right|
\\ &\le\sqrt{\frac{3q\log(2K/\delta)}{X(i,S)+X(j,S)}}+\sqrt{\frac{3q\log(2K/\delta)}{X(i,[n])+X(j,[n])}}
\\ &\le\sqrt{3q\log(2K/\delta)}\sqrt{\frac{2}{X(i,S)+X(j,S)}+\frac{2}{X(i,[n])+X(j,[n])}}
\end{align*}
where the first inequality is the triangle inequality, the second inequality applies~\eqref{eqn:concProperty2}, and the final inequality holds because $\sqrt{a}+\sqrt{b}\le\sqrt{2a+2b}$ for $a,b\ge0$.

For the denominator of $z(i\succ j,S)$, we can apply~\eqref{eqn:concProperty1.5} to see that
\begin{align*}
\frac{\hphi(i,S)+\hphi(i,[n])}{\hphi(i,S)+\hphi(j,S)+\hphi(i,[n])+\hphi(j,[n])}
&\ge \frac{(\phi(i,S)+\phi(i,[n]))/2}{2(\phi(i,S)+\phi(j,S)+\phi(i,[n])+\phi(j,[n]))}
= q/4.
\end{align*}
Similarly we have $\frac{\hphi(j,S)+\hphi(j,[n])}{\hphi(i,S)+\hphi(j,S)+\hphi(i,[n])+\hphi(j,[n])}\ge (1-q)/4\ge1/8$ (recalling that $q\le 1/2$).

Putting it together, we have shown that $|z(i\succ j,S)|$ is at most
\begin{align*}
&\sqrt{3q\log(2K/\delta)}\sqrt{\frac{2}{X(i,S)+X(j,S)}+\frac{2}{X(i,[n])+X(j,[n])}}\Bigg/\sqrt{\frac q{32}(\frac{1/m}{\hphi(i,S)+\hphi(j,S)}+\frac{1/m}{\hphi(i,[n])+\hphi(j,[n])})}
\end{align*}
which equals $8\sqrt{3\log(2K/\delta)}$, completing the proof that~\eqref{eqn:sampCompCond} holds.

\paragraph{Case 2: $\mathsf{BF}(i, S) \ne \mathsf{BF}(j, S)$.}
We know that 
$
\left|\frac{\phi(i,S)}{\phi(i,S)+\phi(j,S)}-\frac{\phi(i,[n])}{\phi(i,[n])+\phi(j,[n])}\right|\ge\Delta
$
by assumption.
We must show that~\eqref{eqn:sampCompCond} does not hold.

For the numerator of $z(i\succ j,S)$, we derive
\begin{align*}
\left|\frac{\hphi(i,S)}{\hphi(i,S)+\hphi(j,S)}-\frac{\phi(i,S)}{\phi(i,S)+\phi(j,S)}\right|
&\le\max\left\{\frac{1+\Delta/C}{1-\Delta/C}-1,1-\frac{1-\Delta/C}{1+\Delta/C}\right\}\frac{\phi(i,S)}{\phi(i,S)+\phi(j,S)}
\\ &\le\max\left\{\frac{2\Delta/C}{1-\Delta/C},\frac{2\Delta/C}{1+\Delta/C}\right\}
\\ &\le 4\Delta/C
\end{align*}
where the first inequality applies~\eqref{eqn:concProperty1} and the other inequalities are elementary (recalling that $\Delta<1$ and $C\ge 2$).
We can similarly derive
$
\left|\frac{\hphi(i,[n])}{\hphi(i,[n])+\hphi(j,[n])}-\frac{\phi(i,[n])}{\phi(i,[n])+\phi(j,[n])}\right|\le4\Delta/C
$
and hence by the triangle inequality,
$$
\left|\frac{\hphi(i,S)}{\hphi(i,S)+\hphi(j,S)}-\frac{\hphi(i,[n])}{\hphi(i,[n])+\hphi(j,[n])}\right|\ge(1-\frac8C)\Delta.
$$

Meanwhile, the denominator of $z(i\succ j,S)$ is at most
\begin{align*}
\sqrt{\frac{1/m}{\hphi(i,S)+\hphi(j,S)}+\frac{1/m}{\hphi(i,[n])+\hphi(j,[n])}}
&\le \sqrt{\frac{2/m}{\phi(i,S)+\phi(j,S)}+\frac{2/m}{\phi(i,[n])+\phi(j,[n])}}
\\ &\le \sqrt{\frac{2/m}{2\rho}+\frac{2/m}{2\rho}}=\sqrt{\frac{2}{m\rho}}
\end{align*}
where the first inequality applies~\eqref{eqn:concProperty1.5} and the second inequality holds because $\phi(i,S)\ge\phi(i,[n])\ge\rho$ for all $i\in S\subseteq [n]$.

Putting it together, we have shown that $|z(i\succ j,S)|$ is at least
\begin{align*}
(1-\frac8C)\Delta\sqrt{\frac{m\rho}2}
\ge (1-\frac8C)\Delta\sqrt{3C^2\frac{\log(2K/\delta)}{2\Delta^2}}
= (1-\frac8C)C\sqrt{3\frac{\log(2K/\delta)}{2}}
\end{align*}
where the inequality applies the lower bound on the number of samples $m$.
This is greater than $8\sqrt{3\log(2K/\delta)}$ for a sufficiently large constant $C$ (we need $C>8+8\sqrt{2}$), completing the proof that~\eqref{eqn:sampCompCond} does not hold and completing the proof of Theorem~\ref{thm:sampleComplexity}.

\subsection{Proof of Theorem~\ref{thm:idenWithoutOutside}} \label{sec:idenWithoutOutsidePf}

Like in Algorithm~\ref{alg:exactWithOutside}, we note that the symmetry of matrix $E$ is preserved throughout the algorithm, so we can argue based on unordered pairs or items.
First, by the final two bullets of \Cref{prop:noOutsideNew}, the deductions on lines~\eqref{alg2:line4} and~\eqref{alg2:pos_deduct1} of Algorithm~\ref{alg:exactWithoutOutside} are correct.

We now show that the update that occurs on either line~\eqref{alg2:neg_deduct} or~\eqref{alg2:pos_deduct2} is correct, except that line~\eqref{alg2:pos_deduct2} may be (incorrectly) entered in the case where $\minBF(S)$ consists of multiple singleton nests.  We show this inductively as we iterate over the experimental assortments $S\in\mathcal{S}$.
\begin{itemize}
\item First suppose that all items in $\minBF(S)$ are in the same nest (where this nest may or may not have items outside of $S$).  By the induction hypothesis, all preceding steps that could have set 0's are correct.  Therefore, it is not possible for $E[i,j]$ to be (incorrectly) set to 0 for some $i,j\in\minBF(S)$.  This ensures that we (correctly) enter line~\eqref{alg2:pos_deduct2}.
\item Otherwise, if $\minBF(S)$ intersects multiple nests, we know from \Cref{prop:noOutsideNew} that $\minBF(S)$ must consist of multiple nests in their entirety.  There is nothing to prove in the case where all of these nests are singletons, so suppose without loss that $N(i)=N(j)\neq N(k)$ for distinct items $i,j,k\in\minBF(S)$, where we know that $\minBF(S)$ consists of at least two nests. 
We claim that either $E[i,k]=0$ or $E[j,k]=0$, causing line~\eqref{alg2:neg_deduct} to be (correctly) entered.

To see this, note that $\sigma_\ell(i)\neq \sigma_\ell(j)$ for some position $\ell$, and $\sigma_\ell(k)$ must be distinct from at least one of them, where we assume by symmetry that $\sigma_\ell(j)\neq \sigma_\ell(k)$.  Then for $S:=S_{\ell,-\sigma_\ell(j)}$, we have $i,k\in S$ and $j\notin S$, and because $N(i)\nsubseteq S$, this guarantees that we will see different boost factors $\mathsf{BF}(i,S)\neq \mathsf{BF}(k,S)$ by \Cref{prop:noOutsideNew} II., ensuring that $E[i,k]$ would have been set to 0 on line~\eqref{alg2:line4}.
\end{itemize}
We have shown that all updates after completing those on lines~\eqref{alg2:neg_deduct} and~\eqref{alg2:pos_deduct2} are correct, with the exception that $E[i,j]$ may be set to 1 even if $|N(i)|=|N(j)|=1$.  This property is preserved after the updates on line~\eqref{alg2:transitivity}, noting that the $\texttt{OneHopTransitivity}$ operation may propagate more errors where $E[i,j]$ is incorrectly set to 1 when $|N(i)|=|N(j)|=1$.

To complete the proof, we show that after the $\texttt{IdentifyMissingPairs}$ operation on line~\eqref{alg2:uniqueNest}, for all $i,j\in[n]$ with $i\neq j$, we have $E[i,j]=1$ if $N(i)=N(j)$, and have $E[i,j]\neq 1$ if $N(i)\neq N(j)$ and either $|N(i)|>1$ or $|N(j)|>1$.
\begin{itemize}
\item We know that if $N(i)=N(j)$ with $|N(i)|\ge 3$, then $E[i,j]=1$ by the end of $\texttt{OneHopTransitivity}$.  This follows from the proof of Theorem~\ref{thm:mainResult}, because we have shown that we would always enter line~\eqref{alg2:pos_deduct2} if $\minBF(S)\subsetneq N$ for a single nest $N\in\mathcal{N}$, and hence whenever we would have set $E[i,j]\gets 1$ on line~\eqref{alg1:line7} in Algorithm~\ref{alg:exactWithOutside}, we would set $E[i,j]\gets 1$ on either line~\eqref{alg2:pos_deduct1} or line~\eqref{alg2:pos_deduct2} in Algorithm~\ref{alg:exactWithoutOutside}.
\item If $N(i)=N(j)=\{i,j\}$, then $E[i,j]$ cannot get incorrectly set to 0 before line~\eqref{alg2:uniqueNest}, and if $E[i,j]=\texttt{null}$ then line~\eqref{alg2:uniqueNest} would set it to 1, because it is also guaranteed that neither $E[i,k]$ nor $E[j,k]$ can get (incorrectly) set to 1 for some $k\notin\{i,j\}$ (we are not worried about incorrect 1's for singleton nests, because $|N(i)|=2$).
\item Finally, we must show that $E[i,j]$ cannot get set to 1 on line~\eqref{alg2:uniqueNest} if $N(i)\neq N(j)$ and either $|N(i)|>1$ or $|N(j)|>1$.
If $|N(i)|\ge 3$ then $E[i,k]$ has already been set to 1 for all $k\in N(i)$ by the end of $\texttt{OneHopTransitivity}$, as we argued in the first bullet, and hence it is not possible for $(i,j)\in\texttt{IdentifyMissingPairs}$.  This is also the case if $|N(j)|\ge 3$.  Therefore we can assume without loss that either $|N(i)|=2$ or $|N(j)|=2$ (or both); by symmetry we assume $N(i)=\{i,k\}$ for some $k\notin\{i,j\}$.

We know $\sigma_\ell(j)\neq\sigma_\ell(k)$ for some position $\ell$, and either $\sigma_\ell(i)\neq\sigma_\ell(k)$ or $\sigma_\ell(i)\neq\sigma_\ell(j)$.
\begin{itemize}
\item In the first case, for $S:=S_{\ell,-\sigma_\ell(k)}$, we have $i,j\in S$ and $k\notin S$.  Because $k\in N(i)$, we have $N(i)\nsubseteq S$, which guarantees that we will see different boost factors $\mathsf{BF}(i,S)\neq\mathsf{BF}(j,S)$ by \Cref{prop:noOutsideNew} II., and hence 
we will have set $E[i,j]=0$ already on line~\eqref{alg2:line4}.
\item In the second case where $\sigma_\ell(i)\neq\sigma_\ell(j)$, for $S:=S_{\ell,-\sigma_\ell(j)}$, we have $i,k\in S$ and $j\notin S$.  Because $N(i)=\{i,k\}\subseteq S$, it is guaranteed that $i,k\in\minBF(S)$, which means that either $E[i,j]$ would get set to 0 via line~\eqref{alg2:neg_deduct} (where $j\notin S$) or $E[i,k]$ would get set to 1 via line~\eqref{alg2:pos_deduct2} (where $i,k\in\minBF(S)$).
Either case guarantees that we cannot set $E[i,j]$ to 1 on line~\eqref{alg2:uniqueNest}.
\end{itemize}
\end{itemize}
This completes the proof of Theorem~\ref{thm:idenWithoutOutside}.
We note that the final matrix $E$ may not be a collection of disjoint cliques, in that there could be violations of transitivity where $E[i,j]=E[j,k]=1$ but $E[i,k]=0$, but this can only occur if \(|N(i)| = |N(j)| = |N(k)| = 1\).  In such cases, we can arbitrarily divide $\{i,j,k\}$ into nests, and the statement of Theorem~\ref{thm:idenWithoutOutside} would be satisfied.

\section{Recovering Remaining Parameters Given the Nest Partition}
\label{sec:param_recovery}

After identifying the nest partition \(\mathcal{N}\), we describe how to recover the remaining Nested Logit parameters. 
Throughout this section, the choice probabilities \(\phi(i,S)\) are observed exactly for all assortments \(S\in\mathcal{S}\cup\{[n]\}\) and all \(i\in S\). 
Consequently, all derived nest shares \(P(N\mid S)\) (and the log-ratio quantities defined below) are known exactly.

Recall the model in \eqref{eqn:nestProb}: each item \(i\in[n]\) has a preference weight \(v_i>0\), and each nest \(N\in\mathcal{N}\) has a dissimilarity parameter \(\lambda_N\in[0,1]\). 

We first identify \textbf{within-nest preference-weight ratios} from the control assortment. For any nest $N$ with $|N| \ge 2$, for \(i,j\in N\), equation \eqref{eqn:nestProb} implies
\[
\frac{\phi(i,[n])}{\phi(j,[n])}=\frac{v_i}{v_j}.
\]
Let \(i_N\) be a designated base item, and define normalized weights \(w_i:=v_i/v_{i_N}\) for all \(i\in N\) (so \(w_{i_N}=1\)). 
This reduces the item-level unknowns within each nest to a single scale parameter \(c_N:=v_{i_N}\), since \(v_k = c_N\,w_k\) for all \(k\in N\). 
If $|N|=1$, we set $\lambda_N=1$ and are left with a single unknown $c_N:=v_{i_N}(=v_N)$.

After recovering \(\{w_i\}_{i\in[n]}\) from \(\phi(\cdot,[n])\), the remaining unknowns are at most \(2|\mathcal{N}|\) nest-level parameters \(\{(c_N,\lambda_N):N\in\mathcal{N}\}\).

We will recover these parameters by solving a linear system, which first requires the following non-degeneracy assumption and structural lemma.

\begin{assumption}[Non-degeneracy in General Position]
\label{ass:two_by_two_nonsingular}
Fix two distinct nests \(N,N'\in\mathcal{N}\). Consider any two experimental
assortments \(S,S'\in\mathcal{S}\) satisfying the following three requirements:
\begin{enumerate}[(i)]
    \item all relevant intersections are nonempty:
    \[
    N\cap S,\quad N\cap S',\quad N'\cap S,\quad N'\cap S' \neq \emptyset;
    \]

    \item for each nest, at least one of the two assortments intersects it
    partially:
    \[
    \big(N\cap S \subsetneq N \ \text{or}\ N\cap S' \subsetneq N\big),
    \qquad
    \big(N'\cap S \subsetneq N' \ \text{or}\ N'\cap S' \subsetneq N'\big);
    \]

    \item the induced intersection pairs are distinct:
    \[
    (N\cap S,\;N'\cap S)\neq (N\cap S',\;N'\cap S').
    \]
    that is, at least one of the two intersections differs.
\end{enumerate}
Then the coefficient matrix built from the log intersection fractions is
nonsingular:
\[
\det
\begin{pmatrix}
\log\!\left(\dfrac{\sum_{i\in N\cap S} v_i}{\sum_{i\in N} v_i}\right) &
\log\!\left(\dfrac{\sum_{i\in N'\cap S} v_i}{\sum_{i\in N'} v_i}\right)\\[6pt]
\log\!\left(\dfrac{\sum_{i\in N\cap S'} v_i}{\sum_{i\in N} v_i}\right) &
\log\!\left(\dfrac{\sum_{i\in N'\cap S'} v_i}{\sum_{i\in N'} v_i}\right)
\end{pmatrix}
\ne 0.
\]
\end{assumption}

\begin{lemma}
\label{lem:two_S_distinct_pairs}
Fix two distinct nests \(N,N'\in\mathcal{N}\) with \(|N|\ge 2\) and
\(|N'|\ge 2\). Under the encoding-based experiment design in
Section~\ref{sec:ourExpDesign}, there exist two proper experimental
assortments \(S,S'\in\mathcal{S}\), with \(S\subsetneq [n]\) and
\(S'\subsetneq [n]\), such that the following conditions hold.

First, both assortments intersect both nests:
\[
N\cap S\neq\emptyset,\quad N\cap S'\neq\emptyset,\quad
N'\cap S\neq\emptyset,\quad N'\cap S'\neq\emptyset.
\]

Second, for each nest, at least one of the two assortments intersects it
partially:
\[
\big(N\cap S \subsetneq N \ \text{or}\ N\cap S' \subsetneq N\big),
\qquad
\big(N'\cap S \subsetneq N' \ \text{or}\ N'\cap S' \subsetneq N'\big).
\]

Third, the two assortments induce distinct intersection pairs across the two
nests:
\[
\big(N\cap S,\; N'\cap S\big)\neq \big(N\cap S',\; N'\cap S'\big),
\]
equivalently,
\[
N\cap S\neq N\cap S'
\quad\text{or}\quad
N'\cap S\neq N'\cap S'.
\]

Finally, the second induced pair is not the pair of full nests:
\[
\big(N\cap S',\; N'\cap S'\big)\neq (N,N').
\]
\end{lemma}

\begin{proof}
Fix two nests \(N:=N(i)\) and \(N':=N(j)\) with \(|N|\ge 2\) and \(|N'|\ge 2\). Choose two distinct
items \(i,i'\in N\) and two distinct items \(j,j'\in N'\).
We construct two experimental assortments \(S,S'\in\mathcal{S}\) such that
\(\big(N\cap S,\,N'\cap S\big)\neq \big(N\cap S',\,N'\cap S'\big)\).

\begin{itemize}
    \item \textbf{Case 1:} There exists \(\ell\) such that \(\sigma_{\ell}(i)\neq\sigma_{\ell}(i')\) and
    \(\sigma_{\ell}(j)\neq\sigma_{\ell}(j')\).
    \begin{itemize}
        \item \textbf{Subcase 1(a):} \(\sigma_{\ell}(i)=\sigma_{\ell}(j)\).

        Take \(S:=S_{\ell,-\sigma_\ell(i)}\). Then \(i\notin S\) and \(i'\in S\), and also
        \(j\notin S\) and \(j'\in S\). Hence \(i'\in N\cap S\) but \(i\notin N\cap S\), and
        \(j'\in N'\cap S\) but \(j\notin N'\cap S\).

        Take \(S':=S_{\ell,-\sigma_\ell(i')}\). Then \(i\in S'\) and \(i'\notin S'\). 
        Since \(\sigma_\ell(j)=\sigma_\ell(i)\neq\sigma_\ell(i')\), we have \(j\in S'\).
        Therefore \(\big(N\cap S,\,N'\cap S\big)\neq \big(N\cap S',\,N'\cap S'\big)\).
        Moreover, $N\cap S, N\cap S' \ne N$ and $N'\cap S, N'\cap S' \ne N'$.

        \item \textbf{Subcase 1(b):} \(\sigma_{\ell}(i)\neq\sigma_{\ell}(j)\).

        Take \(S:=S_{\ell,-\sigma_\ell(i)}\). Then \(i\notin S\) and \(i'\in S\), while
        \(j\in S\) since \(\sigma_\ell(j)\neq\sigma_\ell(i)\).

        Take \(S':=S_{\ell,-\sigma_\ell(j)}\). Then \(j\notin S'\) and \(j'\in S'\), while
        \(i\in S'\) since \(\sigma_\ell(i)\neq\sigma_\ell(j)\).

        In particular, \(N\cap S\) contains \(i'\) but not \(i\), whereas \(N\cap S'\) contains
        \(i\) but not \(i'\); hence \(\big(N\cap S,\,N'\cap S\big)\neq \big(N\cap S',\,N'\cap S'\big) \ne (N, N')\).
        Moreover, $N\cap S, N\cap S' \ne N$ and $N'\cap S, N'\cap S' \ne N'$.
    \end{itemize}

    \item \textbf{Case 2:} Otherwise, for every position \(\ell_i\) with \(\sigma_{\ell_i}(i)\neq\sigma_{\ell_i}(i')\),
    we have \(\sigma_{\ell_i}(j)=\sigma_{\ell_i}(j')\); and for every position \(\ell_j\) with
    \(\sigma_{\ell_j}(j)\neq\sigma_{\ell_j}(j')\), we have \(\sigma_{\ell_j}(i)=\sigma_{\ell_j}(i')\).

    \begin{itemize}
        \item Pick any \(\ell_i\) such that \(\sigma_{\ell_i}(i)\neq\sigma_{\ell_i}(i')\) and set
        \(S:=S_{\ell_i,-\sigma_{\ell_i}(i)}\). Then \(i\notin S\) and \(i'\in S\).
        Moreover, \(\sigma_{\ell_i}(j)=\sigma_{\ell_i}(j')\) implies that \(j\) and \(j'\) are either
        both included in \(S\) or both excluded from \(S\). In particular,
        \(N'\cap S\) is well-defined, and at least one of
        \(S_{\ell_i,-\sigma_{\ell_i}(i)}\) or \(S_{\ell_i,-\sigma_{\ell_i}(i')}\) yields \(N'\cap S\neq\emptyset\);
        without loss of generality, take the above choice so that \(N'\cap S\neq\emptyset\) and $N \cap S \ne  \emptyset, N$.

        \item Pick any \(\ell_j\) such that \(\sigma_{\ell_j}(j)\neq\sigma_{\ell_j}(j')\) and set
        \(S':=S_{\ell_j,-\sigma_{\ell_j}(j)}\). Then \(j\notin S'\) and \(j'\in S'\).
        Since \(\sigma_{\ell_j}(i)=\sigma_{\ell_j}(i')\), the items \(i\) and \(i'\) are either both
        included or both excluded from \(S'\); again, choose \(\ell_j\) (or swap \(j\) and \(j'\)) so that
        \(N\cap S'\neq\emptyset\) and $N' \cap S \ne  \emptyset, N'$.

        \item With these choices, \(N\cap S\) is a proper subset of \(N\) that contains \(i'\) but not \(i\),
        while \(N'\cap S'\) is a proper subset of \(N'\) that contains \(j'\) but not \(j\). Hence, at least one
        coordinate differs between \(\big(N\cap S,\,N'\cap S\big)\) and \(\big(N\cap S',\,N'\cap S'\big)\),
        so the induced pairs are distinct.
        Also, each of $N$ intersects $S$ or $S'$ partially once. 
    \end{itemize}
\end{itemize}

In all cases, the construction yields $S,S'\in\mathcal{S}$ such that each of $N$ and $N'$
is partially intersected by at least one of the two assortments, and moreover
$\big(N\cap S,\;N'\cap S\big)\neq \big(N\cap S',\;N'\cap S'\big)$,
as required.

\end{proof}

\textbf{Nest Dissimilarity Parameters.}
We now proceed with recovering the dissimilarity parameters, noting that
this argument applies with or without an outside option. 
With an outside option, we anchor on it and normalize $c_{N(0)}=v_0=1$; otherwise, we choose an arbitrary anchor nest $\bar N\in\mathcal N$ and set $c_{\bar N}=1$ if $\lambda_N >0$ and $v_N = 1$ if $\lambda_N = 0$.

For any assortment \(S\in\mathcal{S}\cup\{[n]\}\) and any nest \(N\in\mathcal{N}\), the nest share $P(N\mid S):=\sum_{i\in N\cap S}\phi(i,S)$ is observable.
Under normalization of  \(v_i=c_Nw_i\), we have
\begin{equation} \label{eq:nest_ratio}
    \frac{P(\bar N\mid S)}{P(N\mid S)}
= \begin{cases}
{\Big(\sum_{i\in \bar N\cap S} w_i\Big)^{\lambda_{\bar N}}} /
{\Big(c_N\sum_{i\in N\cap S} w_i\Big)^{\lambda_N}} \quad & \text{if $\lambda_N>0, \lambda_{\bar N} > 0$}
\\
{\Big(\sum_{i\in \bar N\cap S} w_i\Big)^{\lambda_{\bar N}}} /
{v_N} \quad & \text{if $\lambda_N = 0, \lambda_{\bar N} > 0$}
\\
{1} / 
{\Big(c_N\sum_{i\in N\cap S} w_i\Big)^{\lambda_N}} \quad & \text{if $\lambda_N > 0, \lambda_{\bar N} = 0$}
\\
{1}  / {v_N} \quad & \text{if $\lambda_N = 0, \lambda_{\bar N} = 0$}
\end{cases}
\end{equation}

Taking logs of \Cref{eq:nest_ratio} yields, for each \(T\in\{S,S'\}\),
\begin{align*}\label{eq:log_nest_ratio}
\log \frac{P(\bar N\mid T)}{P(N\mid T)}
= \lambda_{\bar N} A_T - \lambda_N B_T - s_N ,
\end{align*}
where
$$ y_T
:=
\frac{\log\!P(\bar N\mid T)}{P(N\mid T)},
\
A_T
:=
\log\!\sum_{i\in \bar N\cap T} w_i,
\
B_T
:=
\log\!\sum_{i\in  N\cap T} w_i,
\ 
s_N =
\begin{cases}
    \lambda_N \log c_N \quad &\text{if $\lambda_N > 0$}
    \\
    \log v_N \quad &\text{otherwise}
\end{cases}
$$
We construct an auxiliary variable $s_N$, such that $s_N =  \lambda_N \log c_N $ if $\lambda_N>0$ and $s_N = \log v_N$ otherwise. 
Under \Cref{lem:two_S_distinct_pairs} and \Cref{ass:two_by_two_nonsingular}, there exist two assortments $S$ and $S'$ such that the $3\times 3$ linear system below is nonsingular, and hence has a unique solution of $\lambda_{\bar N}, \lambda_N, s_N$. 

\[
\begin{pmatrix}
A_{[n]} & -\,B_{[n]} & -1\\[3pt]
A_S & -\,B_S & -1\\[3pt]
A_{S'} & -\,B_{S'} & -1
\end{pmatrix}
\begin{pmatrix}
\lambda_{\bar N}\\[2pt]
\lambda_N\\[2pt]
s_N
\end{pmatrix}
=
\begin{pmatrix}
y_{[n]}\\[2pt]
y_S\\[2pt]
y_{S'}
\end{pmatrix}
\quad
\]

We solve for $\lambda_{\bar N}, \lambda_N, s_N$. 
If $\lambda_N = 0$, $v_N = \exp s_N$.
If $\lambda_N > 0$, $c_N = \exp(\frac{s_N}{\lambda_N})$.

The linear system is constructed for each pair $(\bar{N}, N)$, so the anchor parameter $\lambda_{\bar{N}}$ is solved once per non-anchor nest $N \in \mathcal{N} \setminus \{\bar{N}\}$. 
Under the true model, all such systems yield the same $\lambda_{\bar{N}}$, since the equations are derived from exact choice probabilities of a single Nested Logit instance. 

\begin{remark}
The following computation verifies that \Cref{ass:two_by_two_nonsingular} (stated in terms of the original weights $v_i$) is equivalent to the nonsingularity of the $3\times 3$ coefficient matrix used above (stated in terms of the normalized weights $w_i$).
Under any normalization with $v_i = c_N w_i$, we have
\begin{align*}
 & \det \begin{pmatrix}
\log\!\left(\sum_{i\in N} w_i\right) & -\log\!\left(\sum_{i\in N'} w_i \right)  & -1\\[4pt]
\log\!\left(\sum_{i\in N\cap S} w_i\right) & -\log\!\left(\sum_{i\in N'\cap S} w_i\right)  & -1 \\[4pt]
\log\!\left(\sum_{i\in N\cap S'} w_i\right) & -\log\!\left(\sum_{i\in N'\cap S'} w_i\right)  & -1
\end{pmatrix}
\\
=& \det
\begin{pmatrix}
\log\!\left(\sum_{i\in N} v_i\right) - \log c_N & -\log\!\left(\sum_{i\in N'} v_i \right) + \log c_{N'} & -1 \\
\log\!\left(\sum_{i\in N\cap S} v_i\right) - \log c_N & -\log\!\left(\sum_{i\in N'\cap S} v_i\right)  + \log c_{N'} & -1 \\
\log\!\left(\sum_{i\in N\cap S'} v_i\right) - \log c_N & -\log\!\left(\sum_{i\in N'\cap S'} v_i\right) + \log c_{N'}  & -1
\end{pmatrix}
\\
=& 
\det
\begin{pmatrix}
\log\!\left(\sum_{i\in N} v_i\right) - \log c_N & -\log\!\left(\sum_{i\in N'} v_i \right) + \log c_{N'} & -1 \\
\log\!\left(\dfrac{\sum_{i\in N\cap S} v_i}{\sum_{i\in N} v_i}\right)  & \log\!\left(\dfrac{\sum_{i\in N'} v_i}{\sum_{i\in N'\cap S} v_i}\right)  & 0 \\
\log\!\left(\dfrac{\sum_{i\in N\cap S'} v_i}{\sum_{i\in N} v_i}\right) & \log\!\left(\dfrac{\sum_{i\in N'} v_i}{\sum_{i\in N'\cap S'} v_i}\right)   & 0
\end{pmatrix}
\\
=& 
\det
\begin{pmatrix}
\log\!\left(\dfrac{\sum_{i\in N\cap S} v_i}{\sum_{i\in N} v_i}\right) &
\log\!\left(\dfrac{\sum_{i\in N'\cap S} v_i}{\sum_{i\in N'} v_i}\right)\\[6pt]
\log\!\left(\dfrac{\sum_{i\in N\cap S'} v_i}{\sum_{i\in N} v_i}\right) &
\log\!\left(\dfrac{\sum_{i\in N'\cap S'} v_i}{\sum_{i\in N'} v_i}\right)
\end{pmatrix}
\end{align*}
\end{remark}

We emphasize that, in nested logit estimation, jointly maximizing the full-information likelihood (FIML) empirically outperforms the two-stage approach that first estimates preference weights and then estimates nest dissimilarity. 
Our contribution is not computational: we show that the full set of parameters is identifiable with only $O(\log n)$ experiments, which is new to our knowledge.

\section{$d$-level Nested Logit Model} \label{sec:dlevelLogit}

The two-level nested logit model introduced in Section~\ref{sec:theory}
is a special case of a more general \emph{$d$-level nested logit model}.
In a $d$-level nested logit model extension, the set of items is organized as the leaves of a rooted tree
with $d$ layers of \emph{nonoverlapping nests} (i.e., each layer partitions the items into disjoint nests).

Let $[n]$ be the set of items and let $S\subseteq[n]$ be the offered assortment.
Items are the leaves of a rooted tree $\mathcal T$ of depth $d$, and internal nodes are nests.
For any internal node (nest) $N$, let $\mathrm{Ch}(N)$ denote its set of children and let $\Leaf(N)$ denote the set of leaf items in the subtree rooted at $N$ (i.e., the descendant leaves of $N$). In particular, if $N$ is a leaf item, then $\Leaf(N)=\{N\}$.
For each item $i\in[n]$, let $r=a_0(i)\to a_1(i)\to\cdots\to a_d(i)=i$ be the unique root-to-leaf path from the root $r$ to leaf $i$.

In this section, we work with a tree structure rather than a partition as in \Cref{sec:theory} by recursively building the tree bottom up. 
Using our experimental assortments from Section~\ref{sec:ourExpDesign} and the exact market shares $\phi(\cdot\mid S)$, we explain how to do this for the 3-level Nested Logit model (treating the base model as having two levels of nesting) and explain why the same ideas extend to deeper trees.
The results apply both with and without an outside option.
Each item $i$ has a preference weight $v_i>0$, and each internal node $N$ has a nest dissimilarity parameter $\lambda_N\in(0,1]$ (rather than $\lambda_N\in[0,1]$ as in \Cref{sec:theory}).
For $\lambda_N = 0$, it is generally not possible to recover higher-level structure.

For a given assortment $S$, define the induced weight of a leaf $i$ as
\[
v_i(S)=
\begin{cases}
v_i, & i\in S,\\
0, & i\notin S,
\end{cases}
\qquad \forall i\in[n].
\]
For any internal node (nest) $N$, define its induced weight under $S$ recursively by
\[
v_{N}(S)
:=\left(\sum_{K\in \mathrm{Ch}(N)} v_K(S)\right)^{\lambda_N},
\qquad \forall \text{ internal nodes }N.
\]
For any child $K\in \mathrm{Ch}(N)$ (where $K$ may be an item/leaf or an internal node),
define the conditional probability of choosing $K$ given that the choice process is at $N$ as
\[
P(K\mid N,S)
:=
\frac{v_K(S)}{\sum_{L\in \mathrm{Ch}(N)} v_L(S)},
\quad \forall \text{ internal nodes }N,\ \forall K\in \mathrm{Ch}(N).
\]
Finally, the probability of choosing item $i\in S$ is the product of conditional probabilities
along the unique path from the root to $i$:
\[
\phi(i,S)
=
\prod_{h=0}^{d-1}P\!\left(a_{h+1}(i)\mid a_h(i),S\right),
\qquad \forall i\in S.
\]

We state a ``general position'' assumption, analogous to Assumption~\ref{ass:genPos} for the \(d\)-level tree, which rules out measure-zero degeneracies in which two distinct nests exhibit identical relative changes in their nest shares across an experimental assortment. 
\begin{assumption}[General Position D-level Tree]
\label{ass:genPos_dlevel}
For all $S\in\mathcal{S}$ and any two distinct nests $N\neq N'$ such that
$N\cap S\neq\emptyset$ and $N'\cap S\neq\emptyset$, and such that
at least one of the inclusions $N\cap S\subsetneq N$ or $N'\cap S\subsetneq N'$
holds, we have
$$\frac{P(N\mid S)}{P(N\mid [n])} \frac{\sum_{K\in \mathrm{Ch}(N)} v_K({[n]})}{\sum_{K\in \mathrm{Ch}(N)} v_K(S)}\neq\frac{P(N'\mid S)}{P(N'\mid [n])} \frac{\sum_{K\in \mathrm{Ch}(N')} v_K({[n]})}{\sum_{K\in \mathrm{Ch}(N')} v_K(S)}.$$
\end{assumption}
\Cref{ass:genPos_dlevel} is a generic condition.
Generally $P(N|S)$ is determined by the recursive aggregation up to the root, so this quantity is a non-trivially nonlinear function of all parameters in the tree.
Hence, the equality happens with measure zero. 

The key algebraic tool is a \textbf{sibling-ratio identity}, which holds without additional assumptions.
Assumption~\ref{ass:genPos_dlevel} will only be used later to distinguish different candidates when matching siblings across assortments.
Fix two sibling nodes \(N\) and \(N'\) in the augmented nesting tree, and let \(N^{\mathrm{par}}\) denote their common parent.
Recall that
\[
P(N\mid S)=\sum_{i\in \Leaf(N)} \phi(i,S).
\]
This quantity is observable. 
\begin{definition}[Sibling-ratio Identity]
For any assortment \(S\) such that \(N\cap S\neq\emptyset\) and \(N'\cap S\neq\emptyset\), if $N$ and $N'$ share the nest parent nest $ N^{\mathrm{par}}$, we have
\begin{align}\label{eq:sibling_ratio}
\frac{P(N\mid N^{\mathrm{par}},S)}{P(N'\mid N^{\mathrm{par}},S)}
&=
\frac{P(N\mid N^{\mathrm{par}})P(N^{\mathrm{par}}\mid S)}{P(N'\mid N^{\mathrm{par}})P(N^{\mathrm{par}}\mid S)}
\\
&=
\frac{P(N\mid N^{\mathrm{par}})}{P(N'\mid N^{\mathrm{par}})}
\\
&=
\frac{v_N(S)}{v_{N'}(S)}.
\end{align}
\end{definition}

\paragraph{Generic reduction (informal roadmap)}\label{roadmap:local_identification}
Suppose \Cref{ass:genPos_dlevel} holds and consider the experiment design $\cS$ in \Cref{sec:ourExpDesign}.
None of the inductive arguments depend on the outside option, hence the procedure below applies to both with and without the outside option. 
We treat each item as a leaf (terminal) node in the augmented nesting tree.
Suppose the substructure rooted at $N$ is identified (i.e., the partition of the subtree of $N$ into child components is known).
Then the following reductions apply.

\begin{enumerate}[(i)]
\item (\emph{One-scale representation of item weights under $N$.})
There exist known relative weights $\{w_i\}_{i\in \Leaf(N)}$ and an unknown scalar $c_N>0$ such that
\[
v_i = c_N\, w_i,\qquad \forall i\in \Leaf(N).
\]

\item (\emph{Known dissimilarity parameters for already identified child nests.})
For every (non-leaf) child nest $N^{\mathrm{ch}}$ of $N$ whose internal structure has been identified, the associated dissimilarity parameter $\lambda_{N^{\mathrm{ch}}}$ is known.

\item (\emph{Local identification via sibling ratios.})
Let $N^{\mathrm{par}}$ be the parent of $N$, and let $N'$ be any other child of $N^{\mathrm{par}}$ (so $N$ and $N'$ are siblings under $N^{\mathrm{par}}$).

The observable sibling-ratio statistic \Cref{eq:sibling_ratio} serves two purposes. 
First, it allows us to determine whether two candidate nodes $N$ and $N'$ are siblings (i.e., share the same parent $N^{\mathrm{par}}$). 
Second, once $N$ and $N'$ are confirmed to be siblings, the resulting equations identify the parameters needed to obtain (i)--(ii) for the parent node $N^{\mathrm{par}}$ (in particular, the scale $c_{N^{\mathrm{par}}}$ and the relevant dissimilarity parameters of $N^{\mathrm{par}}$'s children).

\end{enumerate}

We establish parts~(i)--(iii) rigorously for the leaf level (\Cref{sec:leafNode}) and the second-lowest level (\Cref{sec:second_lowest}).
For higher levels (\Cref{sec:higher_level}), we only outline the recursive argument, since the corresponding derivations are purely algebraic, substantially more notationally involved, and do not introduce additional conceptual insights.

\subsection{Item-level Nest Partitions}\label{sec:leafNode}
Let $\Npar(i)$ denote the parent node of item $i$.
In this section, we explain how to determine whether $\Npar(i) = \Npar(j)$ for items $i, j \in[n]$.

\begin{proposition}\label{prop:leaf_bf}
Suppose \Cref{ass:genPos_dlevel} holds, and take any experimental assortment
$S\in\cS$. For all $i, j\in S$:
\begin{enumerate}[I.]
\item If $j \in \Npar(i)$, then $\BF(i,S)=\BF(j,S)$.
\item If $j \notin \Npar(i)$ and $k \in \Npar(i)$, $k \notin S$, then $\BF(i,S)\ne \BF(j,S)$.
\end{enumerate}
\end{proposition}

\begin{proof}
Fix $S\in\cS$ and $i,j\in S$. Under the nested-logit factorization, for any $i\in S$,
\[
\phi(i,S)=P(\Npar(i)\mid S)\cdot P(i\mid \Npar(i),S)=
P(\Npar(i)\mid S)
\cdot
\frac{v_i}{\sum_{k\in \Npar(i)\cap S} v_k},
\qquad \forall i\in S,
\]
Then for any $i\in S$,
$$\BF(i,S)
=
\frac{
P(\Npar(i)\mid S)}{P(\Npar(i)\mid [n])} \cdot \frac{\sum_{k\in \Npar(i)} v_k}{\sum_{k\in \Npar(i)\cap S} v_k},
\qquad \forall i\in S.
$$

\emph{(I)} If $j\in \Npar(i)$, then $\Npar(j)=\Npar(i)=:N$. Therefore
\[
\BF(i,S)
=
\frac{
P(\Npar(i)\mid S)}{P(\Npar(i)\mid [n])} \cdot \frac{\sum_{k\in \Npar(i)} v_k}{\sum_{k\in \Npar(i)\cap S} v_k}
=
\BF(j,S)
\]

\emph{(II)} If $j\notin \Npar(i)$,  for some $k \in \Npar(i)$ and $k \notin S$, 
\begin{align*}
\BF(i,S)
&=
\frac{P(\Npar(j)\mid S)}{P(\Npar(i)\mid [n])}
\cdot
\frac{\sum_{k\in \Npar(i)} v_k}{\sum_{k\in \Npar(i)\cap S} v_k},
\\
\BF(j,S)
&=
\frac{P(\Npar(j)\mid S)}{P(\Npar(j)\mid [n])}
\cdot
\frac{\sum_{k\in \Npar(j)} v_k}{\sum_{k\in \Npar(j)\cap S} v_k}. 
\end{align*}

\Cref{ass:genPos_dlevel} ensures that whenever $\Npar(i) \ne \Npar(j)$,  $\BF(i,S)\neq \BF(j,S)$.
\end{proof}
This is a d-level nested logit analogue of \Cref{prop:outsideNew} and \Cref{prop:noOutsideNew}.

\begin{proposition}\label{prop:leaf_partition}
Recall that $\BF(i,S)$ denotes the boost factor defined in Section~\ref{sec:boostFactor}.
Let $i$ and $j$ be two items with $\Npar(i)\neq \Npar(j)$.
Assume that not both $|\Npar(i)|=1$ and $|\Npar(j)|=1$ hold.
Then at least one of the following holds:
\begin{itemize}
    \item There exists an assortment $S$ such that $\BF(i,S)\neq \BF(j,S)$.
    \item There exist an item $k$ and two distinct experimental assortments $S$ and $S'$ such that
    $\BF(j,S)\neq \BF(k,S)$ and $\BF(i,S')=\BF(k,S')$.
\end{itemize}
\end{proposition}

\begin{proof}
Let $i$ and $j$ belong to different leaf-level nests. If there exists an assortment $S$ for which
\[
\BF(i,S)\neq \BF(j,S),
\]
then the first condition is satisfied, and there is nothing left to prove.
Hence, assume throughout that
\[
\BF(i,S)=\BF(j,S)\qquad \text{for all assortments } S \text{ with } i,j\in S.
\]

By the assumption that $i$ and $j$ are not both in singleton leaf-level nests, at least one of
$|\Npar(i)|$ and $|\Npar(j)|$ exceeds $1$. Without loss of generality, assume $|\Npar(i)|>1$
and pick $k\in \Npar(i)$ with $k\neq i$.

Since $i \neq k$, there must exist $\ell$ such that $\sigma_\ell(i) \ne\sigma_\ell(k) $. 
If $\sigma_\ell(j) \ne \sigma_\ell(k)$, then $i, j \in S_{\ell, -\sigma_\ell(k)}$ while $k \notin S_{\ell, -\sigma_\ell(k)}$, resulting in $\BF(i, S_{\ell, -\sigma_\ell(k)}) \ne \BF(j, S_{\ell, -\sigma_\ell(k)})$. 
Hence, we must have $\sigma_\ell(j) = \sigma_\ell(k)$ whenever $\sigma_\ell(i) \ne\sigma_\ell(k) $. 
Fix such an $\ell$, we have $i \notin S_{\ell, -\sigma_\ell(i)}$ and $j, k \in S_{\ell, -\sigma_\ell(i)}$, hence by \Cref{prop:leaf_bf}
$\BF(j, S_{\ell, -\sigma_\ell(i)}) \ne \BF(k, S_{\ell, -\sigma_\ell(i)})$. 

By \Cref{prop:leaf_bf}, any experimental assortment $S$ that contains both $i$ and $k$ satisfies
$\BF(i,S)=\BF(k,S)$. 
Hence, if such an assortment exists, the second bullet point of
\Cref{prop:leaf_partition} is satisfied, implying $N^{\mathrm{par}}(i)\neq N^{\mathrm{par}}(j)$.

We prove by contradiction.
Suppose that for all $\ell$ and $d$, at most one of $\{i, k\}$ are in $S_{\ell, -d}$.
This is only possible when $b=2$. 
Indeed, if $b\ge 3$, then for any $\ell$ we can pick
$d\in\{0,1,\dots,b-1\}\setminus\{\sigma_\ell(i),\sigma_\ell(k)\}$, so that
$\sigma_\ell(i)\neq d$ and $\sigma_\ell(k)\neq d$, implying $i,k\in S_{\ell,-d}$,
a contradiction.
Hence $b=2$ and necessarily $\sigma_\ell(i)\neq\sigma_\ell(k)$ for all $\ell$.
However, we have already shown that $\sigma_\ell(j) = \sigma_\ell(k)$ whenever $\sigma_\ell(i) \ne\sigma_\ell(k)$, which indicates $\sigma_\ell(j) = \sigma_\ell(k)$ for all $\ell$, that is $j = k$. 
This is a contradiction.
Therefore, there exist $\ell$ and $d$ such that $i,k\in S_{\ell,-d}$. Take $S':=S_{\ell,-d}$.
Then $\BF(i,S')=\BF(k,S')$ by \Cref{prop:leaf_bf}(I), completing the proof.
\end{proof}

Hence, by Proposition~\ref{prop:leaf_partition}, if neither of the two conditions above holds for a pair $(i,j)$ (and we are not in the degenerate case where both $\Npar(i)$ and $\Npar(j)$ are singletons), then $i$ and $j$ must belong to the same leaf-level nest. 

We treat the corner case below via an equivalent normalization.
If a node has both direct item-children and an internal-node child, we group those direct item-children into an auxiliary leaf-level nest $\mathtt{tmp}$ with $\lambda_{\mathtt{tmp}}=1$.
This is without loss of generality, since introducing a $\lambda=1$ intermediate nest is equivalent to flattening and does not affect the nested-logit choice probabilities (for any assortment). To illustrate,

\begin{center}
\begin{tikzpicture}[
  grow=down,
  level 1/.style={sibling distance=20mm},
  level 2/.style={sibling distance=10mm},
  level 3/.style={sibling distance=5mm},
  edge from parent/.style={draw,-latex}
]
\node {$\mathtt{root}$}
  child {node {$1$}}
  child {node {$2$}}
  child {node {$\mathtt{subtree}$}};
\end{tikzpicture}
\end{center}

\medskip
would be treated as:

\begin{center}
\begin{tikzpicture}[
  grow=down,
  level 1/.style={sibling distance=40mm},
  level 2/.style={sibling distance=25mm},
  level 3/.style={sibling distance=15mm},
  edge from parent/.style={draw,-latex}
]
\node {$\mathtt{root}$}
  child {node {$\mathtt{tmp}$}
    child {node {$1$}}
    child {node {$2$}}
  }
  child {node {$\mathtt{subtree}$}};
\end{tikzpicture}
\end{center}

\subsection{Second-Lowest Level Nests} \label{sec:second_lowest}
In this section, we analyze relationships among the second-lowest-level nests, i.e., nests whose children's children are items. 
Two nodes are called \emph{siblings} if they share the same parent in the nesting tree. 
In particular, a nest node and an item node (leaf) can also be siblings when they have the same parent.

Similar to \Cref{sec:param_recovery}, assume the item-level partitions are identified, we represent the item utilities with one scale parameter $c_N$ per partition. 
Within any item-level partition $N$, $|N| \ge 2$, for offered items $i,k\in N$,
$\frac{P(i\mid N)}{P(k\mid N)}=\frac{v_i}{v_k}.$
Fix $i_N\in N$, set $w_i:=v_i/v_{i_N}$ (so $w_{i_N}=1$ and $w_i$ are known), and write $v_i=c_{N}w_i$ for all $i\in N$, where $c_{N}>0$ is the only unknown scale parameter.

When $N$ is an item-level nest with item set $N\subseteq\cI$, and the scale parameter $c_N$, we parameterize
\begin{equation}\label{eq:leaf_weight}
v_N(S) \;=\; v_N(c_N,S) \;:=\; \left(c_N \sum_{k\in S\cap N} w_k\right)^{\lambda_N},
\end{equation}
where $\{w_k\}_{k\in N}$ are normalized item weights within $N$, and $c_N>0$ is a nest-specific scale parameter.
For higher-level nests, $v_N(S)$ is defined analogously through the model's aggregation over its child nests; we keep the notation $v_N(S)$ generic and make its explicit form available when needed.

\begin{lemma}[Nest--Nest case]\label{lem:case_sibling_leaf}
Let $N(i)$ and $N(j)$ denote the two distinct item-level nests containing items $i$ and $j$, respectively, and assume $N(i)$ and $N(j)$ are siblings.
Then, for any offered assortment $S$, the statistic
\[
\log\!\left[\frac{P(N(i) \mid S)}{P( N(j)  \mid S)}\right],
\]
which is computable from observed choice probabilities given the current identified node labels,
can be expressed as a linear equation in the three unknown scalars
\[
\lambda_{N(i)},\qquad \lambda_{N(j)},\qquad \lambda_{N(j)}\log c_{N(j)}.
\]
Moreover, under Assumptions ~\ref{ass:two_by_two_nonsingular} and \ref{ass:genPos_dlevel}, with the experiment design $\cS$ described in \Cref{sec:ourExpDesign}, there exist at least three assortments such that the resulting three equations are linearly independent.
It follows that these three unknowns are identifiable.
\end{lemma}

\begin{proof}
Let $N^{\mathrm{par}}$ denote their common parent nest. 
Accordingly, we write $c_{N(i)},c_{N(j)}$ and $\lambda_{N(i)},\lambda_{N(j)}$ for the corresponding nest-scale and nesting parameters.

Normalize $c_{N(i)}=1$ and take logs of \Cref{eq:sibling_ratio} yields
\begin{equation}\label{eqn:dlevel_params}
\log\!\left[\frac{P( N(i)  \mid S)}{P(N(j)  \mid S)}\right]
=\log v_{N(i)}(1,S)
-\log v_{N(j)}(c_{N(j)},S).
\end{equation}

Furthermore, substituting \Cref{eq:leaf_weight} into \eqref{eqn:dlevel_params} gives
\begin{align}
\underbrace{\log\!\left[\frac{P( N(i)  \mid S)}{P(N(j)  \mid S)}\right]}_{\text{known}}
&=
\underbrace{\lambda_{N(i)}}_{\text{unknown 1}}
\underbrace{\log\!\left(\sum_{k\in S\cap N(i)} w_k\right)}_{\text{known}}
\nonumber
\quad-
\underbrace{\lambda_{N(j)}}_{\text{unknown 2}}
\underbrace{\log\!\left(\sum_{k\in S\cap N(j)} w_k\right)}_{\text{known}}
-
\underbrace{\lambda_{N(j)}\log c_{N(j)}}_{\text{unknown 3}} .
\label{eq:linear_leafleaf}
\end{align}

The statistic yields a linear equation in the three unknown scalars
$\lambda_{N(i)}$, $\lambda_{N(j)}$, and $\lambda_{N(j)}\log c_{N(j)}$.
By \Cref{lem:two_S_distinct_pairs} and \Cref{ass:two_by_two_nonsingular}, a similar proof to \Cref{sec:param_recovery} shows that we can find two assortments $S$ and $S'$ such that the linear system above has a unique solution.
\end{proof}
In \Cref{fig:3level_example}, the blue item-level nests and green item-level nests are an example of the nest--nest case. 

\begin{lemma}[Nest--item case]\label{lem:case_leaf_item}
Assume $N(i)$ is an item-level nest, $j$ is an item (treated as a terminal node in the augmented tree), $N(i)$ and item $j$ are siblings. 
Then, for any offered assortment $S$, the statistic
\[
\log\!\left[\frac{P(N(i) \mid S)}{\phi(j,S)}\right],
\]
which is computable from observed choice probabilities given the current identified node labels (with $N'$ the common parent of $N(i)$ and $j$),
can be expressed as a linear equation in the two unknown scalars
\[
\lambda_{N(i)},\qquad \log (w_j).
\]
where $w_j$ is the preference weight $v_j$ normalized to the same scale as the items in $N(i)$. 
under Assumptions ~\ref{ass:two_by_two_nonsingular} and \ref{ass:genPos_dlevel}, with the experiment design $\cS$ described in \Cref{sec:ourExpDesign}, there exist at least two assortments such that the resulting two equations are linearly independent.
It follows that these two unknowns are identifiable.
\end{lemma}

\begin{proof}
Let $N^{\mathrm{par}}$ denote their common parent nest. 
Accordingly, we write $c_{N(i)}$ and $\lambda_{N(i)}$ for the corresponding nest-scale and nesting parameters.
Substituting \Cref{eq:leaf_weight}  into \Cref{eq:sibling_ratio}, normalizing $c_{N(i)}=1$, and taking log yields a linear system in the two unknowns $\lambda_{N(i)}$ and $\log v_j$.
$$
\log\!\left[\frac{P(N(i) \mid S)}{\phi(j,S)}\right]
=
\underbrace{\lambda_{N(i)}}_{\text{1st unknown }} \ \underbrace{\log \; \Bigl(\sum_{k\in S\cap N(i)} w_k\Bigr)}_{\text{known constant}}  - \underbrace{\log(w_j)}_{\text{2nd unknown}}
$$

The control assortment $S=[n]$ provides one such value (with $N(i)\cap S=N(i)$).
Choose $i,i'\in N(i)$ with $\sigma_\ell(i)\neq \sigma_\ell(i')$ for some position $\ell$, and take an assortment of the form
$S:=S_{\ell,-\sigma_\ell(i)}$ or $S:=S_{\ell,-\sigma_\ell(i')}$ so that $j\in S$ and $N(i)\cap S\neq N(i)$.
This produces a second distinct value of $\log v_{N(i)}(1,S) < \log v_{N(i)}(1,[n])$.
The system 
$$
\begin{pmatrix}
    \log v_{N(i)}(1,[n]) & -1 
    \\
    \log v_{N(i)}(1,S) & -1
\end{pmatrix}
$$
is non-singular, as the determinant equals $ \log v_{N(i)}(1,S) - \log v_{N(i)}(1,[n]) \ne 0 $.
Therefore, the resulting $2\times2$
linear system has a unique solution for $\lambda_{N(i)}$ and $\log(w_j)$.
\end{proof}

In \Cref{fig:3level_example}, the pink item-level nests and purple item are an example of the nest--item case. 

Hence, to test whether two item-level nests share the same parent (i.e., are siblings),
we solve the associated \(3\times 3\) linear system. If the system is feasible, we declare the two nests
to be siblings; conversely, under genericity, non-siblings are not expected to satisfy these linear relations.
Similarly, to test whether an item-level nest and an item share the same parent,
we solve the associated \(2\times 2\) linear system; feasibility again serves as the sibling test.
Refer to \Cref{fig:3level_example} for an example of recursively recovering the structure from the bottom up.

\begin{figure}
    \centering
    \includegraphics[width=0.8\linewidth]{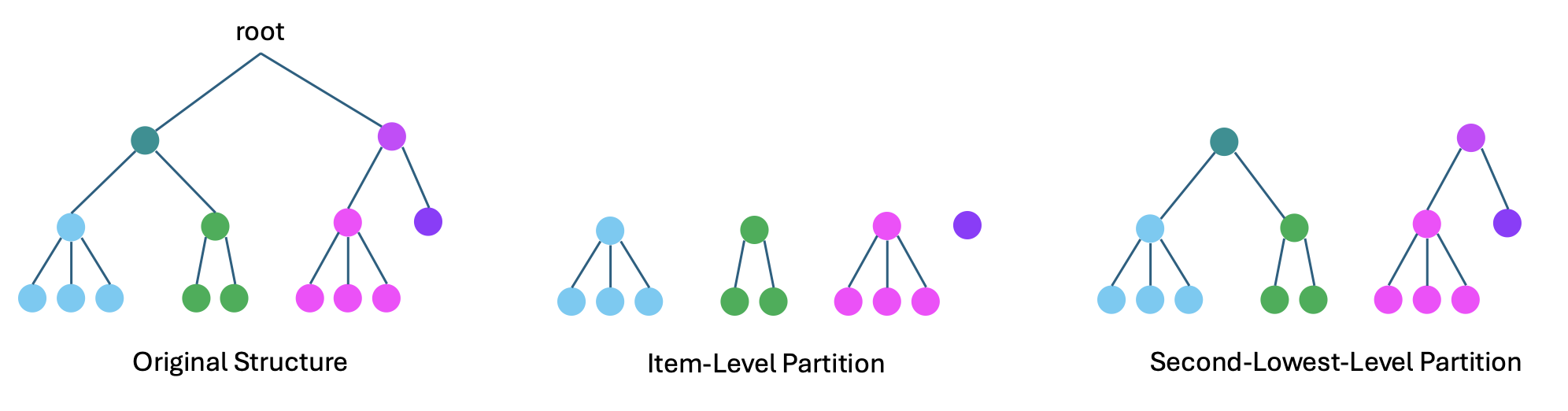}
    \caption{
Illustration of the tree-structure recovery process. Each leaf corresponds to an item.
\Cref{sec:leafNode} (Item-level Partition step) partitions items into item-level nests, and \Cref{sec:second_lowest} (Second-lowest level partition step) further partitions these item-level nests.
The blue and green nodes illustrate the nest--nest case, while the pink (nest) and purple (item) nodes illustrate the nest--item case.
}
\label{fig:3level_example}
\end{figure}

\subsection{Higher-level Nests}\label{sec:higher_level}
Further refinement of the nesting structure is possible by algebraically solving \Cref{eq:sibling_ratio}.
At higher levels, the induced weights $v_N(S)$ become nonlinear functions of the underlying scale parameters, so the resulting identification equations are no longer linear.
Under mild genericity (non-coincidence) and full-rank conditions on the experimental assortments, one can argue that the resulting system admits a unique solution by reducing it to a one-dimensional equation and invoking a standard uniqueness argument.
We omit the details, as they are lengthy and do not introduce additional conceptual insights beyond the bottom two levels.

\section{Supplement to \Cref{sec:numericExpDesign}} 

\subsection{Details of Mis-specified Instances (supplement to \Cref{ssec:misSpec})} \label{app:berbegliaInsts}

\citet[\S3.1.1]{berbeglia2022comparative} describes 1800 random instances in total, where there are 360 instances for each of 5 different values of $T$.
We consider the 360 instances for the 2 smallest and 2 largest values of $T$, resulting in 1440 total instances.
Our \Cref{fig:misSpecified} is analogous to \citet[Fig.~1]{berbeglia2022comparative}, noting that their Fig.~1 also only uses instances with the 2 smallest and 2 largest values of $T$.

\citet{berbeglia2022comparative} consider 9 choice models.
We estimate only 5 of them (EXP, MNL, LC, NL-given, and MKV), omitting MKVR, MKV2, MX, and RL.
To justify why, MKVR and MKV2 are simplified variants of the Markov Chain choice model that always underperformed the full model MKV.
Meanwhile, the MX and RL estimation procedures were really slow and often timed out in our settings with more items---in particular, MX requires simulation-based techniques for likelihood evaluation, while RL involves ranking structures that also led to more complex optimization procedures during estimation.  The performance of these methods was generally worse than MKV even when successfully estimated.  Therefore, we do not consider them in our paper, deferring to MKV as generally the best high-parameter model to estimate given sufficient data.

For the models we do estimate, we replicate the estimation methods implemented in \citet{berbeglia2022comparative}.  EXP is estimated following the Maximum Likelihood for Exponomial method described in \citet{alptekinoglu2016exp}.
MNL is estimated by standard Maximum Likelihood using convex optimization.
LC is estimated following the Conditional Gradient method of \citet{jagabathula2020conditional}.
NL is estimated by solving Maximum Likelihood using non-linear optimization packages (see \citet{berbeglia2022comparative} for details).
MKV is estimated following the Expectation Maximization method of \citet{csimcsek2018expectation}.
We note that their Nested Logit method arbitrarily partitions the items into two nests, without data-driven nest identification.

\subsection{Explanation of Confidence Intervals (supplement to \Cref{ssec:wellSpec})} \label{app:confidenceIntervals}

Letting $\mathcal{I}$ denote the set of instances being averaged together and $\RMSEsoft_I$ denote the soft RMSE between $\phi$ and $\phi^\text{est}$ on a particular instance $I\in\mathcal{I}$, the 95\% Confidence Interval is defined using
\begin{align*}
\overline{\mathrm{RMSE}}^{\mathrm{soft}}
= \frac{1}{|\mathcal{I}|}\sum_{I \in \mathcal{I}}\RMSEsoft_I ,
\qquad
\mathrm{STDEV}=\sqrt{\frac{1}{|\mathcal{I}|-1}\sum_{I \in \mathcal{I}}
\bigl(\RMSEsoft_I-\overline{\mathrm{RMSE}}^{\mathrm{soft}}\bigr)^2 } ,
\\ 95\% \text{ Confidence Interval}
= \overline{\mathrm{RMSE}}^{\mathrm{soft}}
\pm t_{\alpha/2,\,|\mathcal{I}|-1}\,\frac{\mathrm{STDEV}}{\sqrt{|\mathcal{I}|}},
\qquad \alpha = 0.05,
\end{align*}
where $t_{\alpha/2,\,|\mathcal{I}|-1}$ denotes the upper $\alpha/2$ quantile of the $t$-distribution with $|\mathcal{I}|-1$ degrees of freedom.

\subsection{Well-specified Estimation (supplement to \Cref{ssec:wellSpec})} \label{app:otherWellSpec}

We compare experiment designs in well-specified settings, plotting analogues of \Cref{fig:mkvMkv} (which was for the MKV choice model).
We consider the Exponomial model in \Cref{fig:expEXP}, and the MNL model in \Cref{fig:mnlMNL}.

\begin{figure}[!htb]
    \centering
    \includegraphics[width=0.7\linewidth]{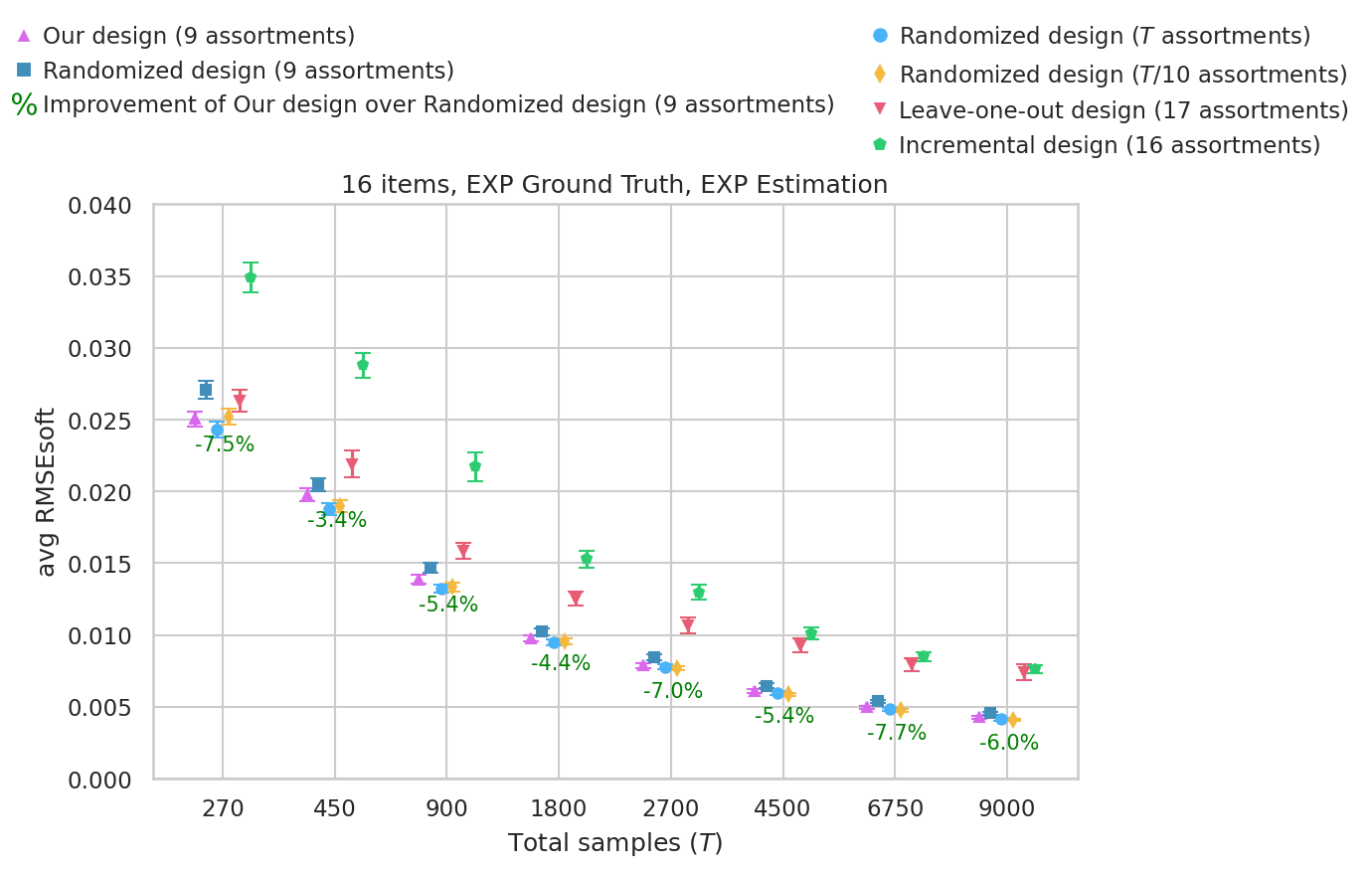}
    \caption{Average $\RMSEsoft$ over 500 Exponomial ground truths and Exponomial estimation}
    \label{fig:expEXP}
\end{figure}

\begin{figure}[!htb]
    \centering
    \includegraphics[width=0.7\linewidth]{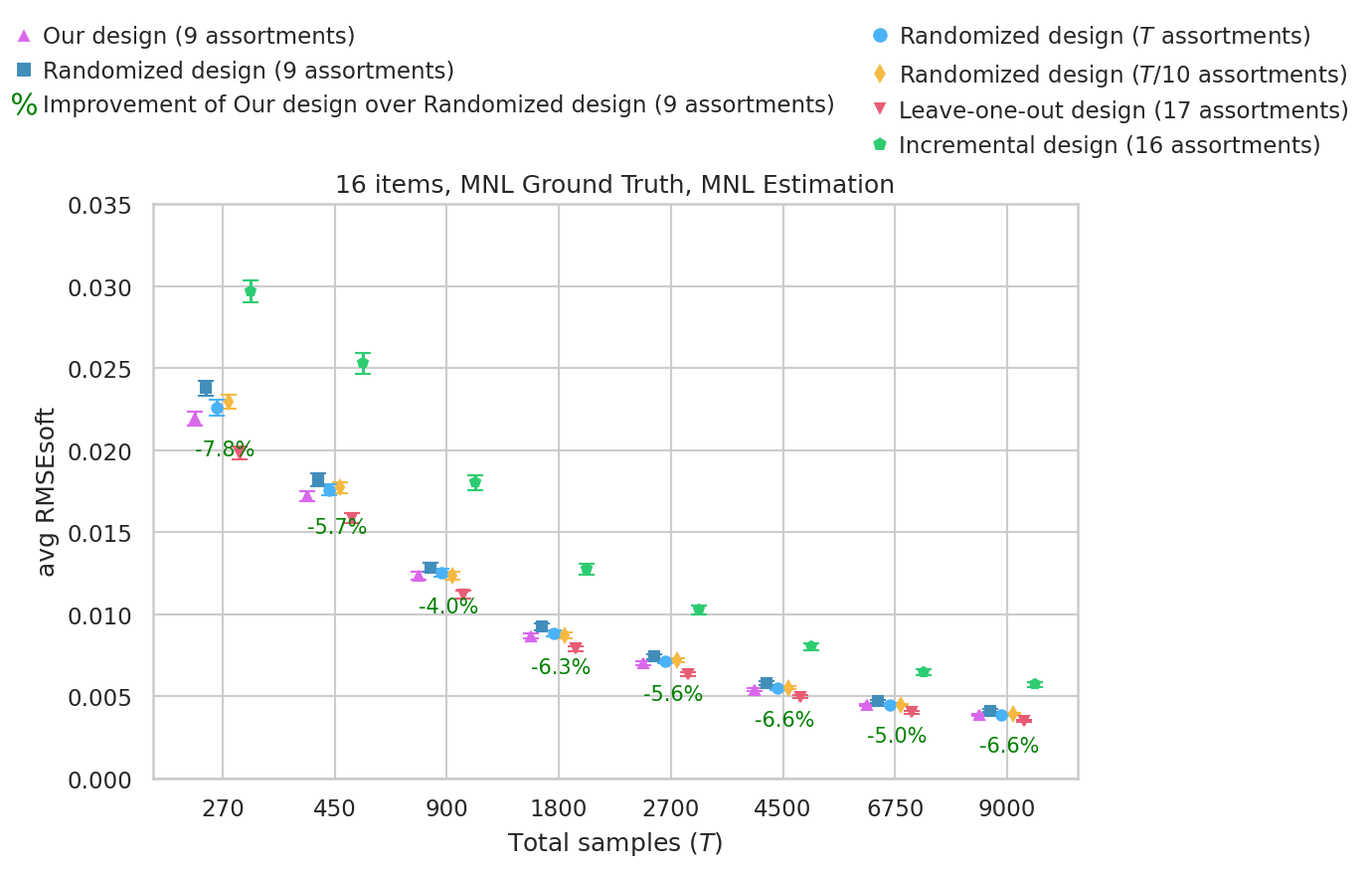}
    \caption{Average $\RMSEsoft$ over 500 MNL ground truths and MNL choice estimation}
    \label{fig:mnlMNL}
\end{figure}

\paragraph{Findings.}
Our experiment design consistently outperforms randomized design with the same number (9) of experimental assortments, by a significant margin that exceeds the margin in the mis-specified setting (\Cref{ssec:misSpec}).
Our experiment design is on par with a randomized design with far more experimental assortments, even as the data size gets larger, so it compares more favorably than in the mis-specified setting but less favorably than in the well-specified Markov Chain setting (\Cref{ssec:wellSpec}).
The Leave-one-out design performs well for MNL (\Cref{fig:mnlMNL}), but poorly in all of the other mis-specified (\Cref{fig:misSpecified}) and well-specified (\Cref{fig:mkvMkv,fig:expEXP}) settings.
This highlights the robustness of our experiment design compared to the other non-random designs like Leave-one-out.

\subsection{Asymptotic Markov Chain Estimation (supplement to \Cref{ssec:wellSpec})} \label{app:mkvAsymptotic}

\Cref{fig:mkvMkv} shows that the Leave-one-out design performs poorly for well-specified MKV estimation with $n=16$ items and $T$ up to $9000$, despite being proposed by \citet{blanchet2016markov} specifically for Markov Chain choice estimation.
Since \citet{blanchet2016markov} proved the asymptotic optimality of Leave-one-out for MKV, we investigate whether it catches up given substantially more data.

To make the asymptotic regime more accessible, we reduce to $n=8$ items, so that the Leave-one-out design has only 9 assortments (vs.\ 17 for $n=16$) and the MKV model has fewer parameters to estimate.
We extend the data range to $T$ up to $140{,}000$ (a 16$\times$ increase over \Cref{fig:mkvMkv}) and generate $500$ random MKV ground truths. 
The results are displayed in \Cref{fig:8mkvMKV}.

\begin{figure}
    \centering
    \includegraphics[width=0.7\linewidth]{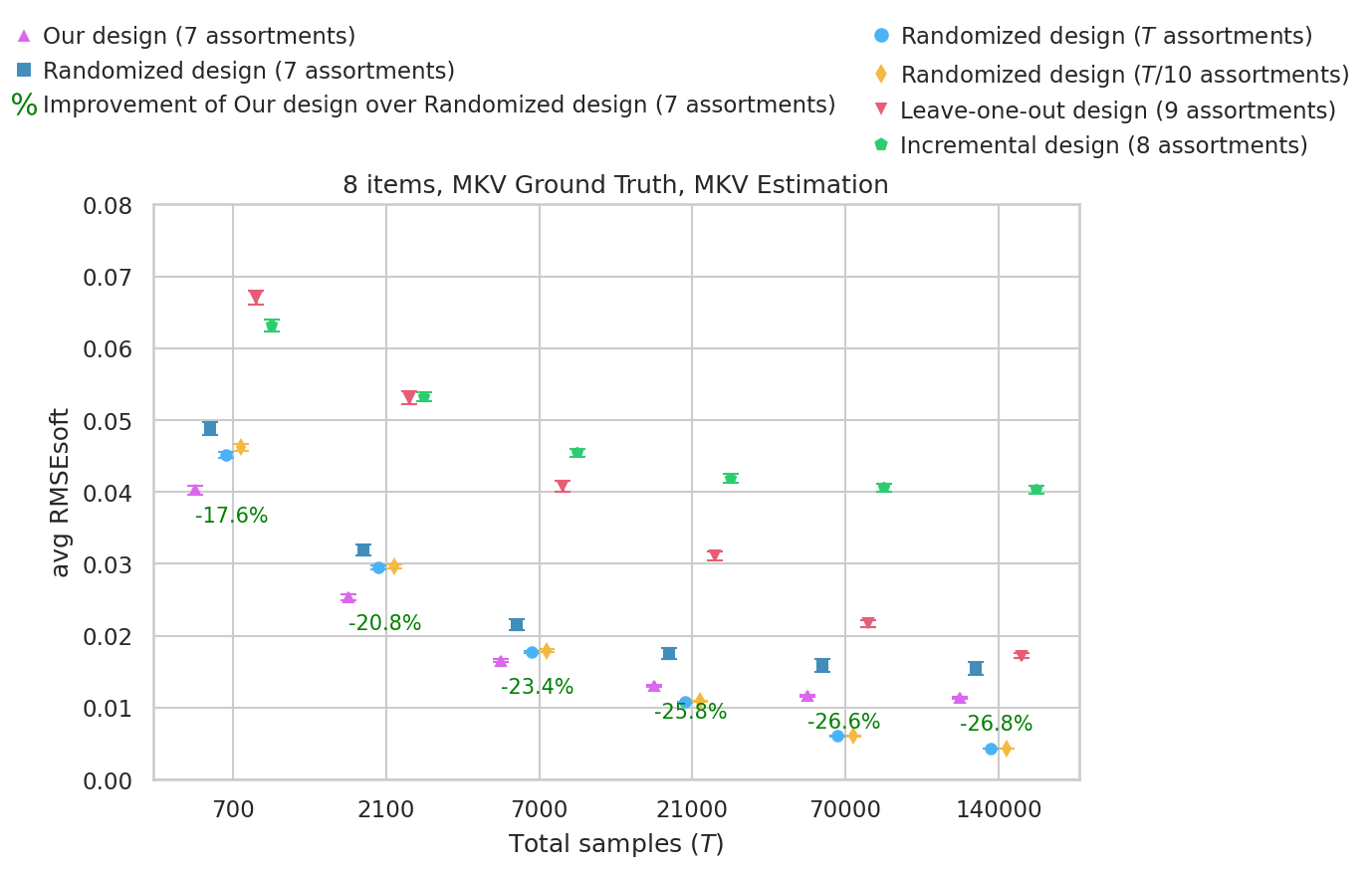}
    \caption{Average $\RMSEsoft$ over 500 MKV ground truths with $n=8$ items under MKV estimation, extending the data range to $T=140{,}000$}
    \label{fig:8mkvMKV}
\end{figure}

\paragraph{Findings.}
The Leave-one-out design improves substantially from its poor starting point as $T$ grows, but remains significantly worse than our design even at $T=140{,}000$---despite the favorable smaller $n=8$ setting.
This demonstrates that the asymptotic optimality of Leave-one-out for MKV does not translate to practical performance, even at data sizes 16$\times$ larger than those in \Cref{fig:mkvMkv} and with fewer items.
However, the randomized designs with more experimental assortments do outperform ours asymptotically, because our $O(\log n)$ experimental assortments are not quite enough for identifying the $n^2$ parameters in a Markov Chain choice model.

\section{Supplement to \Cref{sec:numericNestId}}

\subsection{Modifications to Nest Identification Algorithms} \label{app:algs3and4}

The sample-complexity bound in \Cref{thm:sampleComplexity} assumes error-free inference of all relations \(E[i,j]\), which in practice would require prohibitively large sample sizes.
This raises a natural question: with limited data, often far fewer samples than needed for perfect pairwise identification, can we still recover the nest structure?

Intuitively, yes.
Occasional errors in \(E[i,j]\) do not necessarily preclude accurate recovery, because the non-overlapping grouping constraint of Nested Logit provides strong structural signals.
As a simple illustration, suppose items \(i, j, k, \ldots\) lie in the same nest.
A single false negative can therefore invalidate the sufficient condition used in \Cref{thm:sampleComplexity} and eliminate its exact-recovery guarantee. 
However, such a localized error does not erase the global signal: if $i$ and $j$ each exhibit strong within-group relations with many of the same items, then the overall pattern can still support placing $i$ and $j$ in the same nest.

Motivated by this observation, we propose a \emph{robust} identification strategy that operates under limited data and tolerates small inference errors, implemented via simple post-processing from the inferred relations \(E[i,j]\) to the nest structure.
Our goal is to modify \Cref{alg:exactWithOutside} and \Cref{alg:exactWithoutOutside} to remain effective under realistic noise.

\paragraph{Construction of $E$ under noisy data.} \label{sec:noisyDeduction}

The finite-sample guarantee in \Cref{sec:finiteSample} assumes error-free statistical tests, under which inconsistencies do not arise.
Under limited samples, however, statistical test errors can induce systematic inconsistencies across pairwise relations.
We first define the standard statistical tests that will be used by our algorithms to handle noise.
Following \Cref{sec:finiteSample}, define
\begin{align}
\texttt{p-val}[\BF(i,S)=\BF(j,S)] &:=2(1-\Phi(|z(i\succ j,S)|)); \label{eqn:statTest1}
\\ \texttt{p-val}[\BF(i,S)\le\BF(0,S)] &:=1-\Phi(z(i\succ 0,S)) \label{eqn:statTest2}
\end{align}
which follow the standard definition of $p$-values in statistics.
A small value of $\texttt{p-val}[\BF(i,S)=\BF(j,S)]$ falling below some threshold $\alpha$ (e.g.,~$\alpha=0.05$) rejects the null hypothesis that $\BF(i,S)=\BF(j,S)$, suggesting that $\BF(i,S)\neq\BF(j,S)$.
A small value of $\texttt{p-val}[\BF(i,S)\le\BF(0,S)]$ falling below some threshold $\alpha$ rejects the null hypothesis that $\BF(i,S)\le\BF(0,S)$, suggesting that $\BF(i,S)>\BF(0,S)$.

We now present some examples below illustrating that pairwise decisions based on hypothesis tests need not be transitive and can vary across assortments under finite samples.
\begin{itemize}
\item Within-assortment Inconsistency:
Suppose items $i,j,k$ belong to the same nest.
In an experiment $S \in \mathcal{S}$, it may occur that
\[
\texttt{p-val}[\BF(i,S)=\BF(j,S)] > \alpha, \quad 
\texttt{p-val}[\BF(j,S)=\BF(k,S)] > \alpha, \quad 
\texttt{p-val}[\BF(i,S)=\BF(k,S)] \le \alpha .
\]
That is, at significance level $\alpha$, we do not reject $\BF(i,S)=\BF(j,S)$ and we do not reject $\BF(j,S)=\BF(k,S)$, but we do reject $\BF(i,S)=\BF(k,S)$.
Applying the pairwise decision rule within $S$ therefore yields conflicting classifications that violate transitivity.
\item Across-assortment Inconsistency:
Consider two experiments $S$ and $S'$.
In experiment $S$, it may occur that
\[
\texttt{p-val}[\BF(i,S)\le \BF(0,S)] \le \alpha, \quad
\texttt{p-val}[\BF(j,S)\le \BF(0,S)] \le \alpha, \quad
\texttt{p-val}[\BF(i,S)=\BF(j,S)] > \alpha,
\]
That is, at significance level $\alpha$, we reject the null hypotheses $\BF(i,S)\le \BF(0,S)$ and $\BF(j,S)\le \BF(0,S)$, but we do not reject the null $\BF(i,S)=\BF(j,S)$.
Under the decision rule in \Cref{prop:outsideNew}, this pattern of test outcomes would classify $i$ and $j$ as belonging to the same nest.

Yet in another experiment $S'$, we may obtain
\[
\texttt{p-val}[\BF(i,S')=\BF(j,S')] \le \alpha ,
\]
so we reject $\BF(i,S')=\BF(j,S')$ at level $\alpha$.
Under the same decision rule in \Cref{prop:outsideNew}, this would instead classify $i$ and $j$ as belonging to different nests.
\end{itemize}
In sum, under finite samples, different assortments can yield conflicting pairwise classifications for the same pair $(i,j)$, motivating the need to consolidate evidence across assortments when constructing $E[i,j]$.

Analogous to Algorithm~\ref{alg:exactWithOutside} and Algorithm~\ref{alg:exactWithoutOutside}, we initialize an empty edge table \(E\), where \(E[i,j]\) can now be fractional-valued and is an edge confidence score aggregated across assortments.
We then iteratively update \(E\) using observations from each experimental assortment \(S \in \mathcal{S}\), together with the statistical tests in~\eqref{eqn:statTest1}--\eqref{eqn:statTest2}.
We take the $\min$ across assortments $S\in\cS$ when updating the table $E$, trying to stay conservative in assigning probabilities that two items are in the same nest.

\paragraph{Recover nest structure from $E$.} \label{sec:contraHandle}

After obtaining the updated edge matrix $E$, we infer the nest structure.
In the noiseless case (as in \Cref{alg:exactWithOutside}), this is immediate: the graph induced by $E$ decomposes into disjoint cliques, each corresponding to a nest.
With noise, however, estimation errors may introduce erroneous edges and missing edges, so the graph induced by $E$ may no longer decompose cleanly into disjoint cliques.
We therefore use a more robust procedure to recover the underlying partition.

Let $G=(V,\mathcal{E})$ be a weighted graph with vertex set $V=[n]$, where the edge weight between $i$ and $j$ is $w_{ij}=E_{i,j}$.
As the sample size grows, we expect $E_{ij}$ to be stochastically larger for pairs within the same nest than for pairs across different nests.
Consequently, items within the same nest tend to form dense (high-weight) subgraphs, whereas items from different nests exhibit sparse or low-weight connectivity.

We cast nest recovery as a graph-based clustering problem (community detection). 
Given a symmetric matrix
$E\in[0,1]^{n\times n}$ with $E_{ii}=1$, we partition $[n]$ into groups whose within-group connectivity is
stronger than their between-group connectivity.
We adopt the Walktrap algorithm \citep{pons2005computingcommunitieslargenetworks}, which exploits short
random walks to capture local connectivity patterns and is well suited to recovering dense subgraphs under
noisy edge weights. We use the CDlib implementation \citep{rossetti2019cdlib} with walk length $t=4$, a
standard choice. We do not pre-specify the number of communities; instead, we select the partition that
maximizes modularity.

\Cref{fig:matrixE} illustrates the evolution of $E$ as the sample size increases in an example with $n=16$ items and $5$ nests.
Each plot visualizes the entries of $E$ using a grayscale colormap, where darker colors indicate smaller values and lighter colors indicate larger values (from $0$ to $1$). As the sample size increases, the entries of $E$ become more concentrated near $0$ and $1$, so the plots appear increasingly close to black-and-white.
Bright blue boxes denote the ground-truth nests, while bright pink boxes indicate the nests recovered by our procedure (no pink boxes indicates failure).

\begin{figure}
    \centering
    \includegraphics[width=0.8\linewidth]{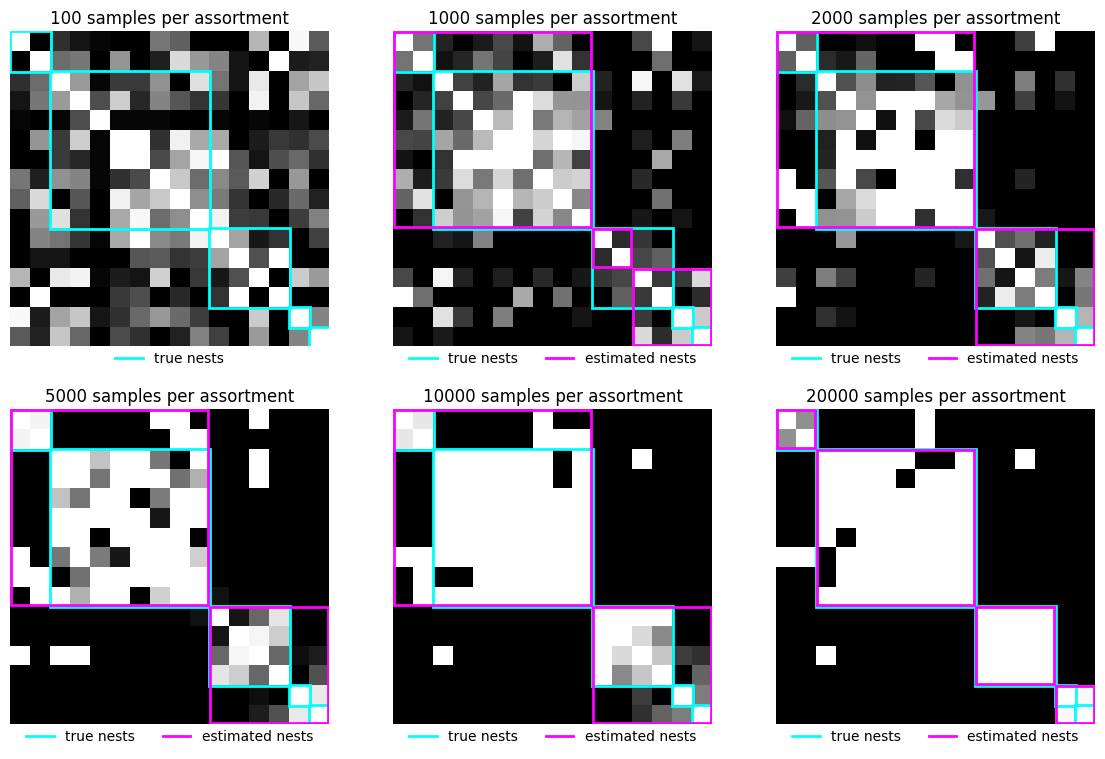}
    \caption{Evolution of $E$ as data increases.  For illustration, the sample sizes shown are selected so that the recovered nests appear along the diagonal; in general, the ordering of items is arbitrary and need not align with the ground-truth ordering.}
    \label{fig:matrixE}
\end{figure}

With $1{,}000$ samples per assortment, the block structure becomes visible, though the recovered nests remain inaccurate.
Starting around $20{,}000$ samples per assortment, the procedure recovers the correct nest structure even though $E$ has not yet become nearly binary.

\paragraph{Nest identification algorithms with statistical tests and community detection.} \label{sec:noisyAlgo}

We now present our two algorithms for nest identification under noisy data, corresponding to settings with and without an outside option.
Indeed, \Cref{alg:noisyNestWithOutside,alg:noisyNestWithoutOutside} correspond to the earlier \Cref{alg:exactWithOutside,alg:exactWithoutOutside}
Both algorithms reuse the terminology and update primitives introduced in \Cref{sec:theory}, and implement the edge aggregation and graph-based recovery described above. 
The updates in \Cref{alg:noisyNestWithoutOutside} are simpler, reflecting the more limited set of statistical tests available when the outside option is absent.

\begin{algorithm*}[!h]
\caption{\textsc{Inexact Nest Identification with Outside Option}}
\label{alg:noisyNestWithOutside}
{\setlength{\baselineskip}{0.8\baselineskip}
\begin{algorithmic}[1]
\State Initialize adjacency matrix $E[i, j] \gets 2$ for all $i, j \in [n]$ \Comment{{2 means $\texttt{null}$}}

\For{$S \in \mathcal{S}$}
\For{$i,j\in S$}
\If{$\texttt{p-val}[\BF(i, S) = \BF(j, S)] \le \alpha$} \Comment{{i.e.,~$\BF(i,S)\neq\BF(j,S)$}}
\State $E[i,j] \gets 0$
\Else \Comment{{null hypothesis $\BF(i, S) = \BF(j, S)$ is accepted}}
\If{$\max\{\texttt{p-val}[\BF(i,S)\le \BF(0,S)],\,\texttt{p-val}[\BF(j,S)\le \BF(0,S)]\}\le \alpha$}
\State $E[i,j] \gets\min(1,E[i,j])$ \Comment{{this is the case where $\BF(i, S) = \BF(j, S)>\BF(0,S)$}} \label{line:noisyOne}
\Else 
\State $E[i,j] \gets \min(\texttt{p-val}[\BF(i, S) = \BF(j, S)],E[i,j])$
\EndIf
\EndIf
\EndFor
\State $\texttt{NoBoost}\gets\{i\in S:\texttt{p-val}[\BF(i, S) \le \BF(0, S)] >\beta\}$ \Comment{{null hypothesis $\BF(i,S)=\BF(0,S)$ is accepted}}
\For{$i \in \texttt{NoBoost}, \; k \notin S$}
\State $E[i,k] \gets \min(1- \texttt{p-val}[\BF(i, S) \le \BF(0, S)], E[i,k])$
\State $E[k,i] \gets \min(1- \texttt{p-val}[\BF(i, S) \le \BF(0, S)], E[k,i])$
\EndFor

\EndFor
\State  $\text{OneHopTransitivity}\gets\{(i,j)\in[n]^2:E[i,j]\ne 0,E[i,k]=E[j,k]=1 \text{ for some } k\in[n]\}$

\Comment{{$E[i,k],E[j,k]=1$ is only possible if they were set in line~\eqref{line:noisyOne}}}
\State \indent $E[i,j]\gets 1$ for all $(i,j)\in \text{OneHopTransitivity}$
\State $E[i,j] \gets 0$  for all $(i,j)\in[n]$ such that $E[i,j] = 2$
\State Community Detection on $E$ to recover nests. 
\end{algorithmic}
}
\end{algorithm*}

\begin{algorithm*}[!h]
\caption{\textsc{Inexact Nest Identification without Outside Option}}
\label{alg:noisyNestWithoutOutside}
\begin{algorithmic}[1]
\State Initialize adjacency matrix $E[i, j] \gets 2$ for all $i, j \in [n]$ \Comment{{2 means $\texttt{null}$}}

\For{$S \in \mathcal{S}$; $i,j\in S$}
\If{$\texttt{p-val}[\BF(i, S) = \BF(j, S)] \le \alpha$}
\State $E[i,j] \gets 0$
\Else
\State $E[i,j] \gets \min \left(\texttt{p-val}[\BF(i, S) = \BF(j, S)], E[i,j]\right)$
\EndIf
\EndFor
\State $E[i,j] \gets 0$ for all $(i,j)\in[n]$ such that $E[i,j] = 2$
\State Community Detection on $E$ to recover nests. 
\end{algorithmic}
\end{algorithm*}

In our nest identification numerics in \Cref{ssec:numericNestIdWith,ssec:numericNestIdWO}, we always set $\alpha=0.05$ and $\beta=1-\alpha$.  In our nest identification numerics in \Cref{sec:deployment}, set $\alpha = 0.0005$ and try varying values of $\beta$. 

\subsection{Nested Logit Ground Truth Generation (supplement to \Cref{ssec:numericNestIdWith})} \label{app:randomNestGeneration}

To generate ground-truth Nested Logit instances for our simulation study, we first sample a random nesting structure over the $n$ items. Specifically, we uniformly permute the items, draw the number of nests $K$ uniformly from $\{1,2,\dots,\lfloor n/2\rfloor\}$, and then place $K-1$ dividers at uniformly random positions in the permuted list to obtain a partition into $K$ nests.

Given the nesting structure, we sample the item preference weights $\{v_i\}_{i=1}^n$ such that the ratio between the largest and smallest weight is bounded, i.e., $\frac{\max_{i \in [n]} v_i}{\min_{i \in [n]} v_i} \le 10.$

Finally, we draw the nest dissimilarity parameter $\lambda_N$ independently from the uniform distribution on $[0.3, 0.6]$.
Note, singleton nests will be reparameterized to $\lambda_N = 1$.

\subsection{Nest Identification Algorithm of \citet{benson2016relevance}} \label{app:BensonImplementation}

We first explain the high-level difference between our nest identification algorithm and that of \citet{benson2016relevance}.
For a pair of items $i,j$, our algorithm ends up comparing (see~\eqref{eqn:statComparison})
\begin{align} \label{eqn:noCrossLearn}
\frac{\hphi(i,[n])}{\hphi(i,[n])+\hphi(j,[n])}\qquad\text{to}\qquad\frac{\hphi(i,S)}{\hphi(i,S)+\hphi(j,S)},
\end{align}
\textit{separately} for each experimental assortment $S\in\cS$ such that $S\supseteq\{i,j\}$.
Meanwhile, their algorithm compares
\begin{align} \label{eqn:withCrossLean}
&\frac{\sum_{S\in\cS:S\supseteq\{i,j,k\}}X(i,S)}{\sum_{S\in\cS:S\supseteq\{i,j,k\}}(X(i,S)+X(j,S))}
\qquad\text{to}
\qquad 
\frac{\sum_{S\in\cS:S\supseteq\{i,j\}}X(i,S)}{\sum_{S\in\cS:S\supseteq\{i,j\}}(X(i,S)+X(j,S))},
\end{align}
where $k\notin\{i,j\}$ is another item and $X(i,S)$ counts the number of times that $S$ is offered and $i$ is chosen (recall that $\hphi(i,S)$ was the empirical probability of $i$ being chosen when $S$ is offered).
We defer full details to \citet[\S4.4]{benson2016relevance}, but the key difference is that unlike~\eqref{eqn:noCrossLearn}, they \textit{aggregate observations} across multiple $S\in\cS$ in~\eqref{eqn:withCrossLean}. 

This can be viewed as a bias-variance tradeoff.  If each experimental assortment $S\in\cS$ has been offered a small number of times, then our algorithm will suffer enormous variance when comparing the ratios in~\eqref{eqn:noCrossLearn}.
However, their comparison is biased depending on $\cS$, even if all assortments in $\cS$ have been offered an identical and asymptotic number of times.  For example, suppose $\cS=\{\{i,j\},\{i,j,k,\ell\}\}$ for distinct items $i,j,k,\ell$.  Then comparison~\eqref{eqn:withCrossLean}, which is supposed to measure whether $k$ is an "irrelevant alternative of $i$ and $j$", is symmetrically also measuring whether $\ell$ is an irrelevant alternative of $i$ and $j$ (see \citet{benson2016relevance} for more details).
This is why the assortments $S$ in their experiment design $\cS$ satisfy $|S|\le 3$---to avoid other items like $\ell$ that may be disproportionately more prevalent in $\{S\in\cS:S\supseteq\{i,j,k\}\}$ than in $\{S\in\cS:S\supseteq\{i,j\}\}$.
However, we demonstrate that adding this constraint to the experiment design $\cS$ actually drastically worsens empirical performance in non-asymptotic settings.

We implement the  "Greedy" version of their nest algorithm discussed in \citet[\S4.6,\S6.1]{benson2016relevance}, which is simpler and better at dealing with data sparsity.
Like in our algorithm, we fix $\alpha=0.05$ for the statistical tests, which are used by their algorithm to test whether the two proportions in~\eqref{eqn:withCrossLean} are equal.
To prevent our algorithm from having an advantage via community detection (see \Cref{app:algs3and4}), we also implement a community detection post-processing step for their algorithm, which helps with nest identification.
For their experiment design, we use the \textit{non-adaptive} version for simplicity (see \citet[\S5]{benson2016relevance}); we also test their nest identification algorithm in combination with our experiment design.

\section{Supplement to \Cref{sec:deployment}}

\subsection{Background} \label{sec:d11background}
We deployed our experiment design at Dream11, the world's largest Daily Fantasy Sports (DFS) platform boasting over 250 million users, of whom 70 million participated in our experiment. Like traditional fantasy sports leagues, DFS also involves constructing a ``fantasy team'' from real sports players with the goal of scoring the most ``fantasy points'', which are earned by the statistical achievements of chosen players in real sports events. However, DFS differs from traditional fantasy in 2 main aspects: (1) DFS competitions span individual sports matches rather than seasons (and consequently remove some of the more complex mechanics like in-season trading), and (2) are generally managed by platforms rather than individuals.

When a user enters the platform, they are greeted by the match selection page (\Cref{fig:d11_1}). Upon selecting a match, a user is then offered
different types of competitions (\Cref{fig:d11_2}), called ``contests'', for them to join and compete against others in. These contests vary along a number of key dimensions: format (depending on the sport), entry fee (if applicable), participant capacity, total prize pool, and prize structure. 

\begin{figure}[!htb]
    \centering
    \subfloat[]{%
        \includegraphics[width=0.3\textwidth]{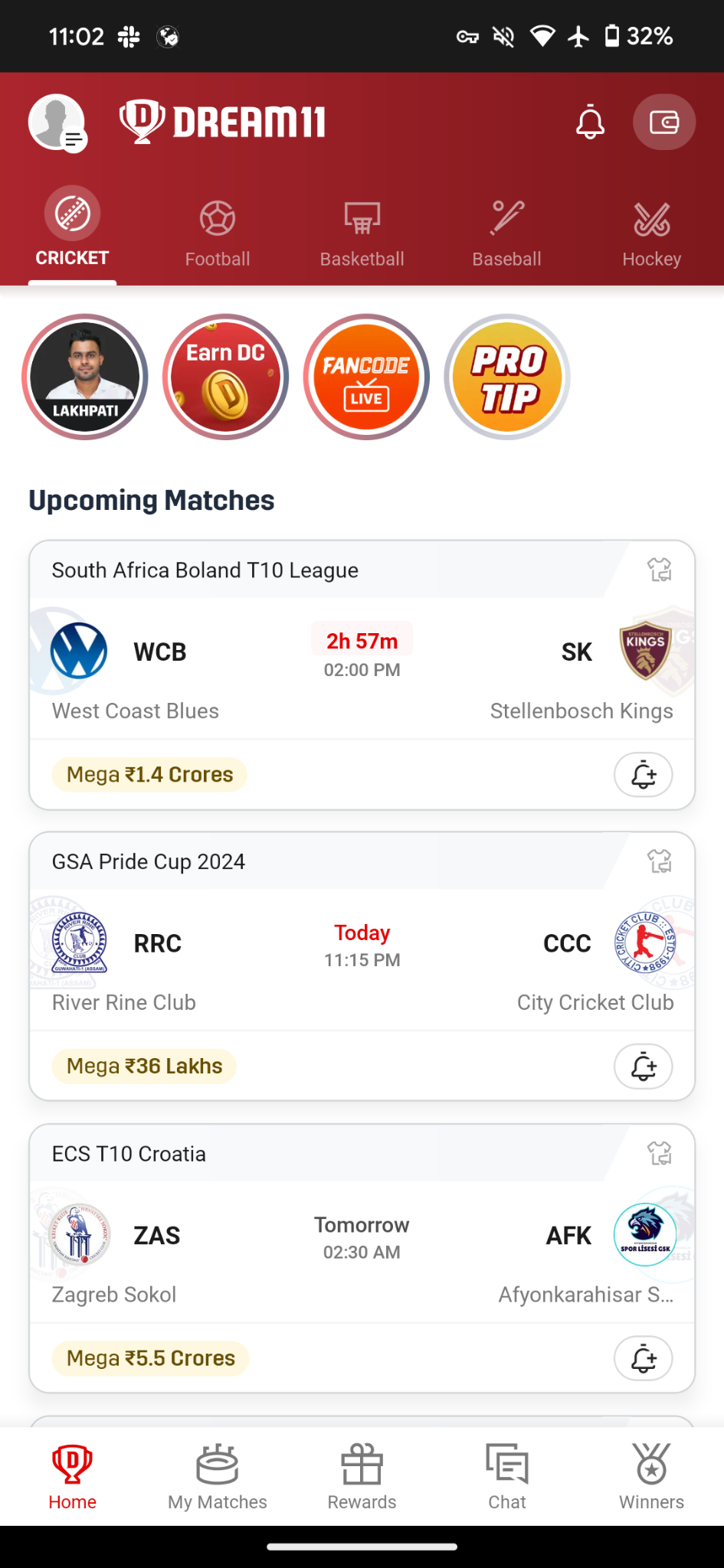}
        \label{fig:d11_1}
    }\hfill
    \subfloat[]{%
        \includegraphics[width=0.3\textwidth]{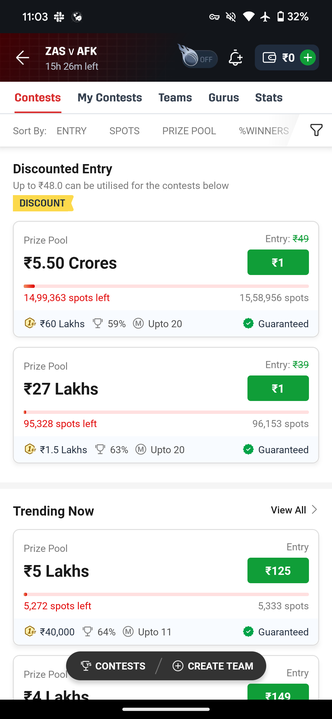}
        \label{fig:d11_2}
    }\hfill
    \subfloat[]{%
        \includegraphics[width=0.3\textwidth]{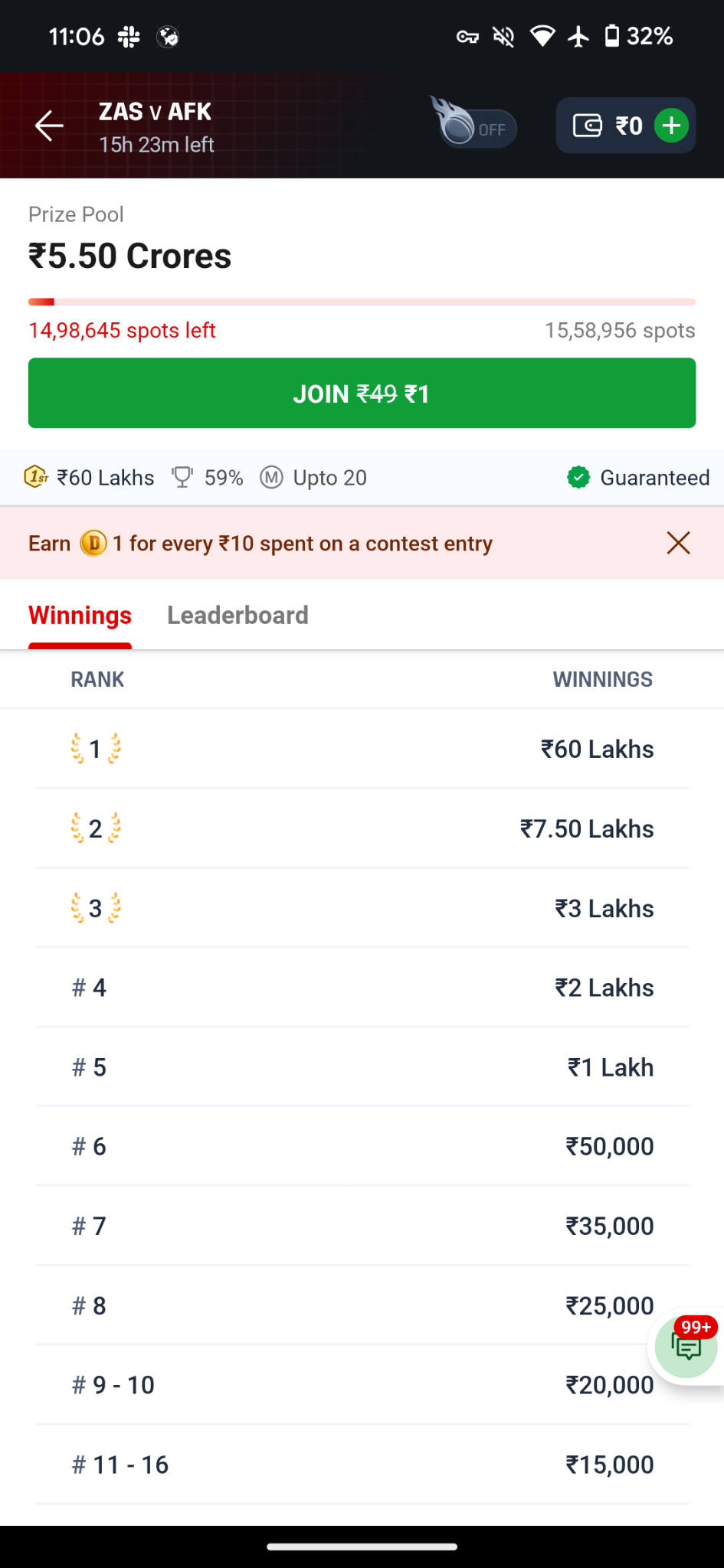}
        \label{fig:d11_3}
    }
    \caption{Screenshots of the Dream11 app, with currency denoted in ``Lakhs'' (one hundred thousand) and ``Crores'' (10 million). (a) shows the upcoming sports matches for which Dream11 hosts contests. (b) shows the currently available contests of a specific match (ZAS vs AFK), where each contest box contains information on the total prize pool, the top prize, the \% of winners, and the maximum number of entries per user. (c) is the contest details page of a specific contest, which further specifies the full prize structure.}
    \label{fig:d11}
\end{figure}

\subsection{Examples of Identified Nests}
\label{app:nest-examples}
The example nests below are drawn from the nest partitions identified on different days of the deployment. Since the partition is re-estimated daily, a given contest may belong to different nests across Tables~\ref{tab:essentially_same} to \ref{tab:irrational}.
\begin{table}[ht]
\centering
\caption{A nest consisting of two contests that are not directly close in the original feature space, but are similar once the Winner Ratio and Prize Ratio are considered.}
\begin{tabular}{lrrrrrr}
\toprule
Contest ID & Entry Fee & Prize Pool & \# of Contestants & \# of Winners & \textit{Winner Ratio} & \textit{Prize Ratio} \\
\midrule
5  & 21 & 35  & 2  & 1 & 0.50 & 0.83 \\
23 & 36 & 300 & 10 & 5 & 0.50 & 0.83 \\
\bottomrule
\end{tabular} 
\label{tab:essentially_same}
\end{table}

\begin{table}[ht]
\centering
\caption{A nest consisting of winner-take-all contests.}
\begin{tabular}{lrrrrrr}
\toprule
Contest ID & Entry Fee & Prize Pool & \# of Contestants & \# of Winners & \textit{Winner Ratio} & \textit{Prize Ratio} \\
\midrule
44 & 179  & 630  & 4 & 1 & 0.25 & 0.88 \\
51 & 1150 & 3000 & 3 & 1 & 0.33 & 0.87 \\
54 & 89   & 310  & 4 & 1 & 0.25 & 0.87 \\
76 & 525  & 1800 & 4 & 1 & 0.25 & 0.86 \\
\bottomrule
\end{tabular}
\label{tab:small_size}
\end{table}

\begin{table}[ht]
\centering
\caption{A nest consisting of low-entry-fee contests with similar Winner Ratios but different tradeoffs between Prize Ratio and number of winners.}
\begin{tabular}{lrrrrrr}
\toprule
Contest ID & Entry Fee & Prize Pool & \# of Contestants & \# of Winners & \textit{Winner Ratio} & \textit{Prize Ratio} \\
\midrule
10 & 21  & 7925 & 500 & 225 & 0.45 & 0.75 \\
20 & 45  & 810  & 21  & 9   & 0.43 & 0.86 \\
22 & 76  & 1368 & 21  & 9   & 0.43 & 0.86 \\
23 & 36  & 300  & 10  & 5   & 0.50 & 0.83 \\
33 & 125 & 5000 & 50  & 25  & 0.50 & 0.80 \\
\bottomrule
\end{tabular}
\label{tab:irrational}
\end{table}
\end{document}